\documentclass[aps,prx,reprint,groupedaddress,floatfix]{revtex4-2}

\usepackage{multirow}
\usepackage{amssymb,amsmath,amsthm,amsfonts}
\usepackage{bm}
\usepackage{textgreek}
\usepackage{mathtools}
\usepackage{enumitem}
\usepackage{graphicx}
\usepackage{tablefootnote} 
\usepackage[ruled,vlined,linesnumbered]{algorithm2e}
\usepackage{algorithmic}

\usepackage{physics}
\usepackage[justification=Justified]{caption}
\usepackage{subcaption}

\usepackage{tikz}
\usetikzlibrary{shapes.geometric, arrows.meta, positioning, fit, calc}

\usepackage{footnote}
\usepackage{xcolor}
\usepackage{mathrsfs}
\usepackage{bbm}
\usepackage{makecell}
\usepackage{pbox}
\usepackage[colorlinks]{hyperref}
\hypersetup{citecolor=blue}

\newtheorem{theorem}{Theorem}

\newtheorem{lemma}{Lemma}

\newtheorem{corollary}{Corollary}

\newcommand{\eq}[1]{(\ref{eq:#1})}
\newcommand{\thm}[1]{\hyperref[thm:#1]{Theorem~\ref*{thm:#1}}}
\newcommand{\cor}[1]{\hyperref[cor:#1]{Corollary~\ref*{cor:#1}}}
\newcommand{\defn}[1]{\hyperref[defn:#1]{Definition~\ref*{defn:#1}}}
\newcommand{\lem}[1]{\hyperref[lem:#1]{Lemma~\ref*{lem:#1}}}
\newcommand{\prop}[1]{\hyperref[prop:#1]{Proposition~\ref*{prop:#1}}}
\newcommand{\assum}[1]{\hyperref[assum:#1]{Assumption~\ref*{assum:#1}}}
\newcommand{\fig}[1]{\hyperref[fig:#1]{Figure~\ref*{fig:#1}}}
\newcommand{\tab}[1]{\hyperref[tab:#1]{Table~\ref*{tab:#1}}}
\newcommand{\algo}[1]{\hyperref[algo:#1]{Algorithm~\ref*{algo:#1}}}
\renewcommand{\sec}[1]{\hyperref[sec:#1]{Section~\ref*{sec:#1}}}
\newcommand{\append}[1]{\hyperref[append:#1]{Appendix~\ref*{append:#1}}}
\newcommand{\fac}[1]{\hyperref[fac:#1]{Fact~\ref*{fac:#1}}}
\newcommand{\lin}[1]{\hyperref[lin:#1]{Line~\ref*{lin:#1}}}

\renewcommand{\d}{\mathrm{d}}
\renewcommand{\L}{\mathcal{L}}

\newcommand{\A}{\mathcal{A}}

\newcommand{\G}{\mathcal{G}}
\newcommand{\E}{\mathcal{E}}
\renewcommand{\comm}{\mathrm{comm}}
\newcommand{\supp}{\mathrm{supp}}
\newcommand{\eps}{\varepsilon}
\newcommand{\mc}{\mathcal}
\newcommand{\mb}{\mathbb}
\def\>{\rangle}
\def\<{\langle}

\renewcommand\bra[1]{{\langle{#1}|}}
\renewcommand\ket[1]{{|{#1}\rangle}}
\renewcommand{\ketbra}[2]{\ket{#1}\!\bra{#2}}
\newcommand{\myparagraph}[1]{\smallskip\textit{#1}\quad}
\DeclareMathOperator{\ad}{ad}

\SetKwComment{Comment}{/* }{ */}
\SetKwInput{KwInput}{Input}                
\SetKwInput{KwOutput}{Output}              
\SetKwInput{KwNotation}{Notation}              
\SetKwInput{KwParameter}{Parameter}             

\begin{document}


\title{Lindbladian Simulation with Commutator Bounds}


\author{Xinzhao Wang$^{1,2}$}
\thanks{Equal Contribution.}

\author{Shuo Zhou$^{1,2,5}$}
\thanks{Equal Contribution.}

\author{Xiaoyang Wang$^{3,4}$} 

\author{\\ Yi-Cong Zheng$^{5}$}

\author{Shengyu Zhang$^{5}$}
\thanks{Corresponding author. Email: shengyzhang@tencent.com}

\author{Tongyang Li$^{1,2}$}
\thanks{Corresponding author. Email: tongyangli@pku.edu.cn}

\affiliation{\vspace{1em}\textsuperscript{1}Center on Frontiers of Computing Studies, Peking University, Beijing 100871, China}

\affiliation{\textsuperscript{2}School of Computer Science, Peking University, Beijing 100871, China}

\affiliation{\textsuperscript{3}RIKEN Center for Interdisciplinary Theoretical and Mathematical Sciences (iTHEMS), Wako, Saitama 351-0198, Japan}

\affiliation{\textsuperscript{4}RIKEN Center for Computational Science (R-CCS), Kobe, Hyogo 650-0047, Japan}

\affiliation{\textsuperscript{5}Tencent Quantum Laboratory, Tencent, Shenzhen, Guangdong 518057, China}


\date{\today}

\begin{abstract}
Trotter decomposition provides a simple approach to simulating open quantum systems by decomposing the Lindbladian into a sum of individual terms. While it is established that Trotter errors in Hamiltonian simulation depend on nested commutators of the summands, such a relationship remains poorly understood for Lindbladian dynamics. In this paper, we derive commutator-based Trotter error bounds for Lindbladian simulation, yielding an $\mathcal{O}(\sqrt{N})$ scaling in the number of Trotter steps for locally interacting systems on $N$ sites. When estimating observable averages, we apply Richardson extrapolation so that the number of Trotter steps per circuit depends polylogarithmically on the inverse precision while maintaining the commutator scaling. To bound the extrapolation remainder, we develop a general truncation bound for the Baker-Campbell-Hausdorff expansion that bypasses common convergence issues in physically relevant systems. For local Lindbladians, our results demonstrate that the Trotter-based methods outperform prior simulation techniques in system-size scaling while requiring only $\mathcal{O}(1)$ ancillas. Numerical simulations further support the predicted system-size and precision scaling.
\end{abstract}


\maketitle

\section{Introduction}
The Lindblad master equation serves as the canonical framework for the study of open quantum systems~\cite{lindblad1976generators, gorini1976completely}. It has been used to model dissipative processes across diverse fields, ranging from quantum optics~\cite{gardiner2004quantum, breuer2002theory} and quantum statistical mechanics~\cite{davies1974markovian, alicki1976detailed} to noise modeling in quantum computation~\cite{van2023probabilistic}. Beyond describing physical phenomena, Lindblad dynamics is applied to design quantum algorithms for the dissipative preparation of a wide variety of quantum states~\cite{lin2025dissipative}, including thermal states~\cite{chen2025efficient}, ground states~\cite{ding2024single}, and excited states~\cite{li2025dissipative}. Furthermore, it extends to solving differential equations~\cite{shang2025designing} and optimization problems~\cite{chen2025quantum,chen2025local}. 

Efficient simulation of Lindblad dynamics is therefore essential, serving both as a tool to investigate open quantum system dynamics and as a primitive for executing such quantum algorithms.
Digital simulation of Lindblad dynamics was initiated using Trotter decomposition~\cite{kliesch2011dissipative, childs2017efficient}. These methods approximate the evolution $e^{t\mc{L}}$ generated by a Lindbladian $\mathcal{L} = \sum_{j=1}^m \mc{L}_j$ using a product formula—a product of the exponentials of the summands. Since then, various algorithmic techniques have been developed~\cite{cattaneo2021collision,schlimgen2022quantum,pocrnic2023quantum,patel2023wave1,patel2023wave2,ding2024simulating,     kato2024exponentially,Chen2025randomizedmethod,yu2025lindbladian,mohammadipour2025reducing}. In particular, methods based on linear combinations of unitaries~\cite{cleve2017efficient,li2023simulating,chen2025efficient,peng2025quantum} achieve near-optimal complexity in terms of evolution time and precision. However, they require substantial ancilla overhead and complex multi-qubit controlled operations. In contrast, Trotter-based methods require fewer ancillas and simpler circuits~\cite{yu2025lindbladian}, making them more suitable for practical implementation. 

Beyond hardware efficiency, Trotter-based methods offer a distinct advantage in Hamiltonian simulation: the error is governed by nested commutators rather than the sum of the norms of the operator summands, enabling tighter complexity bounds~\cite{childs2021theory}. However, analogous commutator bounds are largely missing for Lindbladian simulation. Existing analyses bound the error using the sum of the norms $\sum_{j=1}^{m}\|{\mathcal{L}_j}\|$, which fail to account for the Lie-algebraic structure of the nested commutators of $\{\mathcal{L}_j\}$. Deriving such bounds is difficult, as standard commutator-based Trotter error representation involves inverse time evolutions~\cite{childs2021theory}. For Lindbladians, the norms of these inverse-evolution factors can grow exponentially due to dissipation.

In this paper, we provide the first commutator-based error bounds for Lindbladian product formulas, addressing both the total Trotter error and its arbitrarily high-order remainders.
The bound for the total error implies that standard product formulas with a fixed step size achieve a complexity scaling governed by nested commutators of the operator summands. To achieve higher precision, we use Richardson extrapolation to eliminate the leading-order terms of the Trotter error, such that the final precision is determined by the high-order remainder. The bound for the high-order remainders then gives a simulation algorithm in which the number of Trotter steps per circuit depends polylogarithmically on $1/\eps$ while maintaining the commutator-based scaling. Specifically, for local Lindbladians on $N$ sites, our approach requires only $\mathcal{O}(\sqrt{N})$ Trotter steps, improving the previous $\mathcal{O}(N^{3/2})$ scaling by a factor of $N$. Furthermore, the circuit implementation requires only constant ancillas, making it favorable for early fault-tolerant devices.
In addition, we conduct numerical experiments to examine the scaling with system size and the error reduction from extrapolation.

\section{Problem setup}
We aim to simulate open quantum systems governed by the Lindblad master equation:
\begin{align}
\frac{\mathrm{d} \rho}{\mathrm{d} t} = \mathcal{L}(\rho) \coloneqq -\mathrm{i}[H, \rho] + \sum_{\nu=1}^{m_D} \Big( L_{\nu} \rho L_{\nu}^\dagger - \frac{1}{2} \{ L_{\nu}^\dagger L_{\nu}, \rho \} \Big).\nonumber
\end{align}
While our Trotter error analysis holds for general Lindbladians, we focus on systems defined on a lattice $[N]$. We assume that each Hamiltonian term $H_\mu$ in $H = \sum_{\mu=1}^{m_C}H_\mu$ is supported on at most $k$ sites and that each jump operator $L_\nu = \sum_\gamma d_{\nu,\gamma}$ is a sum of at most $\Gamma$ components, each supported on at most $k$ sites. We refer to such a system as a \textit{$(\Gamma,k)$-local Lindbladian}. Expanding the commutators and anticommutators into elementary left- and right-multiplication maps gives
\begin{align}
&\mc L(\cdot)\nonumber \\
={}&\sum_\mu\big[-\mathrm{i}H_\mu(\cdot)
+\mathrm{i}(\cdot)H_\mu\big]\nonumber\\
&+\sum_{\nu,\gamma_1,\gamma_2}\frac{1}{2}\Big[
2d_{\nu,\gamma_1}(\cdot)d_{\nu,\gamma_2}^\dagger
-d_{\nu,\gamma_2}^\dagger d_{\nu,\gamma_1}(\cdot)
-(\cdot)d_{\nu,\gamma_2}^\dagger d_{\nu,\gamma_1}\Big]\nonumber\\
=: &\sum_{v=1}^M\mc K_v .
\label{eq:local-decompose}
\end{align}
Each $\mc K_v$ denotes one of the elementary superoperators in Eq.~\eq{local-decompose} and is supported on at most $2k$ sites. These superoperators need not themselves be Lindbladians and are used for the locality and commutator estimates. The product formula instead uses the coarser Lindbladian decomposition $\mc L=\sum_j\mc L_j$ in Eq.~\eq{prod-formula}. We assume that the local decomposition is \emph{$g$-extensive}, meaning that the local interaction strength at each site is bounded by
\begin{align}
\textstyle
\label{eq:def-extensive-main}
\sum_{v: \mathrm{supp}(\mathcal{K}_v) \ni j} \|\mathcal{K}_v\|_{\diamond} \le g, \quad \forall j \in [N],
\end{align}
where $\|\cdot\|_{\diamond}$ denotes the diamond norm and $\mathrm{supp}(\cdot)$ represents the support of the superoperator. In the special case where each jump operator is individually supported on at most $k$ sites, the $(1,k)$-local Lindbladian reduces to the standard definition of a \emph{$k$-local Lindbladian}.

We consider two tasks: (i) approximating the evolution channel $e^{t\mathcal{L}}$ within a diamond-norm error $\eps$, and (ii) estimating the expectation value of a given observable $O$ with an additive error $\eps \|O\|$ for the time-evolved state.

\begin{figure}[t]
\centering
\resizebox{\columnwidth}{!}{%
\begin{tikzpicture}[
    x=4.6cm, y=-1.45cm, 
    base/.style = {rectangle, draw=black, thick, align=center, inner sep=2mm, text width=4.0cm},
    task/.style = {base, fill=gray!15, font=\small\bfseries},
    theory/.style = {base, rounded corners, fill=white, font=\footnotesize},
    result/.style = {base, fill=white, font=\footnotesize, densely dotted},
    impl/.style = {base, fill=white, font=\footnotesize, text width=3.8cm},
    arrow/.style = {-{Stealth[scale=1.0]}, thick}
]

\node (top) [base, fill=gray!15, font=\small\bfseries, text width=8.8cm] at (0.5, 0) {Lindbladian simulation via second-order Trotterization};

\node (task1) [task] at (0, 1) {Channel approximation};
\node (task2) [task] at (1, 1) {Observable estimation};

\draw [arrow] (top.south) -- (0.5, 0.53) -| (task1.north);
\draw [arrow] (top.south) -- (0.5, 0.53) -| (task2.north);

\node (theory1) [theory] at (0, 2.7) {\textbf{Direct product formula}\\\thm{commutator}\\ Commutator bounds for Trotter error};
\node (res1) [result] at (0, 4.4) {Trotter depth\\ ${\mathcal{O}}(\sqrt{N}\eps^{-1/2})$};

\draw [arrow] (task1.south) -- (theory1.north);
\draw [arrow] (theory1.south) -- (res1.north);

\node (theory2a) [theory] at (1, 1.93) {\textbf{BCH truncation bound}\\\thm{bch-pf-main}\\ Bypasses BCH divergence};
\node (theory2b) [theory] at (1, 3.23) {\textbf{Richardson extrapolation}\\\thm{general-extra-algo-main}\\ Commutator bounds and exponentially improved precision};
\node (res2) [result] at (1, 4.4) {Trotter depth\\ $\mathcal{O}(\sqrt{N}\mathrm{polylog}(N/\eps))$};

\draw [arrow] (task2.south) -- (theory2a.north);
\draw [arrow] (theory2a.south) -- (theory2b.north);
\draw [arrow] (theory2b.south) -- (res2.north);

\node (impl_title) at (0.5, 5.23) [font=\small\bfseries] {Circuit implementation of Trotter steps};

\node (impl1) [impl] at (0.02, 5.94) {\textbf{Stinespring dilation}\\ $\mathcal{O}(1)$-local Lindbladians\\ $\mathcal{O}(1)$ ancillas};
\node (impl2) [impl] at (0.98, 5.94) {\textbf{Block encoding}\\ $(\Gamma,k)$-local Lindbladians\\ $\mathrm{polylog}$ ancillas};

\draw [dashed, thick] (impl_title.west) -- (-0.48, 5.2) |- (1.48, 6.63) |- (impl_title.east);

\draw [thick] (res1.south) |- (0.5, 4.84);
\draw [thick] (res2.south) |- (0.5, 4.84);
\draw [arrow] (0.5, 4.84) -- (impl_title.north);

\end{tikzpicture}%
}
\caption{Flowchart of our simulation framework. }
\label{fig:flowchart}
\end{figure}

\section{Trotter error analysis}
We first derive Trotter error bounds for general Lindbladians, which determine the number of Trotter steps required for a target precision. We consider the second-order product formula
\begin{align}\label{eq:prod-formula}
\mathcal{S}(t) \coloneqq \prod_{j=m}^1 e^{t \mathcal{L}_j /2} \prod_{j=1}^m e^{t\mathcal{L}_j /2},
\end{align}
based on a decomposition $\mathcal{L} = \sum_{j=1}^m \mathcal{L}_j$ with $m := m_C + m_D$ summands. Each Lindbladian summand $\mathcal{L}_j$ represents either a coherent term $\mathcal{H}_\mu(\cdot) = -\mathrm{i}[H_\mu, \cdot]$ or a dissipator $\mathcal{D}_\nu(\cdot) = L_\nu (\cdot) L_\nu^\dagger - \frac{1}{2}\{L_\nu^\dagger L_\nu, \cdot\}$. We focus on the second-order product formula as it offers better error scaling compared to the first-order case while avoiding the negative-time evolution required by higher-order Suzuki–Trotter decompositions~\cite{blanes2005necessity} that are not experimentally feasible. While such issues can be addressed by applying higher-order formulas to the full system-bath Hamiltonian~\cite{yu2025optimizing} under specific environment and coupling assumptions, direct Trotterization of the Lindbladian remains limited to second-order.

Trotter errors in Hamiltonian simulation are often expressed as right-nested commutators of the operator summands~\cite{childs2021theory}. A right-nested commutator of $\{\mc L_j\}$ is defined as $[\mathcal{L}_{j_1}, \dots, \mathcal{L}_{j_q}] \coloneqq [\mathcal{L}_{j_1}, [\mathcal{L}_{j_2}, \dots, [\mathcal{L}_{j_{q-1}}, \mathcal{L}_{j_q}] \dots]]$. The number of operators $q$ is called the \emph{grade} of the commutator. 

We give new bounds on Trotter error and its higher-order remainders in Lindbladian simulation using the sum of the norms of \emph{doubly right-nested commutators} denoted by $\alpha_{\mathrm{comm}}^{(q_1, \dots, q_{d})}(\mathcal{L}_1,\ldots,\mathcal{L}_m)$, or $\alpha_{\mathrm{comm}}^{(q_1, \dots, q_{d})}$ for short:
\begin{align*}
\sum_{j_{1,1}, \dots, j_{d, q_{d}} } \big\| \big[ [\mathcal{L}_{j_{1,1}}, \dots, \mathcal{L}_{j_{1,q_1}}], \dots, [\mathcal{L}_{j_{d,1}}, \dots, \mathcal{L}_{j_{d,q_d}}] \big] \big\|_{\diamond}.
\end{align*}
For $d=1$, this reduces to the usual sum of grade-$q_1$ right-nested commutator norms. The general form enters the BCH truncation bound in \thm{bch-pf-main}.

\subsection{Channel approximation}
For approximating the evolution channel $e^{t\mathcal{L}}$, we give a commutator bound analogous to Hamiltonian simulation. We decompose the total time $t$ into $r$ segments of duration $\tau = t/r$ and approximate the evolution as $\mathcal{S}(\tau)^r$. Due to the contractivity of Lindbladian evolution, i.e., $\|e^{\tau\mathcal{L}}\|_\diamond = 1$ for $\tau \ge 0$, the total simulation error is bounded by the sum of one-step Trotter errors $r\|\mathcal{S}(\tau) - e^{\tau \mathcal{L}}\|_{\diamond}$. Trotter error bounds for general non-anti-Hermitian generators can grow exponentially with time~\cite{childs2021theory}. This growth typically arises because integral representations of the error contain inverse evolutions, which need not be contractive. For the second-order product formula in Eq.~\eq{prod-formula}, however, each inverse-evolution factor in the error representation combines with an adjacent forward evolution generated by the same Lindbladian, resulting in an evolution with a nonnegative time. Consequently, all propagators in the resulting integrand are quantum channels and have diamond norm one (see \append{comm-bound-product} for detailed derivation). This ensures the one-step Trotter error grows only polynomially with time $\tau$: \begin{align}
\label{eq:one-step-trotter-error}
    \|\mathcal{S}(\tau) - e^{\tau \mathcal{L}}\|_{\diamond} \le \alpha_{\mathrm{comm}}^{(3)} \tau^3,
\end{align} where $\alpha_{\mathrm{comm}}^{(3)}$ is the sum of all grade-$3$ right-nested commutators~\footnote{We provide a tighter bound in \append{comm-bound-product}.}.  To ensure the simulation error is bounded by $\eps$, we set $r$ to satisfy $r \alpha_{\mathrm{comm}}^{(3)} (t/r)^3 \le \eps$. This yields the following requirement for the Trotter number.
\begin{theorem}[Channel approximation]
\label{thm:commutator}
For any evolution time $t > 0$ and precision $\eps > 0$, the Lindbladian evolution $e^{t\mathcal{L}}$ can be approximated to precision $\eps$ in the diamond norm using
\begin{align}
r = \mathcal{O}\big( (\alpha_{\mathrm{comm}}^{(3)})^{1/2}t^{3/2}/\sqrt{\eps} \big) \label{eq:trotter_step}
\end{align}
 steps of the second-order product formula \eq{prod-formula}.
\end{theorem}
For a $(\Gamma,k)$-local and $g$-extensive Lindbladian $\mathcal{L}$ on $[N]$, standard Trotter error bounds~\cite{childs2017efficient} depend on the total interaction strength, which is bounded by summing Eq.~\eq{def-extensive-main} over all sites:
\begin{align}
\label{eq:1-norm-bound}
\textstyle
\sum_v \|\mathcal{K}_v\|_{\diamond}\le\sum_{j=1}^N \sum_{v:\supp(\mc K_v)\ni j}\|\mc K_v\|_{\diamond} \le Ng.
\end{align}This leads to a Trotter number $r = \mathcal{O}(N^{3/2} (gt)^{3/2} \eps^{-1/2})$. In contrast,  our commutator-based estimate \begin{align}
\label{eq:comm-3}
    \alpha_{\mathrm{comm}}^{(3)} = \mathcal{O}(k^2 g^3 N)
\end{align} yields $r = \mathcal{O}(\sqrt{N} k (gt)^{3/2} \eps^{-1/2})$. For $k=\mathcal O(1)$, this reduces the number of Trotter steps by a factor of $N$ compared with the standard bound.

\subsection{Observable estimation}
For the task of estimating the expectation value $\tr[Oe^{t \mathcal{L}} \rho_0]$, \thm{commutator} guarantees that a direct Trotter implementation achieves bias at most $\eps\|O\|$ using $\mathcal{O}(1/\sqrt{\eps})$ Trotter depth. To improve this dependence to $\mathrm{polylog}(1/\eps)$, several Trotter error mitigation techniques based on extrapolation~\cite{endo2019mitigating, watson2025exponentially} and interpolation~\cite{rendon2024improved} have been developed for Hamiltonian simulation. Recent work also applies extrapolation to first-order Kraus and Hamiltonian-dilation approximations of Lindbladian evolution~\cite{mohammadipour2025reducing} and to linear differential equation solving~\cite{fang2026circuitdepthreductiononeancilla}. In this paper, we use extrapolation to suppress the Trotter error in Lindbladian simulation, achieving a $\operatorname{polylog}(1/\eps)$ Trotter depth while preserving the commutator complexity scaling.

We define the step-size-dependent expectation value $f(s) \coloneqq \tr[O \mathcal{S}(st)^{1/s} \rho_0]$ and apply the well-conditioned $p$-th order Richardson extrapolation of Refs.~\cite{low2019well,watson2025exponentially} to estimate the zero-step-size limit $\lim_{s\to 0} f(s) = \tr[O e^{t\mathcal{L}}\rho_0]$:
\begin{align*}
F^{(p)}=\sum_{j=1}^p b_jf(s_j),
\qquad \|\bm b\|_1=\mathcal O(\log p).
\end{align*}
For a fixed extrapolation order $p$, we write
\begin{align*}
f(s)=f(0)+\sum_{\ell=1}^{p-1}c_{2\ell}s^{2\ell}+R_{2p}(s),
\end{align*}
where $R_{2p}(s)$ denotes the remainder. The symmetry of the product formula implies that this expansion contains only even powers of $s$. The coefficients $\{b_j\}$ therefore cancel the first $p-1$ nonconstant terms, of degrees $2,4,\ldots,2p-2$, and hence
\begin{align*}
|F^{(p)}-f(0)|
\le \|\bm b\|_1\max_{j\in[p]}|R_{2p}(s_j)|.
\end{align*}
We denote the largest step size by $s_p$. The step sizes and coefficients are given explicitly in \append{richardson}.

\subsubsection{BCH truncation error}
To bound $R_{2p}(s)$, a standard approach starts from the formal BCH expansion $\mathcal{S}(st) = \exp\big(\sum_{q=1}^{\infty} \Phi_q (st)^q\big)$ and then expands $\mathcal{S}(st)^{1/s}$ in $f(s)$~\cite{aftab2024multi, watson2025exponentially, mizuta2026commutator}. Here $\Phi_q$ is a weighted sum of grade-$q$ right-nested commutators of the Lindbladian summands and is bounded by $\|\Phi_q\|_{\diamond} \le \alpha_{\rm comm}^{(q)}/q^2$~\footnote{See Ref.~\cite[Proposition~5]{aftab2024multi}. Although the original proposition is stated for the operator norm, the derivation relies only on the triangle inequality and extends directly to the diamond norm.}.

If $\sup_{q=3,5,\ldots}(\alpha_{\comm}^{(q)})^{1/q}<\infty$, this formal BCH series converges absolutely for sufficiently small $s$ and represents the product formula. The expansion of $f(s)$ derived in \append{extrapolation} then gives
\begin{align}
\label{eq:full-bch-remainder-main}
|R_{2p}(s)|
=\mathcal O\big(
\|O\|(1+2t\sup_{q=3,5,\ldots}
(\alpha_{\comm}^{(q)})^{1/q})^{3p}s^{2p}
\big).
\end{align}
For local lattice systems, however, the available locality estimate for $\alpha_{\comm}^{(q)}$ contains a factorial factor $q!$ (see Ref.~\cite[Lemma~7]{mizuta2026commutator}), whose $q$-th root grows with $q$ and therefore does not control the supremum in Eq.~\eq{full-bch-remainder-main}. Thus, the full-series argument does not give a useful locality-based bound on $R_{2p}(s)$.

To avoid assuming convergence of the full BCH series, Mizuta~\cite{mizuta2026commutator} showed that the exponential of the truncated BCH expansion $\exp\big(\sum_{q=1}^{q_0}\Phi_q(st)^q\big)$ approximates the product formula for Hamiltonian simulation, assuming that each Hamiltonian summand is a sum of mutually commuting local Hermitian operators. However, this condition does not hold for most practical Lindbladians.

Here, we give a general bound for the BCH truncation error in terms of the doubly right-nested commutator bound $\alpha_{\rm comm}^{(q_1, \ldots, q_d)}$.  
\begin{theorem}[BCH truncation error bound]
\label{thm:bch-pf-main}
    Let $\alpha_{\comm}^{(q_1,\ldots,q_d)}(x)$ denote the doubly right-nested commutator sum with each Lindbladian summand scaled by $x$. Define
    \begin{align}
        \label{eq:alpha-comm-q0-main}
        \alpha_{\comm,q_0}(x)
        ={}& \sum_{d=1}^{\infty}\frac{1}{d!}
        \sum_{\substack{1\le q_1, \ldots, q_{d}\le q_0\\
        q_1+\cdots+q_{d} \ge q_0+1}}
        \alpha_{\comm}^{(q_1, \ldots, q_{d})}(x).
    \end{align}
    If $\alpha_{\comm,q_0}(st)\le 1$, the truncation error of the $q_0$-th order BCH expansion for the product formula $\mathcal{S}(st)$ satisfies
          \begin{align*}
          \textstyle
          \big\|\exp\big(\sum_{q=1}^{q_0} \Phi_q (st)^q\big) -   \mc{S}(st)\big\|_{\diamond}\le e\alpha_{\comm,q_0}(st).
          \end{align*}
\end{theorem}

To prove \thm{bch-pf-main}, we map the product formula $\mathcal{S}(st)$ to the continuous evolution
\begin{align*}
&\frac{\d\mc Y(\tau)}{\d\tau}
=\mc L(\tau)\mc Y(\tau),\qquad \mc Y(0)=\mc I,\\
&\mc L(\tau)
=
\begin{cases}
\dfrac{st}{2}\mc L_v,
&\tau\in(v-1,v],\quad 1\le v\le m,\\[1mm]
\dfrac{st}{2}\mc L_{2m+1-v},
&\tau\in(v-1,v],\quad m+1\le v\le 2m.
\end{cases}
\end{align*}
By construction, $\mc Y(2m)=\mc S(st)$. The Magnus expansion for this evolution gives a sequence of terms $\Omega_q(\tau)$, where each $\Omega_q(\tau)$ is a weighted sum of time integrals of grade-$q$ right-nested commutators of $\mc L(\tau)$. Let $\mc Y_{(q_0)}(\tau):=\exp(\Omega_{(q_0)}(\tau))$ denote the exponential of the truncated Magnus expansion $\Omega_{(q_0)}(\tau):=\sum_{q=1}^{q_0}\Omega_q(\tau)$. For every finite $q$, $\Omega_q(2m)=\Phi_q (st)^q$. This finite-order identity follows from the BCH and Magnus formulas and does not require convergence of either infinite series. Hence, $\mc Y_{(q_0)}(2m)$ is the truncated BCH exponential in \thm{bch-pf-main}. The BCH truncation error is therefore $\|\mc Y_{(q_0)}(2m)-\mc Y(2m)\|_\diamond$.

By the derivative formula for the exponential map, $\mc Y_{(q_0)}(\tau)$ is generated by
\begin{align}
\label{eq:truncated-generator}
 \mc{L}_{(q_0)}(\tau) = \sum_{\ell = 0}^{\infty}\frac{1}{(\ell+1)!}\ad_{\Omega_{(q_0)}(\tau)}^\ell(\dot{\Omega}_{(q_0)}(\tau)).
\end{align}
The variation-of-constants formula and Gronwall's inequality bound $\|\mc Y_{(q_0)}(2m)-\mc Y(2m)\|_\diamond$ in terms of the integrated generator difference $\int_0^{2m}\|\mc L_{(q_0)}(\tau)-\mc L(\tau)\|_{\diamond}\,\d\tau$. The order conditions of Ref.~\cite{fang2025high} imply that $\mc L_{(q_0)}(\tau)$ matches $\mc L(\tau)$ through grade $q_0$, so the lower-grade terms in Eq.~\eq{truncated-generator} cancel. The remaining terms are right-nested commutators of $\Omega_q$ and $\dot\Omega_q$. Since each $\Omega_q$ is itself a weighted sum of right-nested commutators of the original Lindbladian summands, these terms have a doubly right-nested commutator structure. Bounding the terms of grade greater than $q_0$ gives $\alpha_{\comm,q_0}(st)$. The complete derivation is given in \append{bch-truncate}.

For a fixed truncation order $q_0$, we instead expand the exponential of the truncated BCH series $\exp\big(\sum_{q=1}^{q_0}\Phi_q(st)^q/s\big)$. For sufficiently small $s$ such that $\alpha_{\comm,q_0}(st)=\mathcal O(s)$, the two contributions to $R_{2p}(s)$ satisfy
\begin{align}
\label{eq:remainder-bound-main}
|R_{2p}(s)|
={}&\mathcal O\big(
\|O\|(1+2t\max_{q=3,5,\ldots,q_0}
(\alpha_{\comm}^{(q)})^{1/q})^{3p}s^{2p}
\big)\nonumber\\
&+\mathcal O\big(
\|O\|\alpha_{\comm,q_0}(st)/s
\big).
\end{align}
The first term is the tail of the power-series expansion of the truncated BCH exponential. The second is the BCH truncation error in \thm{bch-pf-main} accumulated over $1/s$ steps. Compared with Eq.~\eq{full-bch-remainder-main}, truncation replaces the supremum over all BCH grades by a maximum over the finite range $q\le q_0$, so factorial growth at higher grades no longer enters the first term. The price is the second term, which we bound for local Lindbladians next.

Although $\alpha_{\comm,q_0}(st)$ involves an infinite sum of nested-commutator norms, \cor{local-bch-control-main} will show that $\alpha_{\comm,q_0}(st)$ remains bounded for $(\Gamma,k)$-local Lindbladians at sufficiently small step sizes. We first establish the following estimate for doubly right-nested commutators of local superoperators.
\begin{theorem}
\label{thm:doubly-nested-bound-main}
Let $\mc L=\sum_{v=1}^{M}\mc K_v$ be a decomposition into superoperators on an $N$-site lattice, each supported on at most $2k$ sites, satisfying the $g$-extensiveness condition. For positive integers $q_1,\ldots,q_d$, define $P_r=\sum_{j=r}^d q_j$ for $r\in[d]$ and $P_{d+1}=1$. Then
\begin{align}
    \label{eq:double-comm-k-local-main}
        &\sum_{v_{1,1}, \ldots, v_{d, q_{d}}=1}^{M}
        \big\|\big[ [\mc K_{v_{1,1}}, \ldots, \mc K_{v_{1, q_1}}], \ldots,[\mc K_{v_{d,1}}, \ldots, \mc K_{v_{d, q_{d}}}] \big]
        \big\|_{\diamond}\nonumber\\
        &\le\frac{1}{4kq_d}
        \bigg(\prod_{r=1}^{d} P_{r+1} q_r! (4kg)^{q_r}\bigg) N.
\end{align}
\end{theorem}

The proof proceeds by backward induction over the inner commutators. A right-nested commutator $[\mc K_{v_1},\ldots,\mc K_{v_q}]$ can be nonzero only if each $\mc K_{v_i}$ overlaps the union of the supports of the terms to its right. To apply this observation to Eq.~\eq{double-comm-k-local-main}, write each inner commutator as
\begin{align*}
\mc C_{\vec v_r}:=[\mc K_{v_{r,1}},\ldots,\mc K_{v_{r,q_r}}],
\end{align*}
and write the part to its right as
\begin{align*}
\mc W_{r+1}:=[\mc C_{\vec v_{r+1}},\ldots,\mc C_{\vec v_d}].
\end{align*}
Let $X_{r+1}=\supp(\mc W_{r+1})$. The commutator $[\mc C_{\vec v_r},\mc W_{r+1}]$ can be nonzero only if at least one local term $\mc K_v$ in $\mc C_{\vec v_r}$ overlaps $X_{r+1}$. For each choice of such an overlapping term $\mc K_v$, the local nested-commutator argument of Ref.~\cite[Lemma~7]{mizuta2026commutator} bounds the sum over the other $q_r-1$ local terms in $\mc C_{\vec v_r}$ by $(q_r-1)!(4kg)^{q_r-1}\|\mc K_v\|_\diamond$. Accounting for the $q_r$ possible positions of $\mc K_v$ then gives
\begin{align*}
&\sum_{\vec v_r}\|[\mc C_{\vec v_r},\mc W_{r+1}]\|_\diamond
\\
&\le 2q_r(q_r-1)!(4kg)^{q_r-1}\|\mc W_{r+1}\|_\diamond
\sum_{v:\,\supp(\mc K_v)\cap X_{r+1}\neq\emptyset}
\|\mc K_v\|_\diamond.
\end{align*}
The $g$-extensiveness condition controls the remaining sum through
\begin{align*}
\sum_{v:\,\supp(\mc K_v)\cap X_{r+1}\neq\emptyset}
\|\mc K_v\|_\diamond
&\le \sum_{j\in X_{r+1}}
\sum_{v:\,j\in\supp(\mc K_v)}\|\mc K_v\|_\diamond\\
&\le |X_{r+1}|g.
\end{align*}
Since $\mc W_{r+1}$ contains $P_{r+1}$ local terms, $|X_{r+1}|\le 2kP_{r+1}$. Combining this support bound with the $g$-extensiveness estimate yields
\begin{align*}
\sum_{\vec v_r}\|[\mc C_{\vec v_r},\mc W_{r+1}]\|_\diamond
\le P_{r+1}q_r!(4kg)^{q_r}\|\mc W_{r+1}\|_\diamond.
\end{align*}
At the base of the induction, the bound $\sum_v\|\mc K_v\|_\diamond\le Ng$, obtained by summing Eq.~\eq{def-extensive-main} over all sites, accounts for the factor $N$. Iterating from right to left proves Eq.~\eq{double-comm-k-local-main}.

Expanding each coarse summand $\mc L_j$ into the elementary superoperators $\mc K_v$ in Eq.~\eq{local-decompose} and using the triangle inequality bounds $\alpha_{\comm}^{(q_1,\ldots,q_d)}$ by the corresponding commutator sum over the $\mc K_v$. \thm{doubly-nested-bound-main} then gives, for $q_i\le q_0$ and $q=\sum_i q_i$,
\begin{align*}
\alpha_{\comm}^{(q_1,\ldots,q_d)}(st)
\le q^d(4q_0kgst)^q N.
\end{align*}
Here $q_i\le q_0$ bounds the factorial growth within each inner commutator by $q_i!\le q_0^{q_i}$, giving the factor $q_0^q$ over all inner commutators, while $P_{r+1}\le q$ gives the factor $q^d$.
In Eq.~\eq{alpha-comm-q0-main}, the growth $q^d$ with the outer-commutator grade $d$ is controlled by the weight $1/d!$, since $\sum_{d=1}^{\infty}q^d/d!\le e^q$. Summing over the total grade and applying \thm{bch-pf-main} gives the following corollary. The detailed summation is given in \append{k-local-example}.
\begin{corollary}
\label{cor:local-bch-control-main}
Let $\mc L$ be a $(\Gamma,k)$-local Lindbladian on an $N$-site lattice satisfying the $g$-extensiveness condition. If $8e^2q_0kgst\le 1$, then
\begin{align*}
\alpha_{\comm,q_0}(st)\le Ne^{-q_0}.
\end{align*}
Consequently, if $Ne^{-q_0}\le 1$, the BCH truncation error is at most $eNe^{-q_0}$.
\end{corollary}

\subsubsection{Extrapolation error and complexity}
Combining Eq.~\eq{remainder-bound-main} with $\|\bm b\|_1=\mathcal O(\log p)$, choosing, for a sufficiently small constant $c>0$,
\begin{align*}
s_p\le c\big(
(1+2t\max_{q=3,5,\ldots,q_0}
(\alpha_{\comm}^{(q)})^{1/q})^{-3/2}
\big)
\end{align*}
makes the first term decay exponentially with $p$, so a sufficiently large $p=\Theta(\log(1/\eps))$ suffices. Requiring the accumulated BCH truncation error to be of order $\eps\|O\|/\log p$ gives the second condition in the following theorem.

\begin{theorem}[Observable estimation via extrapolation]\label{thm:general-extra-algo-main}
There is a universal constant $c>0$ such that the following holds. Let $t>0$ be the simulation time and $\eps\in(0,1)$ the precision. Given a BCH truncation order $q_0\ge3$, choose a sufficiently large extrapolation order $p=\Theta(\log(1/\eps))$ and a maximum step size $s_p$ satisfying
\begin{align*}
    s_p
    \le{}& c\big(
    (1+2t\max_{q=3,5,\ldots,q_0}
    (\alpha_{\mathrm{comm}}^{(q)})^{1/q})^{-3/2}
    \big),\\
    \alpha_{\mathrm{comm}, q_0}(s_p t)
    \le{}&\frac{cs_p\eps}{\log p}.
\end{align*}
Then the extrapolation algorithm estimates $\tr[O e^{t\mathcal{L}}\rho_0]$ to precision $\eps\|O\|$ with probability at least $2/3$, using $\widetilde{\mathcal{O}}(1/\eps^2)$~\footnote{The notation {$\widetilde{\mathcal{O}}(\cdot)$} hides polylogarithmic factors.} circuit runs, each consisting of at most $\widetilde{\mathcal{O}}(1 / s_p)$ Trotter steps.
\end{theorem}

For a $(\Gamma,k)$-local and $g$-extensive Lindbladian $\mc L$ on an $N$-site lattice, \thm{doubly-nested-bound-main} and \cor{local-bch-control-main} determine admissible parameters for \thm{general-extra-algo-main}. Specifically, the choices
\begin{align*}
q_0=\mathcal{O}(\log(Nkgt/\eps)),\quad
\frac{1}{s_p}=\widetilde{\mathcal{O}}\big(\sqrt N(kgt)^{3/2}\big)
\end{align*}
satisfy the two conditions in \thm{general-extra-algo-main} and give the following corollary.

\begin{corollary}
\label{cor:local-observable-estimation-main}
Let $\mc L$ be a $(\Gamma,k)$-local and $g$-extensive Lindbladian on an $N$-site lattice. For any $t>0$ and $\eps\in(0,1)$, there exists an extrapolation algorithm that estimates $\tr[Oe^{t\mc L}\rho_0]$ to precision $\eps\|O\|$ with probability at least $2/3$, using $\widetilde{\mathcal{O}}(1/\eps^2)$ circuit runs, each consisting of
\begin{align*}
\widetilde{\mathcal{O}}\big(\sqrt{N}(kgt)^{3/2}\big)
\end{align*}
Trotter steps.
\end{corollary}

The Trotter depth therefore depends polylogarithmically on the precision and scales as $\widetilde{\mathcal{O}}(\sqrt{N})$ with the system size. The detailed parameter choices and complexity analysis are given in \append{k-local-example}.
\section{Circuit implementation}
We consider two ways to implement the local channels in each Trotter step.

\emph{Stinespring dilation.}
For a $g$-extensive $k$-local Lindbladian with $k=\mathcal{O}(1)$, each $\mc L_j$ is supported on at most $k$ qubits. By Stinespring dilation~\cite{stinespring1955positive}, $e^{\tau\mc L_j}$ can be realized by a unitary $U_j$ acting on the $k$ system qubits and at most $2k$ ancillas, followed by tracing out the ancillas. The ancillas can be reset and reused for the next local channel. Approximating $U_j$ to accuracy $\delta$ costs $\operatorname{polylog}(1/\delta)$ elementary gates by the Solovay--Kitaev theorem~\cite{dawson2006solovay}. Choosing $\delta$ so that the synthesis errors accumulated over the full circuit are at most $\eps$ changes the complexity only by polylogarithmic factors. Since $\mc S(st)$ contains $2m$ local channels, one Trotter step uses $\widetilde{\mathcal O}(m)$ gates and $\mathcal O(1)$ reusable ancillas. Combining this cost with \thm{commutator} and \cor{local-observable-estimation-main}, respectively, gives the per-run gate complexities
\begin{align*}
\widetilde{\mathcal O}\big(m\sqrt N\,(gt)^{3/2}\eps^{-1/2}\big)
\quad\text{and}\quad
\widetilde{\mathcal O}\big(m\sqrt N\,(gt)^{3/2}\big)
\end{align*}
for channel approximation and observable estimation. The latter uses $\widetilde{\mathcal O}(1/\eps^2)$ independent circuit runs, giving a total gate count of $\widetilde{\mathcal O}(m\sqrt N\,(gt)^{3/2}/\eps^2)$.

\emph{Block encoding.}
For a $(\Gamma,k)$-local and $g$-extensive Lindbladian, suppose that we have access to $(\alpha_\mu,a)$-block encodings of the Hamiltonian terms $H_\mu$ and $(\alpha_{\nu,\gamma},a)$-block encodings of the local jump components $d_{\nu,\gamma}$, with $\alpha_\mu=\Theta(\|H_\mu\|)$ and $\alpha_{\nu,\gamma}=\Theta(\|d_{\nu,\gamma}\|)$. The LCU construction gives a block encoding of $L_\nu=\sum_\gamma d_{\nu,\gamma}$ with normalization $\sum_\gamma\alpha_{\nu,\gamma}$ using $\mathcal O(\Gamma)$ queries and gates~\cite{childs2012hamiltonian,gilyen2019quantum}. The dissipative evolution generated by $L_\nu$ for time $\tau$ can then be implemented using $\widetilde{\mathcal O}\big(\tau\big(\sum_\gamma\|d_{\nu,\gamma}\|\big)^2+1\big)$ queries to this block encoding~\cite{li2023simulating}. Similarly, the evolution generated by $H_\mu$ uses $\widetilde{\mathcal O}(\tau\|H_\mu\|+1)$ queries~\cite{low2019hamiltonian}. Since $\|A(\cdot)B\|_\diamond=\|A\|\|B\|$, the decomposition in Eq.~\eq{local-decompose} and Eq.~\eq{1-norm-bound} imply
\begin{align*}
\sum_\mu\|H_\mu\|
+\sum_\nu\Big(\sum_\gamma\|d_{\nu,\gamma}\|\Big)^2
\le \sum_v\|\mc K_v\|_\diamond
\le Ng.
\end{align*}
Accounting for the cost of constructing the block encoding of each $L_\nu$, one Trotter step of duration $st$ therefore costs
\begin{align*}
&\widetilde{\mathcal O}\Big(
st\sum_\mu\|H_\mu\|+m\Gamma
+st\Gamma\sum_\nu\Big(\sum_\gamma\|d_{\nu,\gamma}\|\Big)^2
\Big)\\
\le{}&\widetilde{\mathcal O}\big(
st\Gamma\sum_v\|\mc K_v\|_\diamond+m\Gamma
\big)\\
\le{}&\widetilde{\mathcal O}\big(\Gamma(Ngst+m)\big)
\end{align*}
queries and gates. A circuit with $1/s$ Trotter steps thus costs $\widetilde{\mathcal O}(\Gamma Ngt+m\Gamma/s)$. Using the step count in \cor{local-observable-estimation-main}, the per-run complexity for observable estimation is
\begin{align}
\label{eq:block-encoding-cost-main}
\widetilde{\mathcal O}\big(
\Gamma Ngt+m\Gamma\sqrt N\,(kgt)^{3/2}
\big).
\end{align}
Multiplying Eq.~\eqref{eq:block-encoding-cost-main} by the $\widetilde{\mathcal O}(1/\eps^2)$ circuit runs gives the total query and gate complexity
\begin{align*}
\widetilde{\mathcal O}\Big(
\frac{\Gamma Ngt+m\Gamma\sqrt N\,(kgt)^{3/2}}{\eps^2}
\Big).
\end{align*}
It remains to account for errors in implementing the local channels. A circuit with step size $s_j$ contains $\mathcal O(m/s_j)$ such channels, so approximating each one to diamond-norm error $\mathcal O(\eps s_j/(m\log p))$ limits the error of that circuit to $\mathcal O(\eps/\log p)$. Since $\|\bm b\|_1=\mathcal O(\log p)$, the error after Richardson extrapolation is $\mathcal O(\eps)$. This precision requirement contributes only polylogarithmic factors to the query and gate counts, which are absorbed in $\widetilde{\mathcal O}$. The work registers can be reused, so the implementation requires $a+\mathcal O(\log\Gamma)+\operatorname{polylog}(m\Gamma Ngt/\eps)$ ancillas, which is polylogarithmic when $a$ is polylogarithmic.

Table~\ref{tab:complexity_comparison} separates channel approximation from observable estimation and reports the gate or query cost of one circuit and the number of circuit runs. To compare the system-size dependence, we specialize to finite-range models on bounded-degree interaction graphs, including the spin chains in Sec.~V. Each Hamiltonian term and jump operator then acts on $\mathcal O(1)$ sites, and each site belongs to $\mathcal O(1)$ terms, so $\Gamma=1$, $k=\mathcal O(1)$, and $m=\mathcal O(N)$. In this setting, our per-circuit gate count scales as $\widetilde{\mathcal O}(N^{3/2})$ for both tasks, compared with $\widetilde{\mathcal O}(N^2)$ for channel simulation by the linear combination of unitaries (LCU)~\cite{li2023simulating,peng2025quantum} and for observable estimation by randomized dissipation (RD)~\cite{kato2024exponentially}, and $\widetilde{\mathcal O}(N^3)$ for the incoherent linear combination of superoperators (LCS)~\cite{yu2025lindbladian}. With the same gate-counting assumptions, we estimate a per-circuit cost $\widetilde{\mathcal O}(m(Ngt)^2)=\widetilde{\mathcal O}(N^3(gt)^2)$ for the step-size extrapolation schemes of Mohammadipour and Li~\cite{mohammadipour2025reducing}, with the same $N$ and $t$ dependence as LCS.

The improvement is therefore primarily in the system-size dependence. For channel approximation, LCU methods have better time and precision dependence but use polylogarithmically many ancillas, whereas our method uses $\mathcal O(1)$ reusable ancillas. The higher-order ($q>2$) dilated-Hamiltonian (DH) method~\cite{ding2024simulating} likewise has better time and precision exponents, while its ancilla cost grows with the approximation order and the number of terms. For the DH entry in Table~\ref{tab:complexity_comparison}, we specialize the local-component factor to the present input model; the time, precision, and ancilla dependences follow the resource accounting of Ref.~\cite{kato2024exponentially} under standard LCU block encodings and QSVT simulation of the dilated Hamiltonian. For observable estimation, our per-circuit time dependence is $t^{3/2}$, compared with $t^2$ for RD, LCS, and the step-size extrapolation schemes, while all require $\widetilde{\mathcal O}(1/\eps^2)$ circuit runs. In the $\Gamma=1$ setting, our method, LCS, and RD all use $\mathcal O(1)$ ancillas (two for LCS). For general $\Gamma$, our block-encoding ancilla cost is polylogarithmic when the input block encodings have polylogarithmic ancilla cost.
\begin{table}[t]
\centering
\caption{Complexity and ancilla requirements for $(\Gamma,\mathcal{O}(1))$-local Lindbladian simulation. The top and bottom sections denote channel approximation and observable estimation, respectively. All observable estimation methods require $\widetilde{\mathcal{O}}(1/\eps^2)$ circuit runs.}
\label{tab:complexity_comparison}
\begin{ruledtabular}
\begin{tabular}{l l c}
Scheme & Complexity per circuit run & Ancilla \\
\hline
LCU~\cite{li2023simulating,peng2025quantum} & $\widetilde{\mathcal{O}}(m\Gamma Ngt)$ & $\mathrm{polylog}(m\Gamma Ngt/\eps)$ \\
$q$-order DH \cite{ding2024simulating} & $\widetilde{\mathcal{O}}((m\Gamma)^q Ngt (Ngt/\eps)^{1/q})$ & $\Omega(q\log(m\Gamma))$ \\
\textbf{Ours ($\Gamma=1$)} & $\widetilde{\mathcal{O}}(m\sqrt{N}(gt)^{3/2}\eps^{-1/2})$ & $\mathcal{O}(1)$ \\
\hline
LCS~\cite{yu2025lindbladian} & $\widetilde{\mathcal{O}}(m\Gamma(Ngt)^2)$ & $2$ \\
RD \cite{kato2024exponentially} & $\widetilde{\mathcal{O}}(\Gamma(Ngt)^2)$ & $\mathcal{O}(\log \Gamma)$ \\
\textbf{Ours ($\Gamma=1$)} & $\widetilde{\mathcal{O}}(m\sqrt{N}(gt)^{3/2})$ & $\mathcal{O}(1)$ \\
\textbf{Ours } & $\widetilde{\mathcal{O}}(\Gamma Ngt + m\Gamma\sqrt{N}(gt)^{3/2})$ & $\mathrm{polylog}(m\Gamma Ngt/\eps)$ \\
\end{tabular}
\end{ruledtabular}
\end{table}
\section{Numerical demonstration}
We use two representative open spin-chain models to examine the predicted
commutator scaling. The transverse-field Ising model, a standard model of
quantum magnetism~\cite{coldea2010quantum}, provides a simple benchmark with
two noncommuting Hamiltonian components and local amplitude
damping~\cite{overbeck2017multicritical}. 
In addition, motivated by quantum simulations of one-dimensional
Heisenberg-type chains~\cite{joshi2023exploring}, we consider an isotropic
Heisenberg chain with a uniform $X$ field as a complementary benchmark. Its
$XX$, $YY$, and $ZZ$ exchange terms give a more varied commutator structure.
For this model, we consider both
amplitude damping (AD) and phase damping (PD, equivalently pure dephasing),
which are standard non-unital and unital Markovian noise mechanisms,
respectively~\cite{gardiner2004quantum,
breuer2002theory}.

Specifically, for the Ising chain, the system Hamiltonian is decomposed as $H=H_X+H_Z$, with
\begin{align*}
    H_X=-J\sum_{j=1}^{N-1}X_{j}X_{j+1},\quad H_Z=-h\sum_{j=1}^NZ_j.
\vspace{-.2em}
\end{align*}

For the Heisenberg benchmark, we similarly decompose the Hamiltonian into its
$XX$, $YY$, and $ZZ$ exchange components together with the uniform $X$-field
term. 
\begin{align*}
H=J\sum_{j=1}^{N-1}(X_jX_{j+1}+Y_jY_{j+1}+Z_jZ_{j+1})+h\sum_{j=1}^N X_j .
\end{align*}

To model local Markovian noise, the Ising benchmark uses amplitude
damping, whereas the Heisenberg benchmark is studied separately under
amplitude damping and phase damping. We label the corresponding jump
operators by AD and PD:
\begin{align*}
L_\nu^{(\mathrm{AD})}
 =\sqrt{\gamma}\ket{0}_\nu\bra{1}_\nu,\quad
L_\nu^{(\mathrm{PD})}
 =\sqrt{\gamma}Z_\nu,\quad \nu=1,\ldots,N.
\end{align*}
Here, $\gamma$ is the dissipation strength. Amplitude damping relaxes each
spin toward $\ket{0}$, whereas phase damping suppresses coherences without
changing populations.

For channel approximation, we report the fixed-input trace-norm error
$\|(e^{t\mathcal{L}}-\mathcal{S}(t/r)^r)\rho_0\|_1$, which is bounded above
by the corresponding diamond-norm error. The reference evolution is computed
by high-precision numerical integration. For observable estimation, we compute
the expectation value $\tr[O\mathcal{S}(t/r)^r\rho_0]$. We consider four
initial states to examine the input dependence: the pure all-ones state
$\ket{1^N}\bra{1^N}$, the pure all-zeros state
$\ket{0^N}\bra{0^N}$, the pure uniform superposition state
$\ketbra{+^N}{+^N}$, and the maximally mixed state $I^{\otimes N}/2^N$.
Further numerical details are provided in \append{numerics}.

\emph{Commutator Scaling.} We first examine the system-size dependence of the
fixed-input Trotter error. For a $g$-extensive $(\Gamma,k)$-local Lindbladian,
Eqs.~\eq{one-step-trotter-error} and \eq{comm-3} give an
$\mathcal{O}(N)$ upper bound on the diamond-norm Trotter error. As shown in
\fig{commutator} and \fig{heisenberg-noise}, we plot the simulation error versus the qubit number
$N\in[4,10]$ for various inputs on a log-log scale and fit the data by linear
regression. All fitted exponents are below $1$, consistent with this upper
bound. In particular, \fig{c} has the largest fitted exponent, $0.962$.

\begin{figure*}[htbp]
    \centering
    \begin{subfigure}[ht]{.5\linewidth}
        \includegraphics[width=\linewidth]{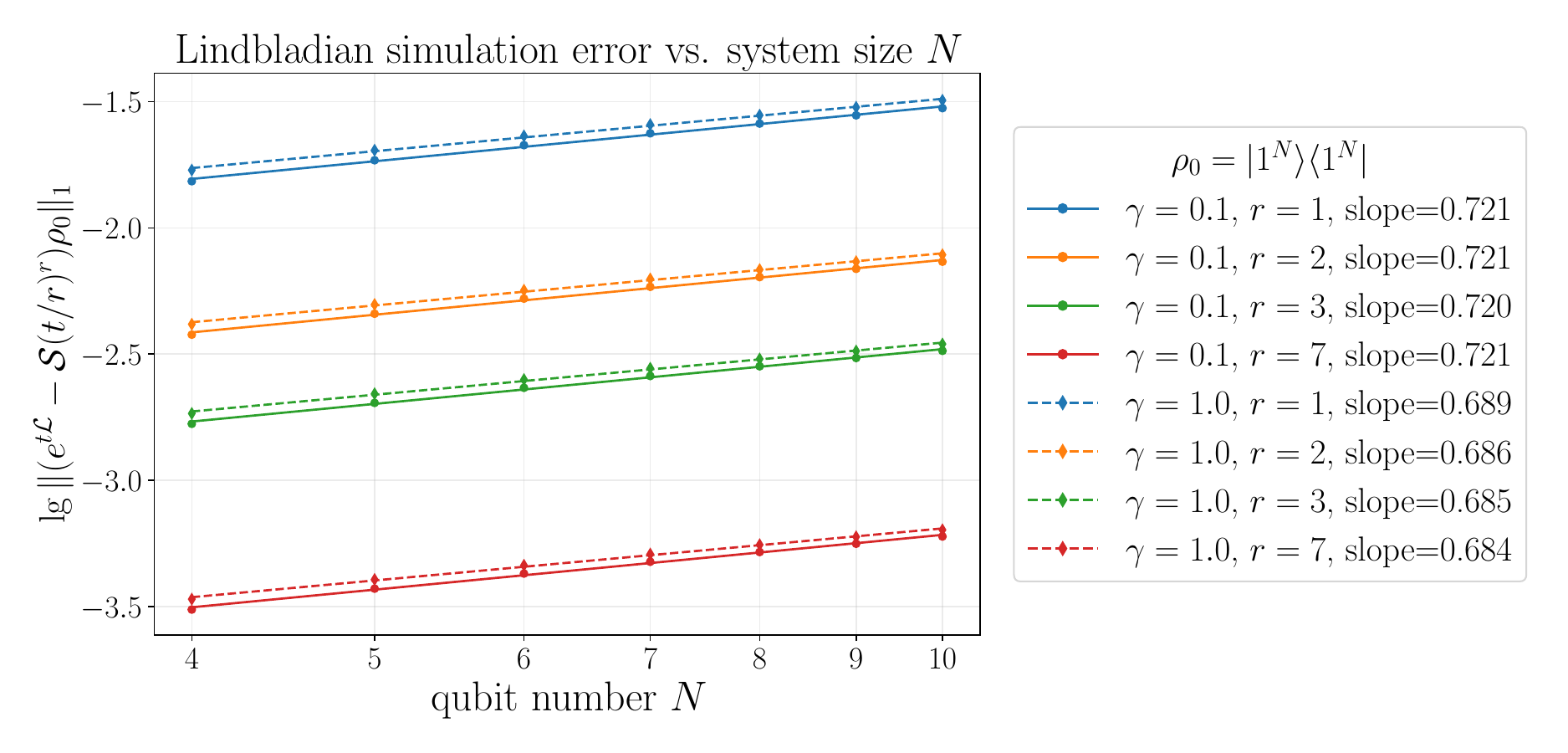}
        \caption{$\rho_0=\ket{1^N}\bra{1^N}$\hspace*{4em}}
        \label{fig:a}
    \end{subfigure}\hfill
    \begin{subfigure}[ht]{.5\linewidth}
        \includegraphics[width=\linewidth]{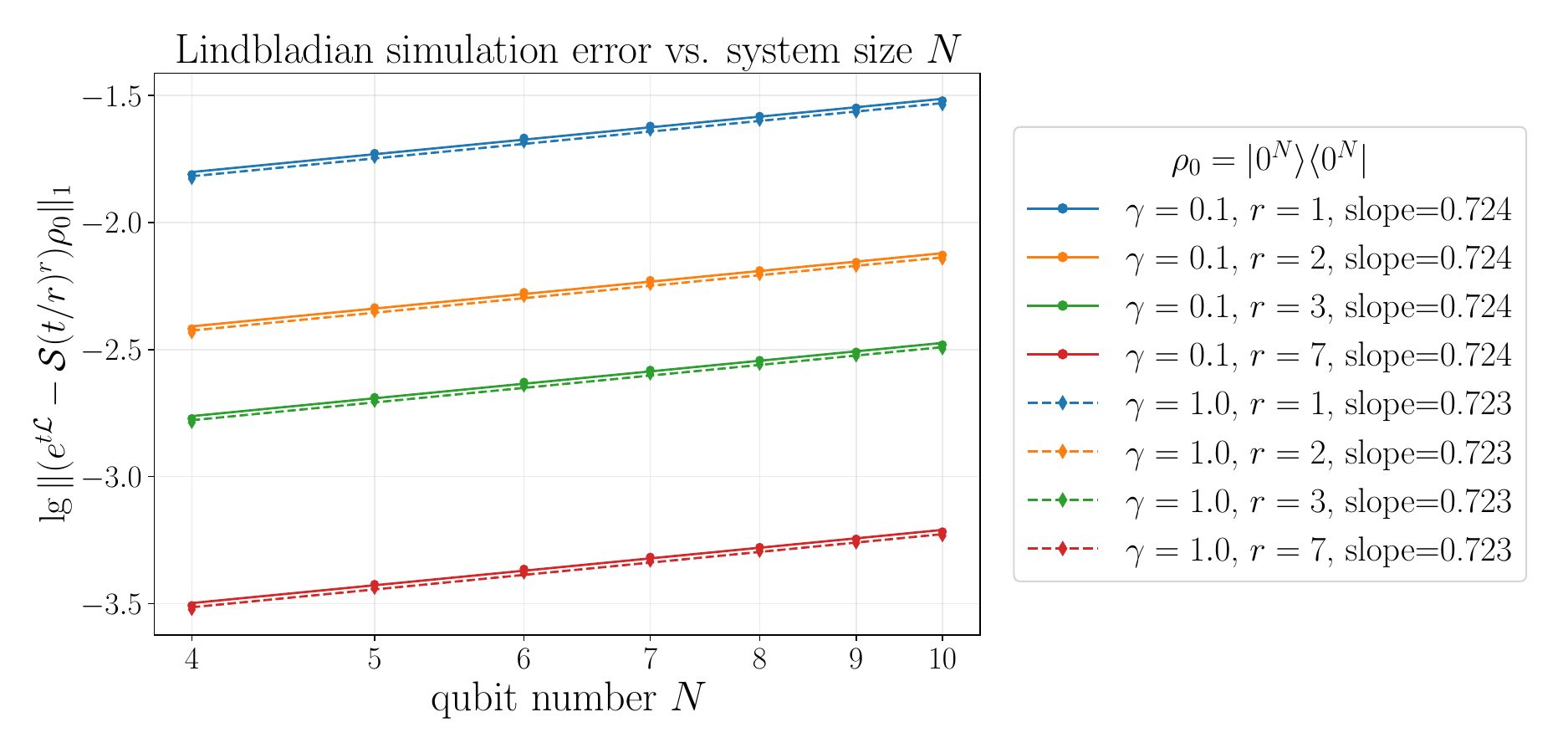}
        \caption{$\rho_0=\ket{0^N}\bra{0^N}$\hspace*{4em}}
        \label{fig:b}
    \end{subfigure}\\[1em]
    \begin{subfigure}[ht]{.5\linewidth}
        \includegraphics[width=\linewidth]{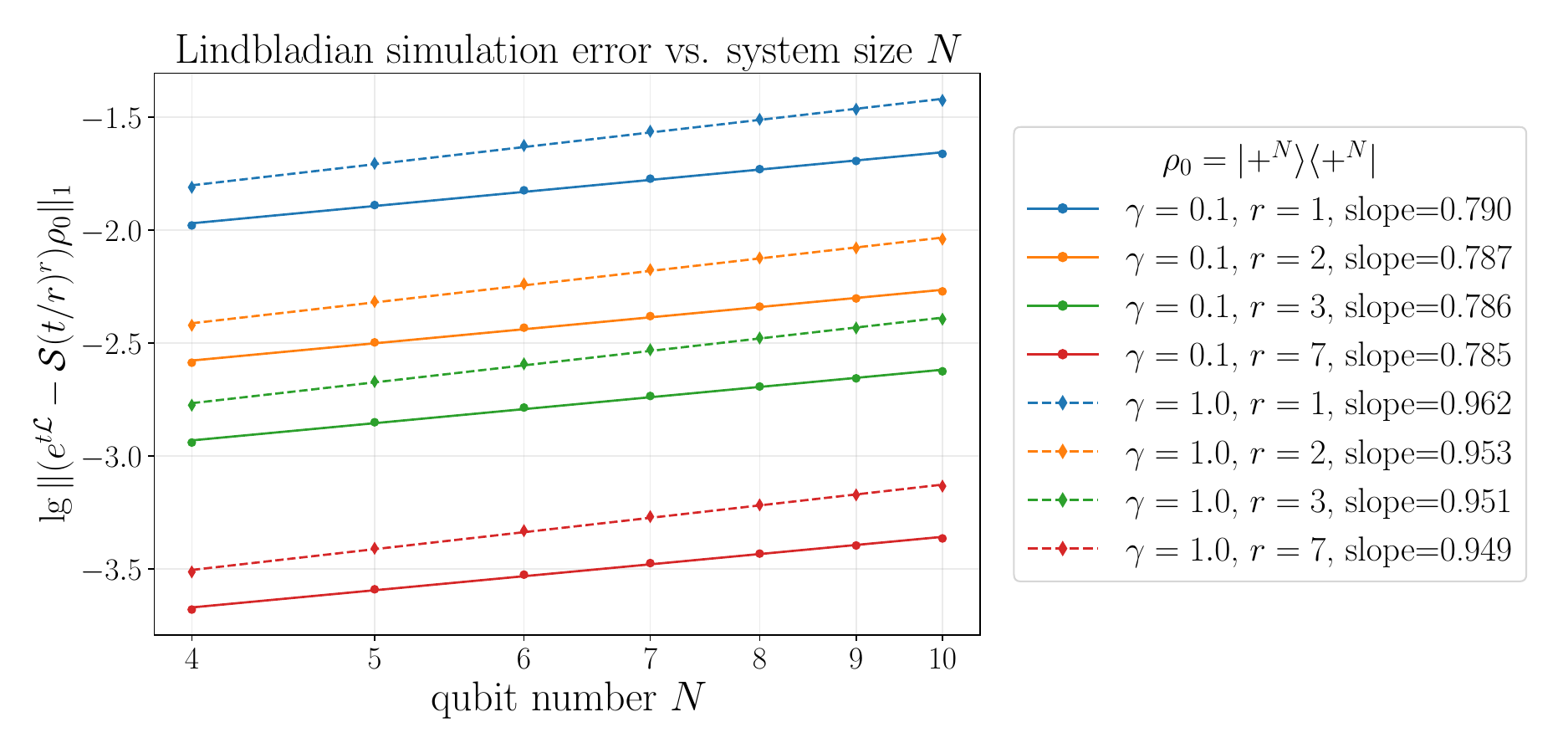}
        \caption{$\rho_0=\ket{+^N}\bra{+^N}$\hspace*{4em}}
        \label{fig:c}
    \end{subfigure}\hfill
    \begin{subfigure}[ht]{.5\linewidth}
        \includegraphics[width=\linewidth]{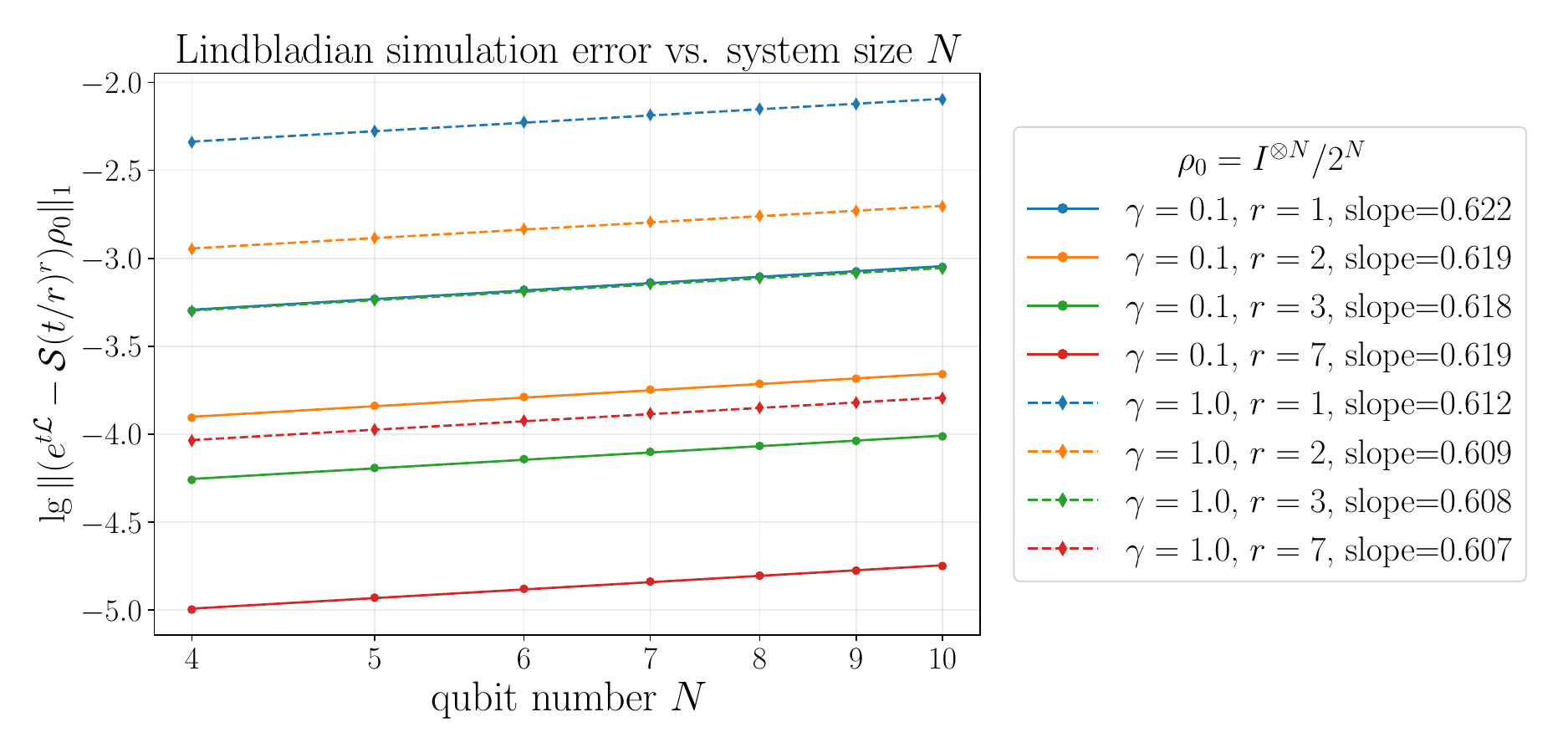}
        \caption{$\rho_0=I^{\otimes N}/2^N$\hspace*{5em}}
        \label{fig:d}
    \end{subfigure}
    \caption{\textbf{Trotter error for Ising-chain Lindbladian simulation with local relaxation versus system size $N$ with various initial states $\rho_0$.} We fix the parameters $J=1.0$, $h=0.5$, and $t=0.2$. The data points are plotted on a log-log scale and fitted by linear regression. The solid lines with circle markers denote dissipation strength $\gamma=0.1$, whereas the dashed lines with diamond markers denote $\gamma=1.0$. The colors distinguish the number of Trotter steps $r$. (a) For $\gamma=0.1$, the errors scale approximately as $\mathcal{O}(N^{0.721})$. For $\gamma=1.0$, they scale approximately as $\mathcal{O}(N^{0.689})$. (b) For $\gamma=0.1$, the errors scale approximately as $\mathcal{O}(N^{0.724})$. For $\gamma=1.0$, they scale approximately as $\mathcal{O}(N^{0.723})$. (c) The fitted exponents are the largest among the tested initial states. For $\gamma=0.1$, the errors scale approximately as $\mathcal{O}(N^{0.790})$. For $\gamma=1.0$, they scale approximately as $\mathcal{O}(N^{0.962})$. (d) The fitted exponents are the smallest, and the errors are approximately one order of magnitude smaller than for the other initial state, which might be explained by the fact that high von Neumann entropy accelerates quantum simulation~\cite{Zhao_2025}. For $\gamma=0.1$, the errors scale approximately as $\mathcal{O}(N^{0.622})$. For $\gamma=1.0$, they scale approximately as $\mathcal{O}(N^{0.612})$.}
    \label{fig:commutator}
\end{figure*}

\begin{figure*}[htbp]
    \centering
    \begin{subfigure}[ht]{.5\linewidth}
        \includegraphics[width=\linewidth]{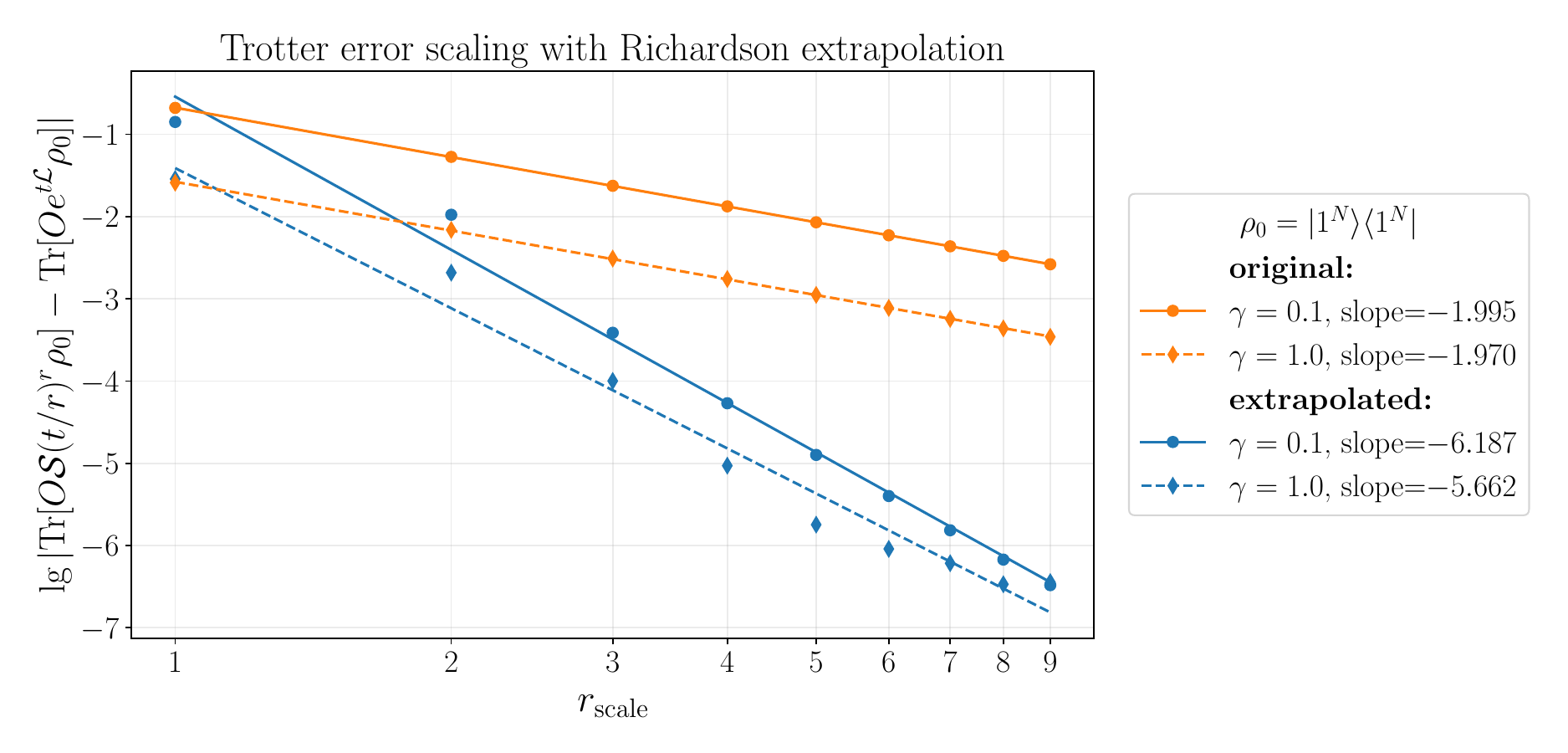}
        \caption{$O=\sum_{j=1}^N Z_j$\hspace*{4em}}
    \end{subfigure}\hfill
    \begin{subfigure}[ht]{.5\linewidth}
        \includegraphics[width=\linewidth]{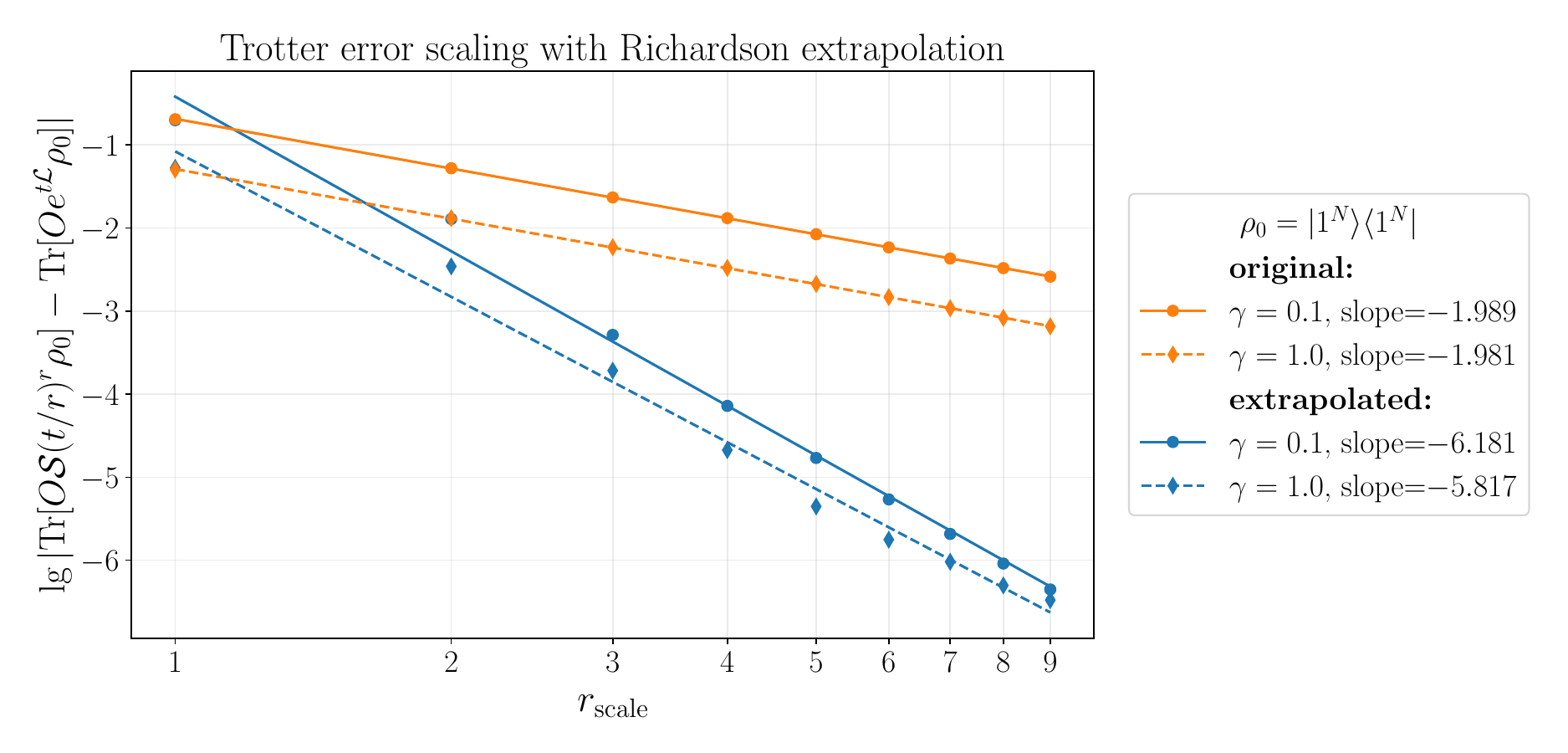}
        \caption{$O=\sum_{j=1}^{N-1} Z_jZ_{j+1}$\hspace*{4em}}
    \end{subfigure}
    \caption{\textbf{Error scaling of Ising-chain observable expectations with and without Richardson extrapolation ($N=5$).} The orange lines denote raw Trotter errors with $r=4r_{\text{scale}}$, whereas the blue lines represent extrapolated results using $r\in\{r_{\text{scale}},2r_{\text{scale}},4r_{\text{scale}}\}$. The solid and dashed lines correspond to different dissipation strengths $\gamma$.}
    \label{fig:extrapolation}
\end{figure*}

\begin{figure*}[htbp]
    \centering
    \begin{subfigure}[ht]{.5\linewidth}
        \includegraphics[width=\linewidth]{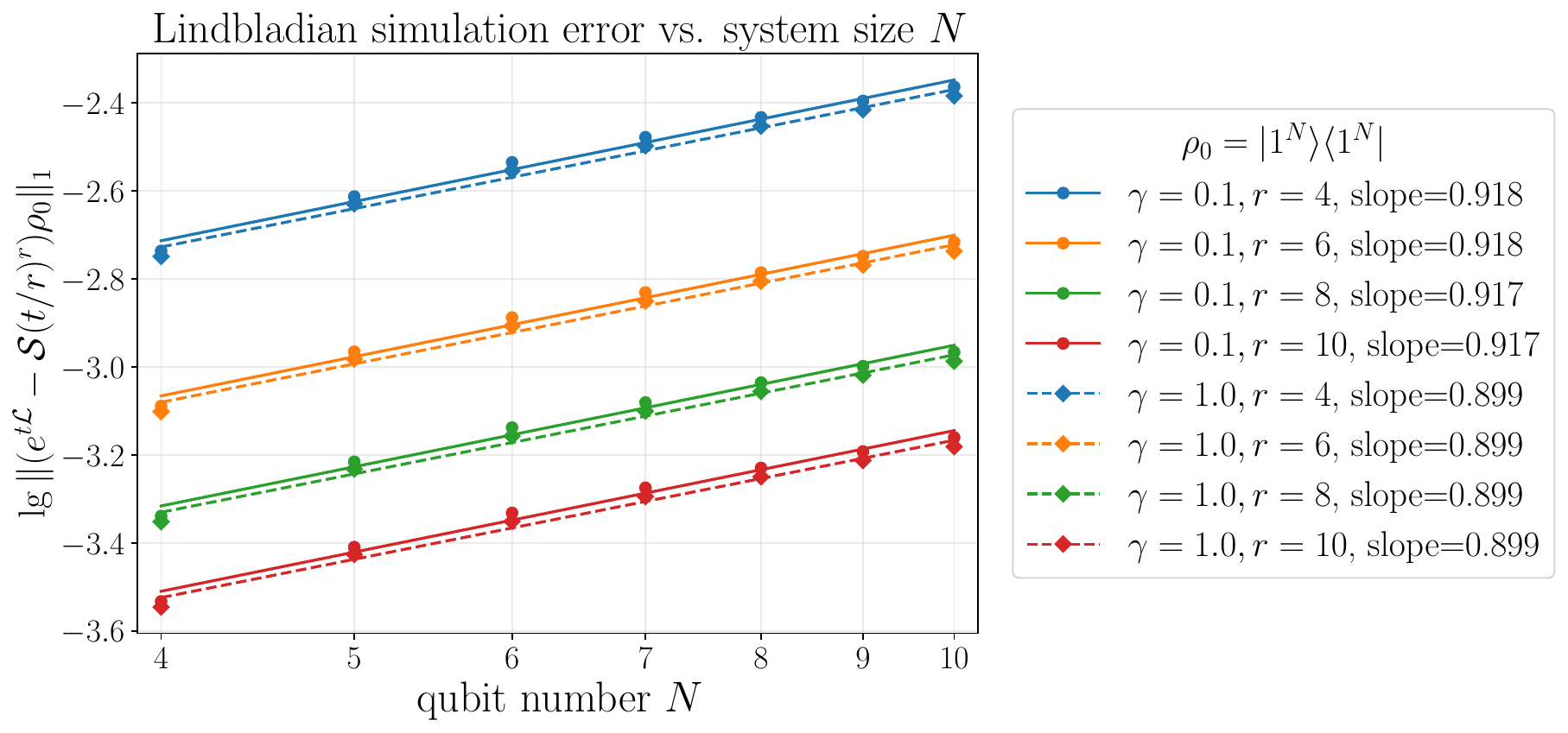}
        \caption{$\rho_0=\ket{1^N}\bra{1^N}$, $L_\nu^{(\mathrm{AD})}$}
        \label{fig:heisenberg-noise-a}
    \end{subfigure}\hfill
    \begin{subfigure}[ht]{.5\linewidth}
        \includegraphics[width=\linewidth]{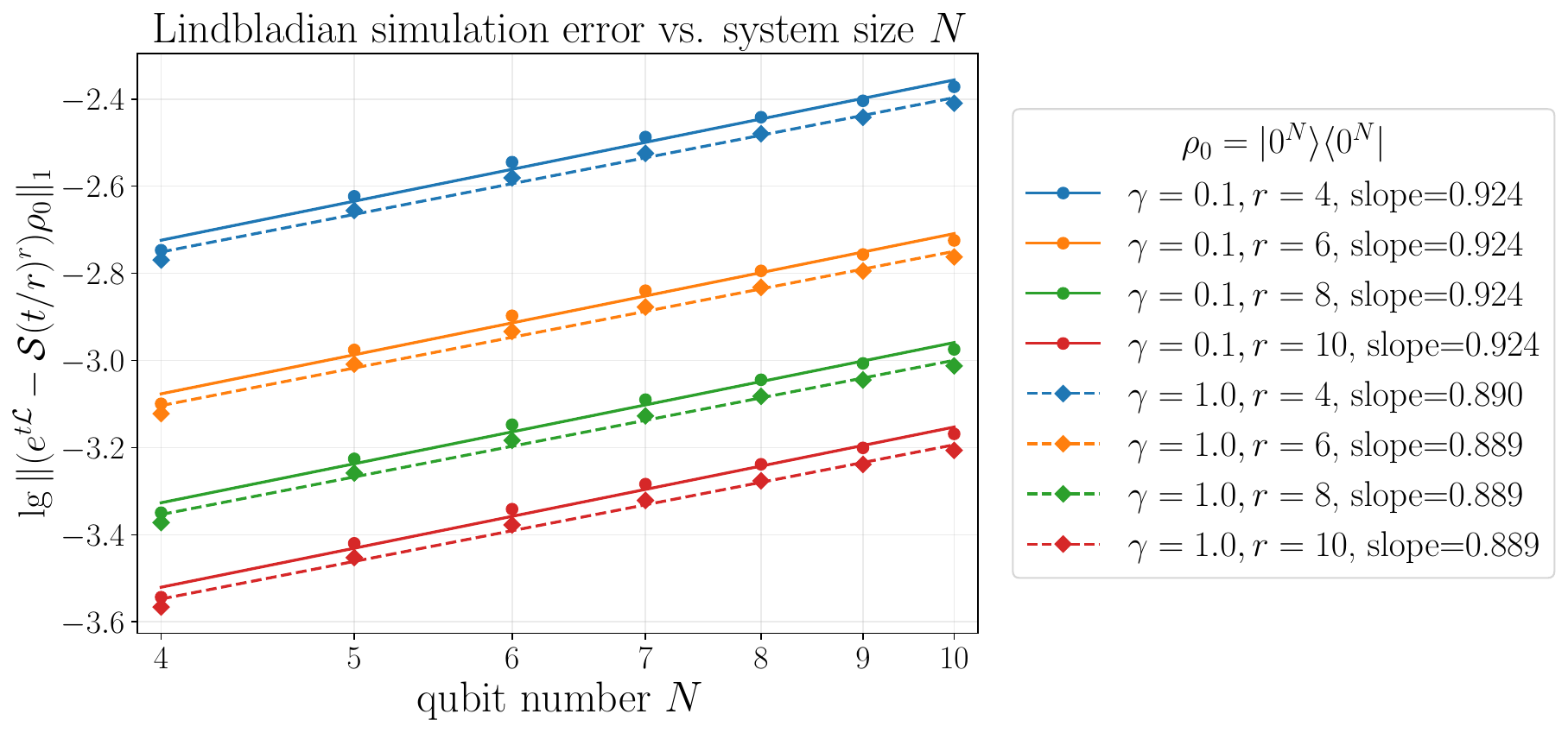}
        \caption{$\rho_0=\ket{0^N}\bra{0^N}$, $L_\nu^{(\mathrm{PD})}$}
        \label{fig:heisenberg-noise-b}
    \end{subfigure}\\[1em]
    \begin{subfigure}[ht]{.5\linewidth}
        \includegraphics[width=\linewidth]{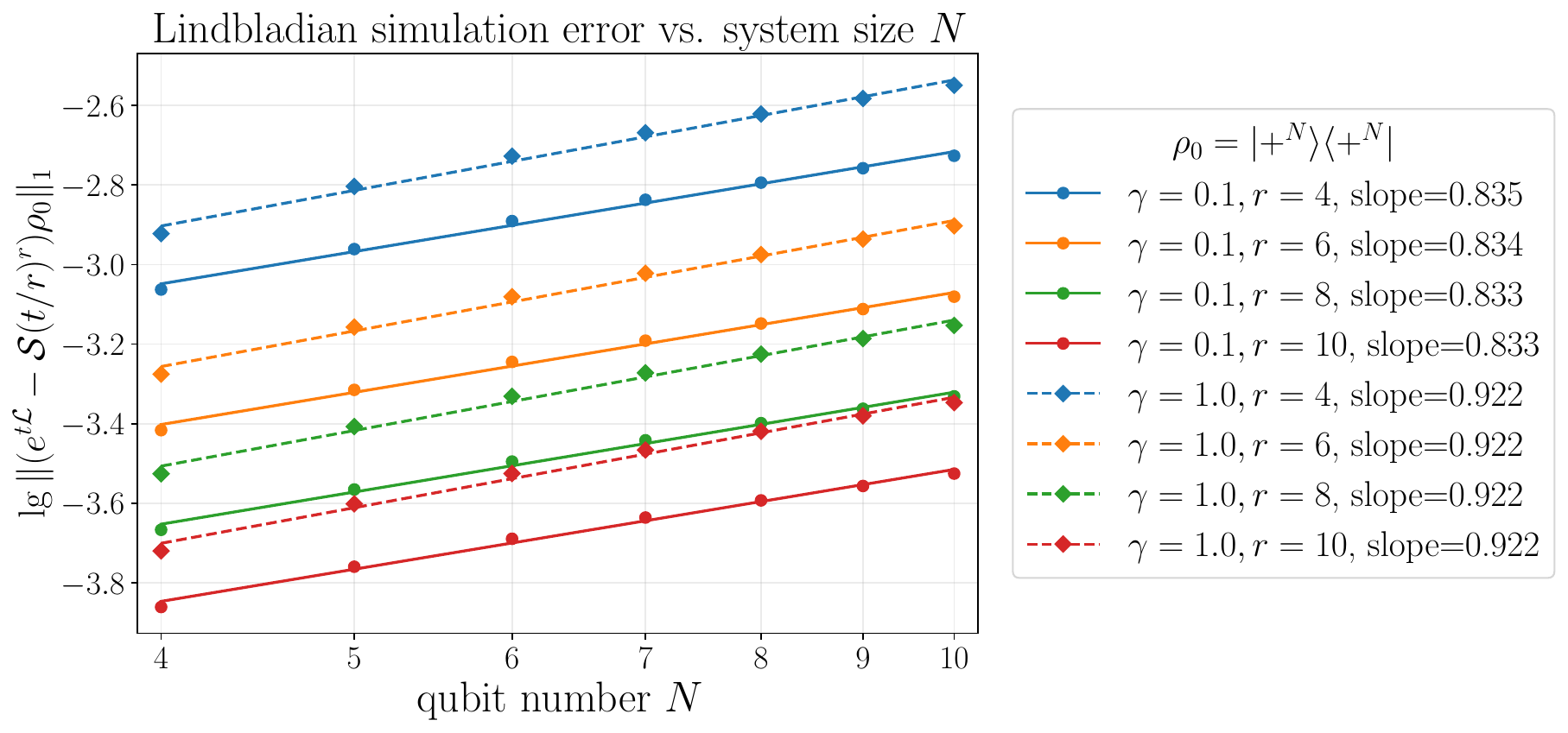}
        \caption{$\rho_0=\ket{+^N}\bra{+^N}$, $L_\nu^{(\mathrm{PD})}$}
        \label{fig:heisenberg-noise-c}
    \end{subfigure}\hfill
    \begin{subfigure}[ht]{.5\linewidth}
        \includegraphics[width=\linewidth]{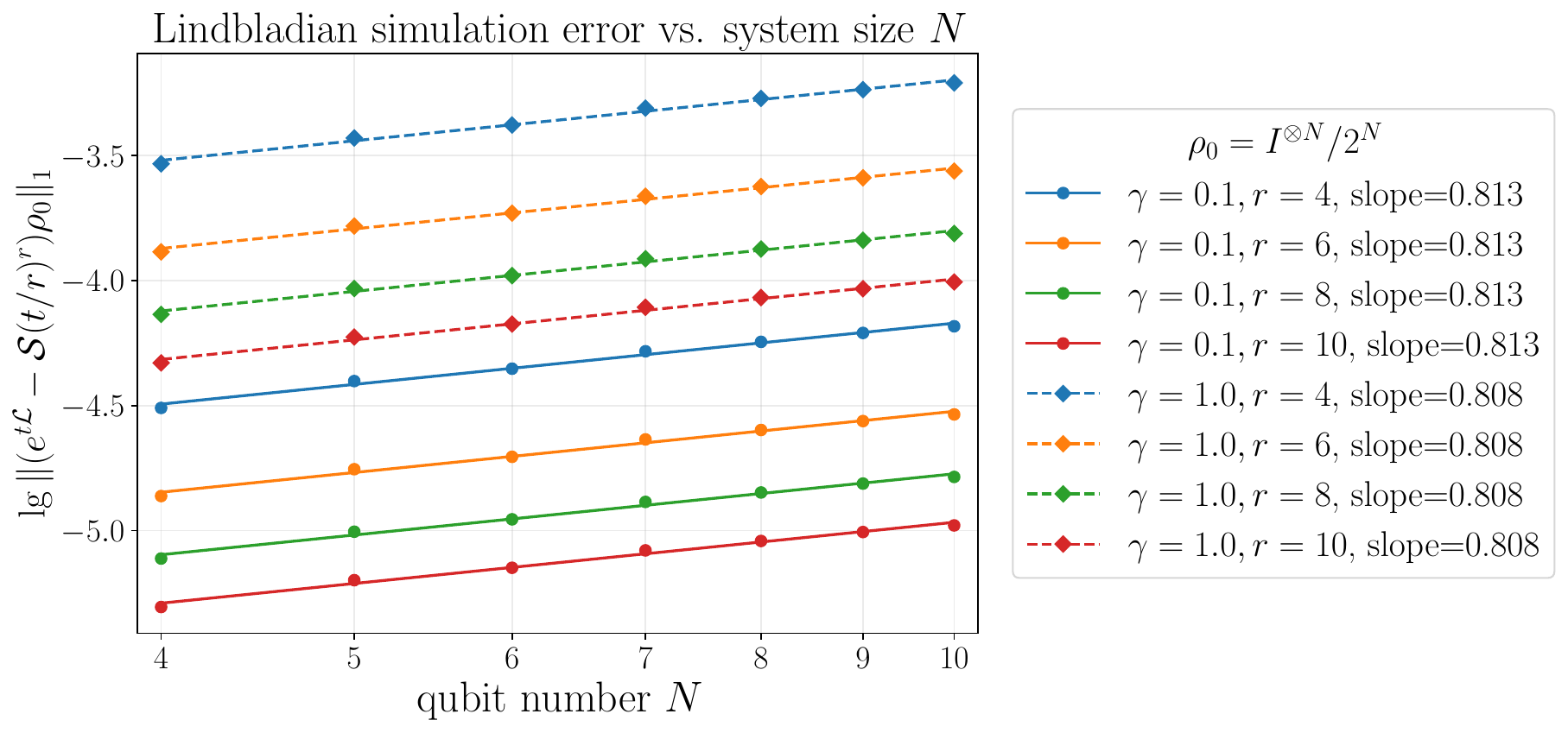}
        \caption{$\rho_0=I^{\otimes N}/2^N$, $L_\nu^{(\mathrm{AD})}$}
        \label{fig:heisenberg-noise-d}
    \end{subfigure}
    \caption{\textbf{Trotter error for Heisenberg-chain Lindbladian simulation with local amplitude damping and local phase damping versus system size $N$.}
    We fix the parameters $J=1.0$, $h=0.5$, and $t=0.15$. Panels (a) and (d) use local amplitude damping, whereas panels (b) and (c) use local phase damping. The data points are plotted on a log-log scale and fitted by linear regression. The solid lines with circle markers denote dissipation strength $\gamma=0.1$, whereas the dashed lines with diamond markers denote $\gamma=1.0$. The colors distinguish the number of Trotter steps $r$. 
    The fitted exponents in panels (a)--(d) are, respectively,
    $0.918$, $0.924$, $0.835$, and $0.813$ for $\gamma=0.1$,
    and $0.899$, $0.890$, $0.922$, and $0.808$ for $\gamma=1.0$.
    Thus, across the tested initial states and both noise mechanisms, all fitted
    exponents remain below $1$, consistent with the predicted $\mathcal{O}(N)$
    upper bound.}
    \label{fig:heisenberg-noise}
\end{figure*}

\begin{figure*}[htbp]
    \centering
    \begin{subfigure}[ht]{.5\linewidth}
        \includegraphics[width=\linewidth]{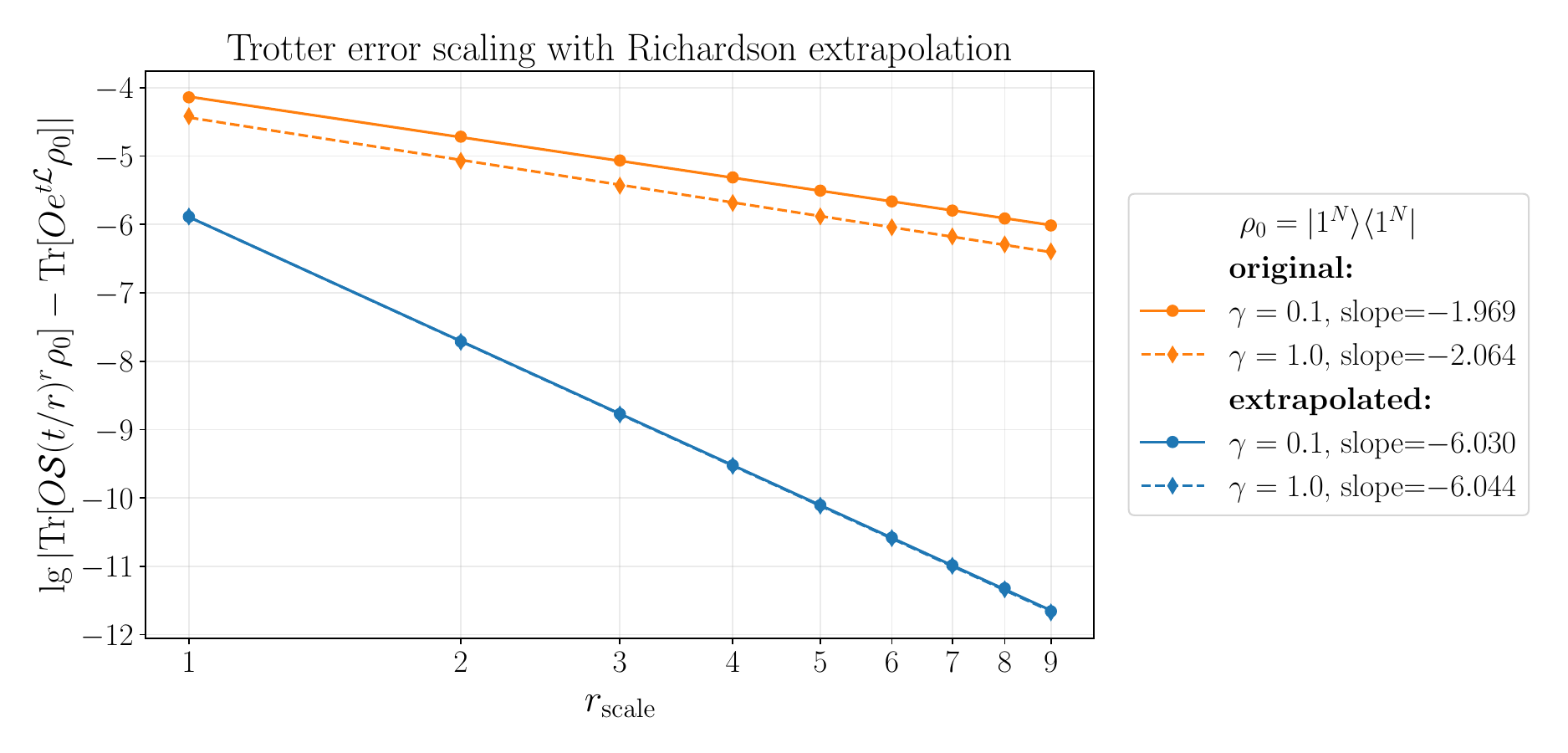}
        \caption{$O=\sum_{j=1}^N Z_j$\hspace*{4em}}
        \label{fig:heisenberg-extrap-z}
    \end{subfigure}\hfill
    \begin{subfigure}[ht]{.5\linewidth}
        \includegraphics[width=\linewidth]{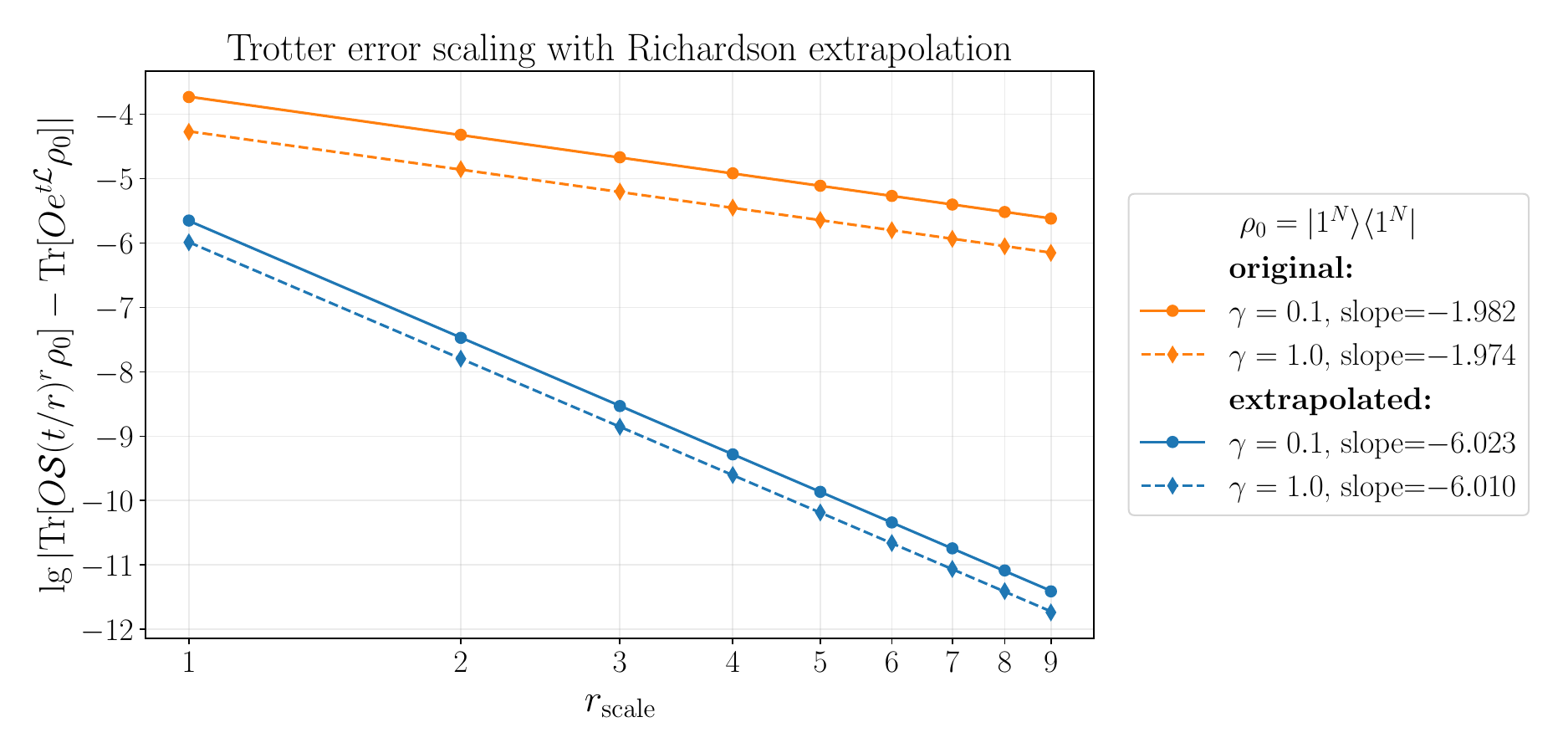}
        \caption{$O=\sum_{j=1}^{N-1} Z_jZ_{j+1}$\hspace*{4em}}
        \label{fig:heisenberg-extrap-zz}
    \end{subfigure}
    \caption{\textbf{Error scaling of Heisenberg-chain observable expectations with and without Richardson extrapolation ($N=5$).} We consider local amplitude damping and the all-ones initial state $\rho_0=\ket{1^N}\bra{1^N}$. The orange lines denote raw Trotter errors with $r=4r_{\mathrm{scale}}$, whereas the blue lines represent extrapolated results using $r\in\{r_{\mathrm{scale}},2r_{\mathrm{scale}},4r_{\mathrm{scale}}\}$. The solid and dashed lines correspond to different dissipation strengths $\gamma$.}
    \label{fig:heisenberg-extrapolation}
\end{figure*}

\emph{Richardson Extrapolation.} We next examine the suppression of Trotter bias via Richardson extrapolation. To resolve the higher-order scaling, we simulate a small system ($N=5$) whose evolution channels can be computed directly via superoperator exponentials.
As shown in \fig{extrapolation} and \fig{heisenberg-extrapolation}, after extrapolation with $p=3$ step sizes, the errors decrease by several orders of magnitude. The extrapolated errors scale approximately as $r_{\mathrm{scale}}^{-6}$ for both the magnetization observable $O=\sum_{j=1}^N Z_j$ and the two-point correlator $O=\sum_{j=1}^{N-1} Z_jZ_{j+1}$, consistent with the predicted $2p$-th order extrapolation remainder for $p=3$.

Taken together, these two numerical studies examine complementary consequences of our analysis. The system-size sweeps probe the system-size dependence of the fixed-input Trotter error: across the TFIM and Heisenberg benchmarks, all nontrivial state-dependent errors remain consistent with the $\mathcal{O}(N)$ diamond-norm upper bound. By contrast, the step-size sweeps probe the cancellation underlying Richardson extrapolation, changing the observed leading error from approximately second to sixth order.
The system-size tests cover two models, four initial states, and two noise mechanisms, while the Richardson tests cover two models and two observables.

\section{Discussion}
In this paper, we analyze two Trotter-based quantum algorithms for Lindbladian simulation and show that their complexities scale with the nested commutators of the operator summands. The first approximates the time-evolution channel within a small diamond distance, while the second uses Richardson extrapolation so that the number of Trotter steps per circuit depends polylogarithmically on $1/\eps$. To derive these commutator bounds, we develop a general truncation bound for the BCH expansion that bypasses common convergence issues in analyzing high-precision Trotter methods for both open and closed systems~\cite{aftab2024multi,watson2025exponentially, chakraborty2025quantum,mizuta2026commutator}. Our analysis shows that product formulas achieve improved system-size complexity scaling for simulating $g$-extensive $(\Gamma,k)$-local Lindbladians relative to prior bounds. While we use two specific circuit implementations as examples, the Trotter error bounds are independent of these details. This suggests that the cost per Trotter step could be further reduced using methods tailored to specific systems. Finally, we conduct numerical simulations to support our theory.

\section*{Acknowledgements}
Xinzhao Wang, Shuo Zhou, and Tongyang Li are supported by the National Natural Science Foundation of China (Grant Numbers 62372006 and 92365117). Xiaoyang Wang is supported by the RIKEN TRIP initiative (RIKEN Quantum) and the UTokyo Quantum Initiative.

\section*{Data Availability}
The data and code that support the findings of this article are openly available~\cite{data}.

\clearpage
\appendix

\section{Preliminaries}\label{append:pre}
\subsection{Notations}
For a sequence of operators $\A_1, \ldots, \A_{\Gamma}$, we denote $\prod_{\gamma=1}^\Gamma \A_{\gamma}= A_{\Gamma}\A_{\Gamma-1} \ldots \A_1$ and $\prod_{\gamma=\Gamma}^1 \A_{\gamma}= A_{1}\A_{2} \ldots \A_\Gamma$. We denote the right-nested commutator $[X_1, [X_2, [\cdots[X_{n-1}, X_n]\cdots]]]$ by $[X_1, X_2, \ldots, X_n]$. We use $\supp(X)$ to denote the set of sites that $X$ nontrivially acts on.  
Define the time-ordering operator $\mc T\{\cdot\}$ as \begin{align*}
    &\mc T\big\{[O_1(t_1),O_2(t_2),\ldots, O_q(t_q) ]\big\} \\
    &\quad= [O_{\pi(1)}(t_{\pi(1)}),O_{\pi(2)}(t_{\pi(2)}),\ldots,
        O_{\pi(q)}(t_{\pi(q)})],
\end{align*}
where $\pi\colon [q]\to[q]$ is a permutation of $[q]$ such that \[
    t_{\pi(1)}\ge t_{\pi(2)}\ge \cdots \ge t_{\pi(q)}.
    \]
For any permutation $\sigma\colon [q]\to [q]$, define the reordering operator $\mc R_{\sigma}\{\cdot\}$ as \begin{align*}
    &\mc R_{\sigma}\big\{[O_1(t_1),O_2(t_2),\ldots, O_q(t_q)]\big\}\\
    &\quad= [O_{\sigma(1)}(t_{\sigma(1)}),O_{\sigma(2)}(t_{\sigma(2)}),\ldots,
        O_{\sigma(q)}(t_{\sigma(q)})].
\end{align*}
\subsection{Magnus expansion}
Consider a time-dependent linear differential equation \begin{align*}
    \frac{\d Y(t)}{\d t} = A(t) Y(t), \quad Y(0) = I.
\end{align*} Formally, the Magnus expansion expresses $Y(t)$ as the exponential of a time-dependent operator $Y(t) = \exp(\Omega(t))$, where the exponent $\Omega(t)$ admits a series expansion~\cite{blanes2009magnus} \begin{align*}
    \Omega(t) = \sum_{q = 1}^{\infty} \Omega_q(t).
\end{align*} 
The $p$-th order Magnus expansion is denoted by \begin{align}
    \Omega_{(p)} = \sum_{q=1}^{p} \Omega_q. 
\end{align}
Each $\Omega_q(t)$ can be written as an integral of nested commutators of $A(t)$:
\begin{align}
\label{eq:magnus-formula}
    \Omega_q(t)
    &= \sum_{\sigma\in S_q} c_{\sigma, q}
        \int_{0}^t \d \tau_1 \int_{0}^{\tau_1}\d \tau_2
        \cdots \int_{0}^{\tau_{q-1}}\d \tau_q \nonumber\\
    &\quad \mc R_{\sigma}\big\{[A(\tau_1), A(\tau_2),
        \ldots, A(\tau_q)] \big\},
\end{align}
where $S_q$ denotes the set of permutations of $[q]$, \begin{align}
\label{eq:def-c-sigma}
    c_{\sigma, q}=\frac{1}{q^2}\frac{(-1)^{d_{\sigma}}}{\binom{q-1}{d_\sigma}} 
\end{align} satisfying $|c_{\sigma, q}|\le 1/q^2$, and $d_{\sigma}$ denotes the number of $i\in[q-1]$ such that $\sigma(i)>\sigma(i+1)$~\cite{Arnal_2018}. The derivative of $\Omega_q(t)$ is
\begin{align}
    \dot{\Omega}_q(t)
    ={}& \sum_{\sigma\in S_q} c_{\sigma, q}
        \int_{0}^{t}\d \tau_2 \int_0^{\tau_2} \d\tau_3
        \cdots \int_{0}^{\tau_{q-1}}\d \tau_q \nonumber\\
    &\qquad \mc R_{\sigma}\big\{[A(t), A(\tau_2),
        \ldots, A(\tau_q)] \big\}.
\end{align}
\subsection{Richardson extrapolation}\label{append:richardson}

\begin{lemma}[{see \cite[Lemma~5]{watson2025exponentially}}]\label{lem:richardson-even}
    Let $f(x)\in C^{2p+2}([-1,1])$ be an  even function, and let $Q_j$ and $R_j$ be a degree-$(j-1)$ Taylor polynomial and Taylor remainder, respectively, such that $f(x)=Q_j(x)+R_j(x)$. 
    Let $F^{(p)}(s)$ be the $p$-th order Richardson extrapolation of $f(x)$ at points $s_j = s/r_j$ for $j \in [p]$, defined as
    \[
    F^{(p)}(s) = \sum_{j=1}^p b_j f(s_j),
    \]
    where the repetition numbers $r_j$ and coefficients $b_j$ are given by
    \begin{align*}
        r_j = \left\lceil \frac{\sqrt{8} p}{\pi \sin (\pi(2 j-1) / 8 p)} \right\rceil, \quad 
        b_j = \prod_{\ell \neq j} \frac{1}{1 - r_\ell^2/r_j^2}.
    \end{align*}
    Then, the extrapolation error satisfies
    \[
    |F^{(p)}(s) - f(0)| \le \|\bm{b}\|_1 \max_{j \in [p]} |R_{2p}(s_j)|,
    \]
    where $\|\bm{b}\|_1 \le C \log p$ for some absolute constant $C > 0$.
\end{lemma}
\subsection{Technical lemmas}
Let $\mc A(\tau)$ and $\mc B(\tau)$ be two bounded piecewise-continuous operator-valued functions. For any bounded piecewise-continuous generator $\mc G(\tau)$, define its time-ordered propagator by
\begin{align*}
\mc U_{\mc G}(t,s)
:=\exp_{\mc T}\biggl(\int_s^t \mc G(\tau)\,\d\tau\biggr).
\end{align*}
We will repeatedly use the following standard identity for linear differential equations, known as the variation-of-constants formula (or Duhamel's principle)~\cite{pazy1983semigroups}.
\begin{lemma}[Variation of constants formula]
\label{lem:variation-constants}
Let $\mc H(\tau)=\mc A(\tau)+\mc B(\tau)$ be a time-dependent operator with bounded piecewise-continuous summands. Then
\begin{align*}
\mc U_{\mc H}(t,0)
={}& \mc U_{\mc A}(t,0)+\int_0^t \mc U_{\mc A}(t,\tau)
\mc B(\tau)\mc U_{\mc H}(\tau,0)\,\d\tau.
\end{align*}
\end{lemma}

We also rely on the following interaction-picture representation of the evolution generated by $\mc H(\tau)=\mc A(\tau)+\mc B(\tau)$.
\begin{lemma}[Interaction picture representation, see {\cite[Page 21]{dollard1984product}}]
\label{lem:interaction-pic}
Let $\mc H(\tau)=\mc A(\tau)+\mc B(\tau)$ be a time-dependent operator with bounded piecewise-continuous summands. Then
\begin{align*}
\mc U_{\mc H}(t,0)
={}& \mc U_{\mc A}(t,0)\,
\exp_{\mathcal T}\bigg(\int_0^t
\mc U_{\mc A}(0,\tau)\mc B(\tau)
\mc U_{\mc A}(\tau,0)\,\d\tau\bigg).
\end{align*}
\end{lemma} 
We use the following lemma to bound the difference between powers of a superoperator and a quantum channel. 
\begin{lemma}
\label{lem:power-bound}
    Let $\mc N_1, \mc N_2$ be two superoperators, where $\mc N_1$ is a quantum channel. For any positive integer $k$, the difference between $\mc N_1^k$ and $\mc N_2^{k}$ can be bounded by
    \begin{align*}
        \|\mc N_1^k-\mc N_2^k\|_{\diamond} \le k \max\{1, \|\mc N_2\|_{\diamond}\}^{k-1} \|\mc N_1-\mc N_2\|_{\diamond}.
    \end{align*}
\end{lemma}
\begin{proof}
    We use the telescoping sum identity for the difference between powers of two operators:
    \begin{align*}
        \mc N_1^k - \mc N_2^k = \sum_{j=0}^{k-1} \mc N_1^{k-1-j} (\mc N_1 - \mc N_2) \mc N_2^j.
    \end{align*}
    Taking the diamond norm on both sides, we obtain
    \begin{align*}
        \|\mc N_1^k - \mc N_2^k\|_{\diamond} &\le \sum_{j=0}^{k-1} \|\mc N_1\|_{\diamond}^{k-1-j} \|\mc N_1 - \mc N_2\|_{\diamond} \|\mc N_2\|_{\diamond}^j \\
        &= \|\mc N_1 - \mc N_2\|_{\diamond} \sum_{j=0}^{k-1} \|\mc N_2\|_{\diamond}^j \\
        &\le k\max\{1, \|\mc N_2\|_{\diamond}\}^{k-1}\|\mc N_1-\mc N_2\|_{\diamond},
    \end{align*}
    where the second line follows from the fact that $\mc N_1$ is a quantum channel, so $\|\mc N_1\|_{\diamond}=1$. 
\end{proof}

\section{Commutator bounds for the product formula of Lindbladians}
\label{append:comm-bound-product}
In this section, we derive commutator-based error bounds for the first- and second-order product formulas of Lindbladians. The following theorem establishes the bound for the second-order case.
\begin{theorem}
    Let $\mathcal{L} = \sum_{j=1}^m \mathcal{L}_j$ be a Lindbladian consisting of $m$ summands and let $t > 0$. 
    Let $\mathcal{S}(t) = \prod_{j=m}^1 e^{t \mathcal{L}_j /2} \prod_{j=1}^m e^{t \mathcal{L}_j /2}$ be the second-order product formula.
    The additive Trotter error is bounded by
    \begin{align*}
        \|\mathcal{S}(t) - e^{t\mathcal{L}}\|_{\diamond}
        \le {}& \frac{t^3}{12}\sum_{j_1 =1}^m
        \bigg\|\bigg[\sum_{j_3=j_1+1}^{m} \mc L_{j_3},
        \sum_{j_2=j_1+1}^{m} \mc L_{j_2}, \mc L_{j_1}\bigg]\bigg\|_{\diamond}\\
        &\quad+\frac{t^3}{24}\sum_{j_1=1}^m
        \bigg\|\bigg[\mc L_{j_1}, \mc L_{j_1},
        \sum_{j_2=j_1+1}^{m} \mc L_{j_2} \bigg]\bigg\|_{\diamond}.
    \end{align*}
\end{theorem}
\begin{proof}
    The proof generalize the derivation in \cite[Appendix L]{childs2021theory} from Hamiltonians to Lindbladians. We first consider the $m=2$ case. Following \cite[Eq.~(L4)]{childs2021theory}, the additive Trotter error can be expressed as
    \begin{widetext}
    \begin{align*}
        &\|e^{t\mc L_1/2}e^{t\mc L_2}e^{t\mc L_1/2}
        - e^{t(\mc L_1+\mc L_2)}\|_{\diamond}
        \\
        ={}&\left\|
        \int_{0}^t \d \tau_1\int_{0}^{\tau_1}\d \tau_2
        \int_{0}^{\tau_2}\d \tau_3\,
        e^{(t-\tau_1)\mc L}e^{\tau_1 \mc L_1/2}\Big[e^{\tau_3 \ad_{\mc L_2}}\ad^2_{\mc L_2}(\mc L_1 / 2 )
        +e^{-\tau_3 \ad_{\mc L_1}/2}\ad_{-\mc L_1/2}^2(\mc L_2)\Big]
        e^{\tau_1\mc L_2}e^{\tau_1\mc L_1 /2}
        \right\|_{\diamond}.
    \end{align*}
    The adjoint exponentials can be expanded and combined with the adjacent evolution factors as
    \begin{align*}
        e^{\tau_3\ad_{\mc L_2}}\ad^2_{\mc L_2}(\mc L_1/2)e^{\tau_1\mc L_2}
        ={}&e^{\tau_3\mc L_2}\ad^2_{\mc L_2}(\mc L_1/2)e^{(\tau_1-\tau_3)\mc L_2},
        \\
        e^{\tau_1\mc L_1/2}e^{-\tau_3\ad_{\mc L_1}/2}\ad_{-\mc L_1/2}^2(\mc L_2)
        ={}&e^{(\tau_1-\tau_3)\mc L_1/2}\ad_{-\mc L_1/2}^2(\mc L_2)e^{\tau_3\mc L_1/2}.
    \end{align*}
    \end{widetext}
    Since $0\le \tau_3\le \tau_2\le \tau_1\le t$, all evolution times on the right-hand sides are nonnegative. Hence the exponential factors in the expanded integrand are quantum channels and have diamond norm one. The additive Trotter error can therefore be bounded by \begin{align*}
        &\|e^{t\mc L_1/2}e^{t\mc L_2}e^{t\mc L_1/2}
        - e^{t(\mc L_1+\mc L_2)}\|_{\diamond} \\
        \le {}& \int_{0}^t \d \tau_1\int_{0}^{\tau_1}\d \tau_2
        \int_{0}^{\tau_2}\d \tau_3 \Big(\big\|\ad^2_{\mc L_2}(\mc L_1/2)\big\|_{\diamond}
        +\big\|\ad_{\mc L_1/2}^2(\mc L_2)\big\|_{\diamond}\Big) \\
        ={}& \int_{0}^t \d \tau_1\int_{0}^{\tau_1}\d \tau_2
        \int_{0}^{\tau_2}\d \tau_3 \Big(\frac{1}{2}\|[\mc L_2, \mc L_2, \mc L_1]\|_{\diamond}
        +\frac{1}{4}\|[\mc L_1, \mc L_1, \mc L_2]\|_{\diamond}\Big) \\
        ={}& \frac{t^3}{12}\|[\mc L_2, \mc L_2, \mc L_1]\|_{\diamond}+\frac{t^3}{24}\|[\mc L_1, \mc L_1, \mc L_2]\|_{\diamond}. 
    \end{align*}
    For a general Lindbladian $\mc L = \sum_{j=1}^m \mc L_j$, applying the triangle inequality via a telescoping sum argument gives \begin{align*}
        &\bigg\|\prod_{j=m}^1 e^{t \mathcal{L}_j /2} \prod_{j=1}^m e^{t \mathcal{L}_j /2}-e^{t\sum_{j=1}^m \mc L_j}\bigg\|_{\diamond} \\
        \le {}& \sum_{j_1 =1}^m
        \begin{aligned}[t]
        \bigg\|&\prod_{j_2=j_1}^1 e^{t\mc L_{j_2}/2}
        e^{t\sum_{j_2=j_1+1}^m \mc L_{j_2}}
        \prod_{j_2=1}^{j_1}e^{t\mc L_{j_2}/2}\\
        &-\prod_{j_2=j_1-1}^1 e^{t\mc L_{j_2}/2}
        e^{t\sum_{j_2=j_1}^m \mc L_{j_2}}
        \prod_{j_2=1}^{j_1-1}e^{t\mc L_{j_2}/2}\bigg\|_{\diamond} 
        \end{aligned}
        \\
        \le {}& \sum_{j_1 =1}^m
        \begin{aligned}[t]
        \bigg\|&e^{t\mc L_{j_1}/2} e^{t\sum_{j_2=j_1+1}^m \mc L_{j_2}}
        e^{t\mc L_{j_1}/2}- e^{t\sum_{j_2=j_1}^m \mc L_{j_2}}
        \bigg\|_{\diamond}
        \end{aligned} \\
        \le {}& \frac{t^3}{12}\sum_{j_1 =1}^m
        \bigg\|\bigg[\sum_{j_3=j_1+1}^{m} \mc L_{j_3},
        \sum_{j_2=j_1+1}^{m} \mc L_{j_2}, \mc L_{j_1}\bigg]\bigg\|_{\diamond}\\
        &\quad+\frac{t^3}{24}\sum_{j_1=1}^m
        \bigg\|\bigg[\mc L_{j_1}, \mc L_{j_1},
        \sum_{j_2=j_1+1}^{m} \mc L_{j_2} \bigg]\bigg\|_{\diamond},  
    \end{align*}
    where the last inequality follows from the Trotter error bound derived above for the $m=2$ case. 
\end{proof}
This error bound leads to the following requirement on the Trotter number to achieve a target precision $\eps$
\begin{corollary}
    For any evolution time $t > 0$ and target precision $\eps > 0$, the Lindbladian evolution $e^{t\mathcal{L}}$ can be approximated to precision $\eps$ in the diamond norm by $\mathcal{S}(t/r)^r$, where
    \begin{align*}
        r = \mathcal{O}\Bigg( \frac{t^{3/2}}{\sqrt{\eps}}
        \Bigg(
        &\sum_{j_1 =1}^m \bigg\|\bigg[\sum_{j_3=j_1+1}^{m} \mc L_{j_3},
        \sum_{j_2=j_1+1}^{m} \mc L_{j_2}, \mc L_{j_1}\bigg]\bigg\|_{\diamond}\\
        &+ \sum_{j_1=1}^m\bigg\|\bigg[\mc L_{j_1}, \mc L_{j_1},
        \sum_{j_2=j_1+1}^{m} \mc L_{j_2} \bigg]\bigg\|_{\diamond}
        \Bigg)^{1/2} \Bigg).
    \end{align*}
\end{corollary}
\begin{proof}
    By the triangle inequality, the total approximation error is bounded by $\|\mathcal{S}(t/r)^r - e^{t\mathcal{L}}\|_{\diamond} \le r \|\mathcal{S}(t/r) - e^{t\mathcal{L}/r}\|_{\diamond}$. Applying the bound from the preceding theorem with step size $t/r$ yields a total error of $\mathcal{O}(r (t/r)^3) = \mathcal{O}(t^3/r^2)$. Equating this bound to $\eps$ gives the result.
\end{proof}
If $\mc L$ is a $(\Gamma,k)$-Lindbladian, we have  \begin{align*}
    &\sum_{j_1 =1}^m \bigg\|\bigg[\sum_{j_3=j_1+1}^{m} \mc L_{j_3},
    \sum_{j_2=j_1+1}^{m} \mc L_{j_2}, \mc L_{j_1}\bigg]\bigg\|_{\diamond} \\
    &\qquad
    + \sum_{j_1=1}^m\bigg\|\bigg[\mc L_{j_1}, \mc L_{j_1},
    \sum_{j_2=j_1+1}^{m} \mc L_{j_2} \bigg]\bigg\|_{\diamond} \\
    \le {}& 2\sum_{j_1,j_2,j_3 =1}^m
    \bigg\|\bigg[ \mc L_{j_3}, \mc L_{j_2}, \mc L_{j_1}\bigg]\bigg\|_{\diamond} \\
    \le {}& 2\sum_{v_1,v_2,v_3 =1}^M
    \bigg\|\bigg[ \mc K_{v_3}, \mc K_{v_2}, \mc K_{v_1}\bigg]\bigg\|_{\diamond} \\
    ={}& \mathcal{O}(k^2 g^3 N),
\end{align*}
where the third line follows from decomposing the Lindbladians $\{\mc L_j\}$ into local superoperators $\{\mc K_v\}$ each supported on at most $2k$ sites, and applying the triangle inequality. The last line follows from \cor{commu-bound-sum}. The number of Trotter steps then equals
\begin{align*}
    r =  \mathcal{O}(\sqrt{N} k(gt)^{3/2}\eps^{-1/2}).
\end{align*}

Similarly, for the first-order case, we have the following theorem.
\begin{theorem}
    Let $\mathcal{L} = \sum_{j=1}^m \mathcal{L}_j$ be a Lindbladian consisting of $m$ summands and let $t > 0$. 
    Let $\mathcal{S}_1(t) = \prod_{j=1}^m e^{t \mathcal{L}_j}$ be the first-order Lie-Trotter formula.
    The additive Trotter error is bounded by
    \begin{align*}
        \|\mathcal{S}_1(t) - e^{t\mathcal{L}}\|_{\diamond}
        \le {}& \frac{t^2}{2}\sum_{j_1 =1}^m
        \bigg\|\bigg[\sum_{j_2=j_1+1}^{m} \mc L_{j_2},
        \mc L_{j_1}\bigg]\bigg\|_{\diamond}.
    \end{align*}
\end{theorem}

\begin{proof}
    The proof generalizes the derivation in \cite[Eq.~(117)]{childs2021theory} from Hamiltonians to Lindbladians. We first consider the $m=2$ case. Following \cite[Eq.~(117)]{childs2021theory}, the additive Trotter error can be expressed as 
    \begin{align*}
        &\|e^{t\mc L_2}e^{t\mc L_1} - e^{t(\mc L_1+\mc L_2)}\|_{\diamond} \\
        ={}& \bigg\|\int_{0}^t \d \tau_1 \int_{0}^{\tau_1} \d \tau_2\,
        e^{(t-\tau_1)\mc L} e^{\tau_1 \mc L_2} e^{-\tau_2\ad_{\mc L_2}}([\mc L_2, \mc L_1])
        e^{\tau_1 \mc L_1}\bigg\|_{\diamond}.
    \end{align*}
    Expanding the adjoint action $e^{-\tau_2\ad_{\mc L_2}}([\mc L_2, \mc L_1]) = e^{-\tau_2\mc L_2}[\mc L_2, \mc L_1]e^{\tau_2\mc L_2}$, the integrand becomes
    \begin{align*}
        e^{(t-\tau_1)\mc L} e^{(\tau_1-\tau_2)\mc L_2} [\mc L_2, \mc L_1] e^{\tau_2\mc L_2} e^{\tau_1 \mc L_1}.
    \end{align*}
    Observe that the only Lindbladian evolution with negative time in the adjoint expansion, $e^{-\tau_2\mc L_2}$, is absorbed by the preceding term $e^{\tau_1\mc L_2}$, leaving $e^{(\tau_1-\tau_2)\mc L_2}$. Since $\tau_2 \le \tau_1 \le t$, all the evolution times $t-\tau_1$, $\tau_1-\tau_2$, $\tau_2$, and $\tau_1$ are non-negative. Therefore, since $\|e^{\tau \mc L'}\|_{\diamond} = 1$ for any Lindbladian $\mc L'$ and $\tau \ge 0$, the diamond norm of the additive Trotter error can be bounded by 
    \begin{align*}
        \|e^{t\mc L_2}e^{t\mc L_1} - e^{t(\mc L_1+\mc L_2)}\|_{\diamond}
        \le {}& \int_{0}^t \d \tau_1 \int_{0}^{\tau_1} \d \tau_2
        \|[\mc L_2, \mc L_1]\|_{\diamond}\\
        ={}& \frac{t^2}{2} \|[\mc L_2, \mc L_1]\|_{\diamond}.
    \end{align*}
    For a general Lindbladian $\mc L = \sum_{j=1}^m \mc L_j$, applying the triangle inequality via a telescoping sum argument gives 
    \begin{align*}
        &\bigg\|\prod_{j=1}^m e^{t \mathcal{L}_j}
        - e^{t\sum_{j=1}^m \mc L_j}\bigg\|_{\diamond} \\
        \le {}& \sum_{j_1 =1}^m
        \begin{aligned}[t]
        \bigg\|&e^{t\sum_{j_2=j_1+1}^m \mc L_{j_2}}
        \prod_{j_2=1}^{j_1} e^{t\mc L_{j_2}}- e^{t\sum_{j_2=j_1}^m \mc L_{j_2}}
        \prod_{j_2=1}^{j_1-1} e^{t\mc L_{j_2}}
        \bigg\|_{\diamond}
        \end{aligned} \\
        \le {}& \sum_{j_1 =1}^m
        \begin{aligned}[t]
        \bigg\|&\bigg(e^{t\sum_{j_2=j_1+1}^m \mc L_{j_2}} e^{t\mc L_{j_1}}
        - e^{t\sum_{j_2=j_1}^m \mc L_{j_2}}\bigg)\prod_{j_2=1}^{j_1-1} e^{t\mc L_{j_2}}
        \bigg\|_{\diamond}
        \end{aligned} \\
        \le {}& \sum_{j_1 =1}^m
        \begin{aligned}[t]
        \bigg\|&e^{t\sum_{j_2=j_1+1}^m \mc L_{j_2}} e^{t\mc L_{j_1}}- e^{t(\sum_{j_2=j_1+1}^m \mc L_{j_2} + \mc L_{j_1})}
        \bigg\|_{\diamond}
        \end{aligned} \\
        \le {}& \frac{t^2}{2}\sum_{j_1 =1}^m
        \bigg\|\bigg[\sum_{j_2=j_1+1}^{m} \mc L_{j_2},
        \mc L_{j_1}\bigg]\bigg\|_{\diamond},  
    \end{align*}
    where the last inequality follows from the Trotter error bound derived above for the $m=2$ case. 
\end{proof}
In the subsequent sections, we focus on the second-order product formula due to its better error scaling compared to the first-order case. Nevertheless, our analytical framework generalizes naturally to the first-order formula.

\section{Truncation error of the BCH formula}
\label{append:bch-truncate}
Let $\mc L_1, \ldots, \mc L_M$ be general Lindbladians. The formal BCH expansion of the product $e^{\mc L_M}\cdots e^{\mc L_2}e^{\mc L_1}$ can be written as
\begin{align}
\label{eq:def-bch}
 	e^{\mc L_M}\cdots e^{\mc L_2}e^{\mc L_1} = \exp(\sum_{q = 1}^{\infty} \Phi_q),
\end{align}
where each $\Phi_q$ is a sum of the weighted right-nested commutators of $\mc L_j$ (see \citet{arnal2021note}):
\begin{align}
\label{eq:bch}
 	\Phi_q
    &= \sum_{\substack{p_v \ge 0,\\p_1+\cdots+p_M = q}}
    \frac{1}{p_1!\cdots p_M!} \sum_{\sigma\in S_q}
    c_{\sigma, q} \nonumber\\
    &\quad \mc R_{\sigma}\{[\underbrace{\mc L_M, \ldots, \mc L_M}_{p_M},
    \ldots, \underbrace{\mc L_1, \ldots, \mc L_1}_{p_1}]\},
\end{align}
and $c_{\sigma,q}$ is the same coefficient as in Eq.~\eq{def-c-sigma}. Define the sum of the $q$-fold right-nested commutator norms as
\begin{align}
\label{eq:def-comm}
 	\alpha_{\rm comm}^{(q)} =\sum_{v_1, \ldots, v_q = 1}^{M} \left\|[\mc L_{v_1}, \ldots, \mc L_{v_q}]\right\|_{\diamond}.
\end{align}
The norm of $\Phi_q$ can be bounded by
\begin{align*}
 	\|\Phi_q\|_{\diamond}\le \frac{1}{q^2} 	\alpha_{\rm comm}^{(q)},
\end{align*}
for any $q\ge 1$ (see, e.g.,  \cite[Proposition 5]{aftab2024multi}). Consequently, the series expansion in Eq.~\eq{def-bch} converges if there exist two constants $J,C\ge 0$ such that
\begin{align}
\label{eq:converg-bch}
 	\sup_{q\ge J}\alpha_{\rm comm}^{(q)} \le C.
\end{align} 
A trivial bound for $\alpha_{\rm comm}^{(q)}$ is
\begin{align}
\label{eq:trivial-alpha}
 	\alpha_{\rm comm}^{(q)} \le \bigg(2\sum_{v = 1}^M \|\mc L_v\|_{\diamond}\bigg)^q,
\end{align}
by $\|[A,B]\|_{\diamond}\le 2\|A\|_{\diamond}\|B\|_{\diamond}$ and the triangle inequality. The quantity $\alpha_{\rm comm}^{(q+1)}$ also appears in the complexity of $q$-th order Trotterization in Hamiltonian simulation~\cite{childs2021theory}, where each $\mc L_v$ is an Hermitian operator in the decomposition of the simulated Hamiltonian. For some physical Hamiltonians, such as $k$-local Hamiltonians on a lattice of size $N$, where each $\mc L_v$ can be decomposed as a sum of terms acting nontrivially on at most $k$ sites, $\alpha_{\rm comm}^{(q)}$ admits a sharper bound than Eq.~\eq{trivial-alpha}. This leads to the commutator scaling of the $q$-th order Trotterization. Specifically, suppose that the maximum interaction strength of $\sum_{v = 1}^M \mc L_v$ on each site is $g$. Eq.~\eq{trivial-alpha} yields $(2Ng)^q$, while the bound considering the commutator structure of $\mc L_v$ gives $Nq!(2kg)^q = \mathcal{O}(N(qkg)^q)$. The latter bound provides a better complexity scaling for local Hamiltonians where $k \ll N$ and $q$ is constant.

However, in the error analysis of advanced Trotter methods aiming for high precision \cite{aftab2024multi, watson2025exponentially}, bounds on $\alpha_{\rm comm}^{(q)}$ for a single $q$ are insufficient. These methods typically rely on the BCH expansion of the product formula $\mc P(st)$ to characterize high-order Trotter remainders. This requires the BCH series $\sum_{q=1}^{\infty}\Phi_q$ to be convergent. Since the interaction strength of the generators in $\mc P(st)$ is proportional to the step size $st$, the sum of $q$-th order nested commutators for $k$-local lattice systems is bounded by $\alpha_{\rm comm}^{(q)} = \mathcal{O}(Nq! (2kgst)^q)$. Due to this factorial growth, this bound cannot guarantee convergence of the full BCH series for any fixed $st>0$. A similar difficulty arises when the generators are $k$-local Lindbladians, as the parameter $\alpha_{\rm comm}^{(q)}$ admits an analogous factorial growth (see \cor{commu-bound-sum}).

To bypass this divergence issue of the BCH formula, \citet{mizuta2026commutator} showed that for $k$-local Hamiltonians on a lattice, the truncated BCH formula could approximate the exponential product even outside its convergence radius. This permits them to anlyze the truncated BCH formula instead of the full series, thereby avoiding the convergence requirement. Their proof relies on a subsystem Trotterization technique specific to lattices, which is difficult to extend to general Hamiltonians. To analyze the BCH truncation error in more general settings, we instead construct an equivalent continuous-time evolution and use its Magnus expansion to derive universal truncation bounds.

We first define a continuous-time differential equation constructed to match the discrete product in Eq.~\eq{def-bch}. Let
\begin{align}
\label{eq:df-eq-general}
&\frac{\d \mc Y(\tau)}{\d \tau}
= \mc L(\tau) \mc Y(\tau), \quad
\mc Y(0)= \mc I, \nonumber\\
&\mc L(\tau)
= \mc L_j
\quad \text{if } \tau \in (j-1, j] \text{ for } j=1,\ldots,M.
\end{align}
The solution to this ODE at time $\tau=M$ is exactly the discrete product in Eq.~\eq{def-bch}:
\[
\mc Y(M)  = e^{\mc L_M} \cdots e^{\mc L_1}.
\]
For the ODE in Eq.~\eq{df-eq-general}, let $\Omega_q(M)$ be the $q$-th Magnus term obtained from Eq.~\eqref{eq:magnus-formula} by setting $A(\tau)=\mc L(\tau)$ and $t=M$. We use these terms only order by order, without assuming convergence of the full Magnus series.

We now show that each term in the BCH formula ($\Phi_q$) and the Magnus expansion ($\Omega_q(M)$) is identical. The proof is based on the formal expression of these terms in Eq.~\eq{bch} and Eq.~\eq{magnus-formula} and does not require the convergence of the BCH formula and the Magnus expansion.
\begin{lemma}[Equivalence of BCH and Magnus Generators]
\label{lem:magnus-bch-general}
The $q$-th order term $\Phi_q$ of the BCH series in Eq.~\eq{def-bch} is identical to the $q$-th order Magnus expansion term $\Omega_q(M)$ in Eq.~\eq{magnus-formula} evaluated at time $M$ for the continuous system defined in Eq.~\eq{df-eq-general}.
\end{lemma}
\begin{proof}
    By definition, the $q$-th order Magnus expansion term $\Omega_q(M)$ for the continuous system in Eq.~\eq{df-eq-general} is given by the integral formula (Eq.~\eq{magnus-formula} with $t=M$)
    \begin{align}
         \Omega_q(M)
         &= \sum_{\sigma\in S_q} c_{\sigma, q}
         \int_{0}^{M} \d \tau_1 \int_{0}^{\tau_1}\d \tau_2
         \cdots \int_{0}^{\tau_{q-1}}\d \tau_q \nonumber\\
         &\quad \mc R_{\sigma}\big\{[\mc L(\tau_1), \ldots,
         \mc L(\tau_q)] \big\}. \nonumber
    \end{align}
    We first symmetrize the time-ordered integral by extending the integration domain to the hypercube $[0, M]^q$ and discretize the integral, which yields
    \begin{align*}
        \Omega_q(M)
        ={}& \frac{1}{ q!}\sum_{\sigma\in S_q} c_{\sigma, q}
        \int_{0}^{M} \d \tau_1 \cdots \int_{0}^{M}\d \tau_q \nonumber\\
        &\quad \mc R_{\sigma} \circ \mc T\big\{[\mc L(\tau_1), \ldots,
        \mc L(\tau_q)]\big \} \\
        ={}& \frac{1}{ q!}\sum_{\sigma\in S_q} c_{\sigma, q}
        \sum_{1\le v_1,\ldots, v_q \le M}
        \int_{v_1-1}^{v_1} \d \tau_1 \cdots
        \int_{v_q-1}^{v_q}\d \tau_q \nonumber\\
        &\quad \mc R_{\sigma} \circ \mc T\big\{[\mc L(\tau_1), \ldots,
        \mc L(\tau_q)] \big\}.
    \end{align*}
    Since $\mc L(\tau)$ is constant on each interval $(v_j-1, v_j]$ the integral simplifies to the sum
    \begin{align}
        \Omega_q(M)
        &= \frac{1}{ q!}\sum_{\sigma\in S_q} c_{\sigma, q}
        \sum_{1\le v_1,\ldots, v_q \le M}
        \mc R_{\sigma} \circ \mc T\big\{[\mc L_{v_1}, \ldots,
        \mc L_{v_q}] \big\}.
        \label{eq:magnus-general-midstep}
    \end{align}
    Now, we regroup the inner sum over the $M^q$ sequences $(v_1, \ldots, v_q)$ according to the number of times each generator $\mc L_j$ appears. Let $p_j \ge 0$ be the count of index $j$ in a sequence, such that $p_1 + \cdots + p_M = q$. The number of sequences $(v_1, \ldots, v_q)$ corresponding to a specific set of counts $(p_1, \ldots, p_M)$ is given by the multinomial coefficient $\binom{q}{p_1, \ldots, p_M}$.
    
    Regrouping the sum in Eq.~\eq{magnus-general-midstep} by these counts yields
    \begin{align*}
         \Omega_q(M)
         ={}& \frac{1}{ q!}\sum_{\sigma\in S_q} c_{\sigma, q}
         \sum_{\substack{p_j \ge 0,\\ p_1+\cdots+p_M = q}}
         \binom{q}{p_1, \ldots, p_M} \nonumber\\
         &\quad \mc R_{\sigma}\{[\underbrace{\mc L_M, \ldots, \mc L_M}_{p_M},
         \ldots, \underbrace{\mc L_1, \ldots, \mc L_1}_{p_1}]\} \\
         ={}& \sum_{\substack{p_j \ge 0,\\ p_1+\cdots+p_M = q}}
         \frac{1}{p_1!\cdots p_M!}
         \sum_{\sigma\in S_q} c_{\sigma, q} \nonumber\\
         &\quad \mc R_{\sigma}\{[\underbrace{\mc L_M, \ldots, \mc L_M}_{p_M},
         \ldots, \underbrace{\mc L_1, \ldots, \mc L_1}_{p_1}]\},
    \end{align*}
    which matches $\Phi_q$ in Eq.~\eq{bch}.
\end{proof}

Establishing this equivalence allows us to bound the BCH truncation error by studying the corresponding Magnus expansion. Before the analysis, we first introduce the necessary concepts using the general notation from our preliminaries.
We define the  $q_0$-th order truncated generator  as
\begin{align*}
 \Omega_{(q_0)}(\tau) := \sum_{q=1}^{q_0} \Omega_{q}(\tau).
\end{align*}
The corresponding $q_0$-th order approximate evolution is
\begin{align*}
 \mc Y_{(q_0)}(\tau):=e^{ \Omega_{(q_0)}(\tau) }.
\end{align*}
Following the derivative formula for the exponential map, the approximate evolution $\mathcal{Y}_{(q_0)}(\tau)$ is itself the solution to a differential equation \begin{align*}
 \frac{\d \mc Y_{(q_0)}(\tau)}{\d \tau} = \mc{L}_{(q_0)}(\tau) \mc Y_{(q_0)}(\tau).
\end{align*} The modified generator $\mathcal{L}_{(q_0)}(\tau)$ is given by 
\begin{align}
&\mathcal{L}_{(q_0)}(\tau)\notag\\
={}& \frac{\mathrm{d}}{\mathrm{d}\tau}\left(e^{\Omega_{(q_0)}(\tau)}\right)
e^{-\Omega_{(q_0)}(\tau)} \notag \\
={}& \left( \int_0^1 e^{x\Omega_{(q_0)}(\tau)}
\dot{\Omega}_{(q_0)}(\tau)
e^{(1-x)\Omega_{(q_0)}(\tau)} \mathrm{d}x \right)
e^{-\Omega_{(q_0)}(\tau)} \notag \\
={}& \int_0^1 e^{x\Omega_{(q_0)}(\tau)}
\dot{\Omega}_{(q_0)}(\tau)
e^{-x\Omega_{(q_0)}(\tau)} \mathrm{d}x \notag \\
={}& \int_0^1 \sum_{\ell=0}^{\infty} \frac{x^\ell}{\ell!}
\mathrm{ad}_{\Omega_{(q_0)}(\tau)}^\ell
\big( \dot{\Omega}_{(q_0)}(\tau) \big) \mathrm{d}x \notag \\
={}& \sum_{\ell=0}^{\infty} \frac{1}{(\ell+1)!}
\mathrm{ad}_{\Omega_{(q_0)}(\tau)}^\ell
\big( \dot{\Omega}_{(q_0)}(\tau) \big)\label{eq:expansion-L-q},
\end{align}
where the second line follows from the derivative of the exponential map, the fourth line follows from the expansion of the adjoint representation. 
Our first step is to reduce the total truncation error to the integral of the generator error.
\begin{lemma}
\label{lem:error-reduction-general}
The error of the $q_0$-th order approximation evolution is bounded by
\begin{align*}
&\big\| e^{\Omega_{(q_0)}(M)} -\mc Y(M)\big\|_{\diamond} \\
\le {}& \bigg( \int_{0}^{M}\| \mc{L}_{(q_0)}(\tau)
-\mc{L}(\tau)\|_{\diamond}\,\d \tau\bigg)
e^{\int_{0}^{M}\| \mc{L}_{(q_0)}(\tau)
-\mc{L}(\tau)\|_{\diamond}\,\d \tau}.
\end{align*}
\end{lemma}
\begin{proof}
Let \(\Delta(\tau)=\mathcal{L}_{(q_0)}(\tau)-\mathcal{L}(\tau)\) and denote by \(\mathcal{Y}(\tau,s)\) the exact time-ordered evolution generated by \(\mathcal{L}(\tau)\). Since a time-dependent Lindbladian generates a completely positive trace-preserving propagator for \(\tau\ge s\), \(\mathcal{Y}(\tau,s)\) is a quantum channel. Hence \(\|\mathcal{Y}(\tau,s)\|_{\diamond}=1\).

By \lem{variation-constants}, the difference between the approximate and exact evolution satisfies
\[
\mathcal{Y}_{(q_0)}(\tau)-\mathcal{Y}(\tau,0)=\int_0^\tau \mathcal{Y}(\tau,s)\,\Delta(s)\,\mathcal{Y}_{(q_0)}(s)\,ds.
\]
Taking diamond norms yields
\[
\|\mathcal{Y}_{(q_0)}(\tau)\|_{\diamond}\le 1+\int_0^\tau \|\Delta(s)\|_{\diamond}\|\mathcal{Y}_{(q_0)}(s)\|_{\diamond}ds.
\]
By Gronwall's inequality~\cite{gronwall1919note},
\[
\|\mathcal{Y}_{(q_0)}(\tau)\|_{\diamond}\le \exp\Big(\int_0^\tau \|\Delta(s)\|_{\diamond}ds\Big).
\]
Applying the same identity to the difference at time \(M\) gives
\begin{align*}
    \|\mathcal{Y}_{(q_0)}(M)-\mathcal{Y}(M,0)\|_{\diamond}
    &\le \int_0^M \|\Delta(\tau)\|_{\diamond}\,\|\mathcal{Y}_{(q_0)}(\tau)\|_{\diamond}d\tau\\
    &\le \exp\Big(\int_0^M \|\Delta(\tau)\|_{\diamond}d\tau\Big)-1
\end{align*}
which implies the stated bound using \(e^x-1\le x e^x\) for \(x\ge 0\).
\end{proof}
\lem{error-reduction-general} reduces our main task to analyzing and bounding $\int_{0}^{M}\| \mc{L}_{(q_0)}(\tau)-\mc{L}(\tau)\|_{\diamond}$. To understand the structure of this error term, we rely on the following property of the modified generator, which are based on the grade of the commutator terms \cite{fang2025high}. For an expression $L$ involving nested time integrals of commutators of $\mc L(\tau)$ (or norms of such terms), its \emph{grade} $\mathrm{gd}(L)$ is defined as the total number of $\mc L$ operators appearing in its nested commutator structure. 

\begin{lemma}
\label{lem:grade-cond}
 For any $1\le k\le q_0$, the sum of all terms in $\mc L_{(q_0)}(\tau) -\mc L(\tau)$ with grade $k$ vanishes. 
\end{lemma} 
\lem{grade-cond} follows from combining Proposition~3 in Ref.~\cite{fang2025high} and the fact that all terms in the remainder in their Eq.~(33) have grade at least $q_0+2$.
This implies that the error integral in \lem{error-reduction-general} contains only terms with grade $q_0+1$ and higher. We use the convention that $[X]=X$. Define the sum of \emph{doubly right-nested commutator} norms as
\begin{align}
\label{eq:def-beta-bar-general}
   \alpha_{\mathrm{comm}}^{(q_1, \ldots, q_{d})}
    =
    \sum_{v_{1,1}, \ldots, v_{d, q_{d}} = 1}^M
    \begin{aligned}[t]
    \big\|&\big[ [\mc L_{v_{1,1}}, \ldots, \mc L_{v_{1,q_1}}], \ldots,\\
    &\quad [\mc L_{v_{d,1}}, \ldots, \mc L_{v_{d,q_{d}}}] \big]
    \big\|_{\diamond}.
    \end{aligned}
\end{align} We now state the lemma that bounds the integrated expansion terms in $\mc{L}_{(q_0)}(\tau)$ as doubly right-nested commutators of the generators. 

\begin{lemma}
\label{lem:doubly-nested-bound-general}
For any non-negative integer $\ell$ and any sequence of positive integers $q_1, \ldots, q_{\ell+1}$, the integrated norm of the corresponding term in the modified generator $\mc L_{(q_0)}$ is bounded by
 \begin{align*}
  &\int_{0}^{M}
  \bigg\|\bigg(\prod_{\ell'=\ell}^{1}
  \ad_{\Omega_{q_{\ell'}}(\tau)}\bigg)
  (\dot\Omega_{q_{\ell+1}}(\tau))
  \bigg\|_{\diamond}
  \,\d\tau\\
  \le {}& \bigg(\prod_{\ell'=1}^{\ell+1} \frac{1}{q_{\ell'}^2}\bigg)
  q_{\ell+1}\alpha_{\comm}^{(q_1, \ldots, q_{\ell+1})}
  \le \alpha_{\comm}^{(q_1, \ldots, q_{\ell+1})}.
 \end{align*}
\end{lemma}
\begin{proof}
Let $T_q(t) := \{(\tau_1, \ldots, \tau_q) \mid t \ge \tau_1 \ge \cdots \ge \tau_q \ge 0\}$ be the time-ordered $q$-simplex, with the usual empty-product and zero-dimensional-integral conventions when $\ell=0$ or $q=0$.
 The term to be bounded is
\begin{align*}
    \mc{X}_{\ell}(\tau) := [\Omega_{q_1}(\tau), \ldots, \Omega_{q_\ell}(\tau), \dot{\Omega}_{q_{\ell+1}}(\tau)].
\end{align*}
By substituting the integral definitions, using the multilinearity of the commutator, and applying the triangle inequality, we obtain
\begin{widetext}
\begin{align}
    \|\mc{X}_{\ell}(\tau)\|_{\diamond}
    &\le \begin{multlined}[t]\sum_{\substack{(\sigma_1,\ldots,\sigma_{\ell+1})\in\\ S_{q_1}\times\cdots\times S_{q_{\ell+1}}}} \bigg(\prod_{\ell'=1}^{\ell+1} |c_{\sigma_{\ell'}, q_{\ell'}}|\bigg) \int_{T_{q_1}(\tau)} \d\tau_{1,1}, \ldots,\d\tau_{1,q_1}\cdots\int_{T_{q_{\ell+1}-1}(\tau)}\d\tau_{\ell+1,1},\ldots,\d\tau_{\ell+1,q_{\ell+1}-1} \\
        \bigg\| \Big[ \mc R_{\sigma_1}\big\{[\mc L(\tau_{1,1}), \ldots, \mc L(\tau_{1,q_1})]\big\}, \ldots, \mc R_{\sigma_\ell}\big\{[\mc L(\tau_{\ell,1}), \ldots, \mc L(\tau_{\ell,q_\ell})]\big\}, \\
        \mc R_{\sigma_{\ell+1}}\big\{[\mc L(\tau), \mc L(\tau_{\ell+1,1}), \ldots, \mc L(\tau_{\ell+1,q_{\ell+1}-1})]\big\} \Big] \bigg\|_{\diamond}\nonumber
    \end{multlined}\\
    & \le \begin{multlined}[t]\sum_{\substack{(\sigma_1,\ldots,\sigma_{\ell+1})\in\\ S_{q_1}\times\cdots\times S_{q_{\ell+1}}}} \bigg(\prod_{\ell'=1}^{\ell+1} \frac{1}{q_{\ell'}^2} \frac{1}{q_{\ell'}!}\bigg)q_{\ell+1} \int_{[0,\tau]^{q_1}} \d\tau_{1,1}, \ldots,\d\tau_{1,q_1}\cdots\int_{[0,\tau]^{q_{\ell+1}-1}}\d\tau_{\ell+1,1},\ldots,\d\tau_{\ell+1,q_{\ell+1}-1} \\
        \bigg\| \Big[ \mc R_{\sigma_1}\circ \mc T\big\{[\mc L(\tau_{1,1}), \ldots, \mc L(\tau_{1,q_1})]\big\}, \ldots, \mc R_{\sigma_\ell}\circ \mc T\big\{[\mc L(\tau_{\ell,1}), \ldots, \mc L(\tau_{\ell,q_\ell})]\big\}, \\
        \mc R_{\sigma_{\ell+1}}\circ \mc T\big\{[\mc L(\tau), \mc L(\tau_{\ell+1,1}), \ldots, \mc L(\tau_{\ell+1,q_{\ell+1}-1})]\big\} \Big] \bigg\|_{\diamond}\nonumber
    \end{multlined}\\
    & \le \begin{multlined}[t]\sum_{\substack{(\sigma_1,\ldots,\sigma_{\ell+1})\in\\ S_{q_1}\times\cdots\times S_{q_{\ell+1}}}} \bigg(\prod_{\ell'=1}^{\ell+1} \frac{1}{q_{\ell'}^2} \frac{1}{q_{\ell'}!}\bigg)q_{\ell+1}
        \sum_{\substack{
             v_{1,1}, \ldots, v_{\ell, q_\ell} = 1 \\
             v_{\ell+1, 1}, \ldots, v_{\ell+1, q_{\ell+1}-1} = 1
            }}^{\lceil\tau\rceil} \bigg\| \Big[ \mc R_{\sigma_1}\circ \mc T\big\{[\mc L(v_{1,1}), \ldots, \mc L(v_{1,q_1})]\big\}, \ldots, \\
        \mc R_{\sigma_\ell}\circ \mc T\big\{[\mc L(v_{\ell,1}), \ldots, \mc L(v_{\ell,q_\ell})]\big\}, 
        \mc R_{\sigma_{\ell+1}}\circ \mc T\big\{[\mc L(\tau), \mc L(v_{\ell+1,1}), \ldots, \mc L(v_{\ell+1,q_{\ell+1}-1})]\big\} \Big] \bigg\|_{\diamond},\label{eq:bound-Xk-general}
    \end{multlined} 
\end{align}
\end{widetext}
where the third line follows from the same integral decomposition argument in the proof of \lem{magnus-bch-general}. For each $i\in[\ell]$, after applying time-ordering operator $\mc T$, each unordered commutator $[\mc L(v_{i,1}), \ldots, \mc L(v_{i,q_i})]$ is mapped to an ordered commutator 
\begin{align}
\label{eq:ordered-comm-1-general}
    &[\underbrace{\mc L(\lceil\tau\rceil), \ldots,
    \mc L(\lceil\tau\rceil)}_{q_{i,\lceil\tau\rceil}},
    \ldots, \underbrace{\mc L(1), \ldots, \mc L(1)}_{q_{i,1}}]\nonumber\\
    ={}&
    [\underbrace{\mc L_{\lceil\tau\rceil}, \ldots,
    \mc L_{\lceil\tau\rceil}}_{q_{i,\lceil\tau\rceil}},
    \ldots, \underbrace{\mc L_{1}, \ldots, \mc L_{1}}_{q_{i,1}}],
\end{align}
where $q_{i,j}$ is the number of $j$ in $v_{i,1}, \ldots, v_{i,q_i}$ and there are
\begin{align*}
	    \binom{q_i}{q_{i, \lceil\tau\rceil}, \ldots, q_{i,1}}
        = \frac{q_i!}{\prod_{j=1}^{\lceil\tau\rceil}q_{i,j}!}
\end{align*}
unordered commutators mapped to the same  ordered commutator in Eq.~\eq{ordered-comm-1-general}. Then, after applying $\mc R_{\sigma_i}$ to each ordered commutator, we regroup the sum according to the index sequence $(u_{i,1}, \ldots, u_{i,q_i})$ of the resulting unordered commutator $[\mc L_{u_{i,1}},\ldots, \mc L_{u_{i,q_i}}]$. Each unordered commutator $[\mc L_{u_{i,1}},\ldots, \mc L_{u_{i,q_i}}]$ is mapped from one ordered commutator $[\underbrace{\mc L_{\lceil\tau\rceil} \ldots,\mc L_{\lceil\tau\rceil}}_{q_{i,\lceil\tau\rceil}},  \ldots, \underbrace{\mc L_{1},  \ldots, \mc L_{1}}_{q_{i,1}}]$ and the number of $\sigma_i\in S_{q_i}$ mapping the ordered commutator to the same unordered commutator is
\begin{align*}
    \prod_{j=1}^{\lceil\tau\rceil}q_{i,j}!.
\end{align*}
Therefore, the combinatorial coefficient of $[\mc L_{u_{i,1}},\ldots, \mc L_{u_{i,q_i}}]$ in the final sum is 
\begin{align*}
    \frac{q_i!}{\prod_{j=1}^{\lceil\tau\rceil}q_{i,j}!}
    \prod_{j=1}^{\lceil\tau\rceil}q_{i,j}! = q_i!.
\end{align*}
For $i= \ell+1$, after applying $\mc T$, each $[\mc L(\tau), \mc L(v_{\ell+1,1}), \ldots, \mc L(v_{\ell+1,q_{\ell+1}-1})]$ is mapped to
\begin{align}
\label{eq:ordered-comm-general}
    [\mc L_{\lceil \tau \rceil},
    \underbrace{\mc L_{\lceil\tau\rceil}, \ldots,
    \mc L_{\lceil\tau\rceil}}_{q_{\ell+1,\lceil\tau\rceil}},
    \ldots,
    \underbrace{\mc L_{1}, \ldots, \mc L_{1}}_{q_{\ell+1,1}}],
\end{align}
where $q_{\ell+1,j}$ is the number of $j$ in $v_{\ell+1, 1},\ldots, v_{\ell+1, q_{\ell+1}-1}$, and there are
\begin{align*}
    \binom{q_{\ell+1}-1}{q_{\ell+1,\lceil \tau \rceil}, \ldots,
    q_{\ell+1,1}}
    = \frac{(q_{\ell+1}-1)!}
    {\prod_{j=1}^{\lceil\tau\rceil}q_{\ell+1,j}!}
\end{align*} unordered commutators mapped to the same ordered commutator in Eq.~\eq{ordered-comm-general}. After applying $\mc R_{\sigma_{\ell+1}}$, we regrouping the sum of the commutators according to the position of the fixed $\mc L_{\lceil \tau \rceil}$ after permutation, $j'$, and the indices of other $\mc L$ operators in a commutator
\begin{align*}
    [\mc L_{v_{\ell+1,1}}, \ldots,\mc L_{\lceil\tau\rceil}, \ldots, \mc L_{v_{\ell+1,q_{\ell+1}-1}}].
\end{align*}
Each such unordered commutator is mapped from one ordered commutator in Eq.~\eq{ordered-comm-general} and the number of $\sigma_{\ell+1}\in S_{q_{\ell+1}}$ satisfying $\sigma_{\ell+1}(j')=1$ and $\mc R_{\sigma_{\ell+1}}$ maps the ordered commutator to the same unordered commutator are
\begin{align*}
    \prod_{j=1}^{\lceil\tau\rceil}q_{\ell+1,j}!.
\end{align*}
Therefore, the combinatorial coefficient of $[\mc L_{v_{\ell+1,1}}, \ldots, \mc L_{v_{\ell+1,j'-1}},\mc L_{\lceil\tau\rceil},\mc L_{v_{\ell+1,j'}}, \ldots, \mc L_{v_{\ell+1,q_{\ell+1}-1}}]$ in the final sum is
\begin{align*}
    &\frac{(q_{\ell+1}-1)!}
    {\prod_{j=1}^{\lceil\tau\rceil}q_{\ell+1,j}!}
    \prod_{j=1}^{\lceil\tau\rceil}q_{\ell+1,j}!
    = (q_{\ell+1}-1)!.
\end{align*}
In conclusion, we have
\begin{widetext}
\begin{align*}
    \|\mc X_\ell(\tau)\|_{\diamond}&\le \begin{multlined}[t]
        \bigg(\prod_{\ell'=1}^{\ell+1}
        \frac{1}{q_{\ell'}^2}\frac{1}{q_{\ell'}!}q_{\ell'}!\bigg)
        \frac{q_{\ell+1}}{q_{\ell+1}}
        \sum_{j'=1}^{q_{\ell+1}} \sum_{\substack{
                 v_{1,1}, \ldots, v_{\ell, q_\ell} = 1 \\
                 v_{\ell+1, 1}, \ldots, v_{\ell+1, q_{\ell+1}-1} = 1
            }}^{\lceil\tau\rceil} \big\| \big[ 
                [\mc L_{v_{1,1}}, \ldots, \mc L_{v_{1,q_1}}], \ldots, \\ 
                    [\mc L_{v_{\ell,1}}, \ldots, \mc L_{v_{\ell,q_\ell}}], 
                    [\mc L_{v_{\ell+1,1}}, \ldots, \mc L_{v_{\ell+1,j'-1}},
                    \mc L_{\lceil\tau\rceil},\mc L_{v_{\ell+1,j'}}, \ldots,
                    \mc L_{v_{\ell+1,q_{\ell+1}-1}}]
                \big] \big\|_{\diamond}.
    \end{multlined}\\
    &=\begin{multlined}[t]
        \bigg(\prod_{\ell'=1}^{\ell+1} \frac{1}{q_{\ell'}^2}\bigg)
        \sum_{j'=1}^{q_{\ell+1}} \sum_{\substack{
                 v_{1,1}, \ldots, v_{\ell, q_\ell} = 1 \\
                 v_{\ell+1, 1}, \ldots, v_{\ell+1, q_{\ell+1}-1} = 1
            }}^{\lceil\tau\rceil} \big\| \big[ 
                [\mc L_{v_{1,1}}, \ldots, \mc L_{v_{1,q_1}}], \ldots, \\ 
                    [\mc L_{v_{\ell,1}}, \ldots, \mc L_{v_{\ell,q_\ell}}], 
                    [\mc L_{v_{\ell+1,1}}, \ldots, \mc L_{v_{\ell+1,j'-1}},
                    \mc L_{\lceil\tau\rceil},\mc L_{v_{\ell+1,j'}}, \ldots,
                    \mc L_{v_{\ell+1,q_{\ell+1}-1}}]
                \big] \big\|_{\diamond}.
    \end{multlined}
\end{align*}
Finally, we integrate both sides from $\tau=0$ to $M$.  Since the right-hand side is piecewise constant in $\tau$, the integral is discretized to sum over $[M]$, where $\lceil\tau\rceil$ corresponds to the index $v_{\ell+1, j'}$. This yields 
\begin{align*}
    \int_{0}^{M} \|\mc{X}_{\ell}(\tau)\|_{\diamond} \,\d\tau 
     \le {}& \bigg(\prod_{\ell'=1}^{\ell+1} \frac{1}{q_{\ell'}^2}\bigg)
    \sum_{j'=1}^{q_{\ell+1}}\sum_{v_{\ell+1,j'}=1}^{M} 
    \sum_{\substack{
                 v_{1,1}, \ldots, v_{\ell, q_\ell} = 1 \\
                 v_{\ell+1, k} = 1 ~(\forall k \neq j')
            }}^{v_{\ell+1,j'}}
     \big\| \big[ 
         [\mc L_{v_{1,1}}, \ldots, \mc L_{v_{1,q_1}}], \ldots, 
        [\mc L_{v_{\ell+1,1}}, \ldots, \mc L_{v_{\ell+1,q_{\ell+1}}}]
     \big] \big\|_{\diamond}, \\
    \le {}& \bigg(\prod_{\ell'=1}^{\ell+1} \frac{1}{q_{\ell'}^2}\bigg) q_{\ell+1}
    \begin{multlined}[t]
 \sum_{v_{1,1}, \ldots, v_{\ell+1, q_{\ell+1}} = 1}^M
 \big\| \big[ 
         [\mc L_{v_{1,1}}, \ldots, \mc L_{v_{1,q_1}}], \ldots, 
        [\mc L_{v_{\ell+1,1}}, \ldots, \mc L_{v_{\ell+1,q_{\ell+1}}}]
     \big] \big\|_{\diamond}.
    \end{multlined}
\end{align*}
\end{widetext}
The commutator sum on the right-hand side is $\alpha_{\comm}^{(q_1, \ldots, q_{\ell+1})}$ from Eq.~\eq{def-beta-bar-general}. Moreover,
\begin{align*}
q_{\ell+1}\prod_{\ell'=1}^{\ell+1}\frac{1}{q_{\ell'}^2}
\le \frac{1}{q_{\ell+1}}\le1,
\end{align*}
which proves both inequalities in the lemma.
\end{proof}

We now state the main result concerning the BCH truncation error bound, which serves to bypass the convergence requirements of the standard BCH series.
\begin{theorem}[BCH truncation error bound]
 \label{thm:general-bch-truncation}
 For any positive integer $q_0$, define 
\begin{align}
\label{eq:def-beta-comm-sum}
\alpha_{\comm,q_0}
&:= \sum_{d=1}^{\infty}\frac{1}{d!}
\sum_{\substack{1\le q_1, \ldots, q_{d}\le q_0\\
q_1+\cdots+q_{d} \ge q_0+1}}
\alpha_{\comm}^{(q_1, \ldots, q_{d})}.
\end{align}
 If $\alpha_{\comm,q_0} \le 1$, the truncation error of the $q_0$-th order BCH expansion is bounded by
 \begin{align*}
 \left\|\exp\bigg(\sum_{q=1}^{q_0} \Phi_q\bigg)
 - \prod_{v = 1}^Me^{\mc L_v} \right\|_{\diamond}
 \le e\alpha_{\comm,q_0}.
 \end{align*}
\end{theorem}
\begin{proof}
    By \lem{magnus-bch-general}, the problem is equivalent to bounding the Magnus expansion truncation error $\|\exp(\Omega_{(q_0)}(M)) - \mc Y(M)\|_{\diamond}$ for the continuous system in Eq.~\eq{df-eq-general}, where $\Omega_{(q_0)}(M)$ is the $q_0$-th order Magnus generator. 

    By \lem{error-reduction-general}, it suffices to bound $\int_{0}^{M}\| \mc L_{(q_0)}(\tau)-\mc L(\tau)\|_{\diamond}\,\d \tau$. By \lem{grade-cond}, we only need to consider terms with a grade greater than $q_0$. 
    For any $q\ge q_0+1$, as $\mc L(\tau)$ is grade-one, the grade-$q$ terms of $\mc L_{(q_0)}(\tau)-\mc L(\tau)$ match the grade-$q$ component of $\mc L_{(q_0)}(\tau)$, namely
    \begin{align*}
        &\sum_{\ell=0}^{\infty}\frac{1}{(\ell+1)!}
        \sum_{\substack{1\le q_1, \ldots, q_{\ell+1}\le q_0\\
        q_1+\cdots+q_{\ell+1}=q}}
        \bigg(\prod_{\ell'=\ell}^1\ad_{\Omega_{q_{\ell'}}(\tau)}\bigg)
        (\dot{\Omega}_{q_{\ell+1}}(\tau)),
    \end{align*} 
    by expanding each $\Omega_{(q_0)}(\tau) = \sum_{q = 1}^{q_0}\Omega_q(\tau)$ in Eq.~\eq{expansion-L-q}. Applying \lem{doubly-nested-bound-general} to the grade-$q$ terms of $\int_{0}^M \|\mc L_{(q_0)}(\tau)-\mc L(\tau)\|_{\diamond}$ and substituting $d := \ell+1$ yields 
    \begin{align*}
     \sum_{d=1}^{\infty}\frac{1}{d!}
     \sum_{\substack{1\le q_1, \ldots, q_{d}\le q_0\\
     q_1+\cdots+q_{d}=q}}
     \alpha_{\comm}^{(q_1, \ldots, q_{d})}.
    \end{align*}
    Therefore, the difference integral can be bounded by
    \begin{align*}
     &\int_{0}^{M}\|\mc L_{(q_0)}(\tau)-\mc L(\tau)\|_{\diamond}\,\d \tau\\
     \le{}& \sum_{q = q_0+1}^{\infty} \sum_{d=1}^{\infty}\frac{1}{d!}
     \sum_{\substack{1\le q_1, \ldots, q_{d}\le q_0\\q_1+\cdots+q_{d}=q}}
     \alpha_{\comm}^{(q_1, \ldots, q_{d})}
     = \alpha_{\mathrm{comm}, q_0}.
    \end{align*}
    By \lem{error-reduction-general}, the error between the two exponentials can be bounded by
    \begin{align*}
    &\|\exp(\Omega_{(q_0)}(M))-\mc Y(M)\|_{\diamond}\\
    \le {} & \bigg( \int_{0}^{M}\| \mc L_{(q_0)}(\tau)
    -\mc L(\tau)\|_{\diamond}\,\d \tau\bigg)e^{\int_{0}^{M}\| \mc L_{(q_0)}(\tau)
    -\mc L(\tau)\|_{\diamond}\,\d \tau}\\
    \le {} & e^{\alpha_{\mathrm{comm}, q_0}}\alpha_{\mathrm{comm}, q_0} \le e\alpha_{\mathrm{comm}, q_0},
    \end{align*}
    which concludes the proof.
\end{proof}
As a direct application of \thm{general-bch-truncation}, we now derive the truncation error bound specifically for the second-order product formula $\mc S(t)$.
\begin{theorem}
\label{thm:bch-pf}
    Let $\mc L=\sum_{j=1}^m \mc L_j$ be a decomposition into Lindbladian summands and let $t>0$. Consider the second-order product formula $\mc S(t) = \prod_{j=m}^1 e^{t \mc L_j/2} \prod_{j=1}^m e^{t \mc L_j /2}$, and let $\{\Phi_q\}_{q\ge1}$ be the coefficients of its formal BCH series. Let
    \begin{align*}
    \alpha_{\comm}^{(q_1, \ldots, q_{d})}(t)
    ={}& \sum_{j_{1,1}, \ldots, j_{d, q_{d}} = 1}^m
    \big\|\big[ [\mc L_{j_{1,1}}, \ldots, \mc L_{j_{1,q_1}}], \ldots,\\
    &\hspace{5.1em} [\mc L_{j_{d,1}}, \ldots, \mc L_{j_{d,q_{d}}}] \big]
    \big\|_{\diamond}\,t^{q_1+\cdots + q_{d}}
    \end{align*} and \begin{align}
    \label{eq:beta-pf}
        \alpha_{\comm,q_0}(t)
        ={}& \sum_{d=1}^{\infty}\frac{1}{d!}
        \sum_{\substack{1\le q_1, \ldots, q_{d}\le q_0\\
        q_1+\cdots+q_{d} \ge q_0+1}}
        \alpha_{\comm}^{(q_1, \ldots, q_{d})}(t).
    \end{align}
    If $\alpha_{\comm,q_0}(t)\le 1$, then
    \begin{align*}
          \left\|\exp\bigg(\sum_{q=1}^{q_0} \Phi_q t^q\bigg)
          - \mc{S}(t)\right\|_{\diamond}
          \le e\alpha_{\comm,q_0}(t).
    \end{align*}
\end{theorem}
\begin{proof}
    We rewrite the product formula as \begin{align*}
        \prod_{j=m}^1 e^{t \mc L_j/2}
        \prod_{j=1}^m e^{t \mc L_j /2}
        =: \prod_{v=1}^{2m} e^{\tilde{\mc L}_v},
    \end{align*}
    where \begin{align*}
        \tilde{\mc L}_v = \begin{cases}
            t\mc L_v/2 & \text{if } v \le m, \\
            t\mc L_{2m+1-v}/2 &\text{if } v \ge m+1.
        \end{cases}
    \end{align*}
    The nested commutator bound $\alpha_{\comm}^{(q_1, \ldots, q_{d})}$ defined in Eq.~\eq{def-beta-bar-general} evaluated for $\{\tilde{\mc L}_v\}_{v=1}^{2m}$ is \begin{align}
    &\sum_{v_{1,1}, \ldots, v_{d, q_{d}} = 1}^{2m}
    \begin{aligned}[t]
    \big\|&\big[ [\tilde{\mc L}_{v_{1,1}}, \ldots, \tilde{\mc L}_{v_{1,q_1}}], \ldots,\\
    &\quad [\tilde{\mc L}_{v_{d,1}}, \ldots, \tilde{\mc L}_{v_{d,q_{d}}}] \big]
    \big\|_{\diamond}
    \end{aligned} \\
     ={}& \sum_{j_{1,1}, \ldots, j_{d, q_{d}} = 1}^m
     \begin{aligned}[t]
     \big\|&\big[ [\mc L_{j_{1,1}}, \ldots, \mc L_{j_{1,q_1}}], \ldots,\\
     &\quad [\mc L_{j_{d,1}}, \ldots, \mc L_{j_{d,q_{d}}}] \big]
     \big\|_{\diamond} t^{q_1+\cdots + q_{d}}
     \end{aligned}\\
     ={}&
     \alpha_{\comm}^{(q_1, \ldots, q_{d})}(t).
    \end{align}
    The first equality holds because the sequence $\{\tilde{\mc L}_v\}_{v=1}^{2m}$ contains exactly two copies of each operator $\mc L_j$. Consequently, the sum over $v$ indices includes every tuple of $j$ indices with multiplicity $2^{q_1+\cdots+q_{d}}$, which cancels the factor $(1/2)^{q_1+\cdots+q_{d}}$ arising from the step size $t/2$. Then, the bound $\alpha_{\comm,q_0}$ in Eq.~\eq{def-beta-comm-sum} defined for $\{\tilde{\mc L}_v\}_{v=1}^{2m}$ is reduced to $\alpha_{\comm,q_0}(t)$. The result then follows from \thm{general-bch-truncation}.
\end{proof}


\section{Lindbladian simulation with extrapolation}
\label{append:extrapolation}
Let \begin{align*}
    \alpha_{\rm comm}^{(q)}=\sum_{j_1, \ldots, j_q=1}^m \|[\mc L_{j_1}, \mc L_{j_2}, \ldots, \mc L_{j_q}]\|_{\diamond}.
\end{align*} 
By \cite[Proposition~5]{aftab2024multi}, the formal BCH expansion of the symmetric product formula $\mc S(1) = \prod_{j=m}^1 e^{ \mc L_j/2} \prod_{j=1}^m e^{\mc L_j /2}$ consists exclusively of odd-order terms: \begin{align*}
    \prod_{j=m}^1 e^{ \mc L_j/2} \prod_{j=1}^m e^{\mc L_j /2}=\exp\bigg(\mc L+\sum_{q=3, 5,\ldots}^{\infty} \Phi_q\bigg),
\end{align*} where \begin{align*}
    \|\Phi_q\|_{\diamond}\le \frac{1}{q^2}\alpha_{\rm comm}^{(q)}\le\alpha_{\rm comm}^{(q)} .
\end{align*}
Since each $\Phi_q$ in the BCH series is a sum of $q$-fold nested commutators of $\mc L_j$, the formal BCH formula of $\mc S(t)$ is \begin{align*}
    \mc S(t) = \prod_{j=m}^1 e^{t \mc L_j/2} \prod_{j=1}^m e^{t\mc L_j /2}=\exp\bigg(t\mc L + \sum_{q=3, 5,\ldots}^{\infty}\Phi_q t^q\bigg).
\end{align*} 
For any $s\in (0,1)$, we define \begin{align}
\label{eq:trunc-eff-gen}
    \mc G_{(q_0)}(s) = \mc L+\frac{1}{st}\sum_{q=3, 5,\ldots}^{q_0} \Phi_q(st)^q = \mc L+ \sum_{q=2, 4,\ldots}^{q_0-1} \Phi_{q+1}(st)^{q}
\end{align}
to be the truncated effective generator of $\mc S(st)$.
\begin{lemma}
\label{lem:bound-E}
    Let \begin{align}
    \label{eq:def-mu-comm}
    \mu_{\mathrm{comm},q_0}:=\max_{q=3, 5,\ldots, q_0} (\alpha_{\rm comm}^{(q)})^{1/q}.
    \end{align}
    The evolution $e^{t\mc G_{(q_0)}(s)}$ admits a series expansion in $s$ given by \begin{align*}
        e^{t\G_{(q_0)}(s) }=e^{t\mc L }+ \sum_{q=2,4,\ldots}^{\infty}{\E}_{(q_0),q} s^q,
    \end{align*}
    where $\mc E_{(q_0),q}$ are superoperators bounded by
    \begin{align*}
       \|\mc E_{(q_0),q}\|_{\diamond}
       \le (1+2\mu_{\mathrm{comm},q_0}t)^{3q/2}.
    \end{align*}
\end{lemma}
\begin{proof}
    Applying \lem{interaction-pic} with $\mc A(t) = \mc L$ and $\mc B(t) = \sum_{q=3, 5,\ldots}^{q_0} \Phi_q (st)^{q-1}$ yields
    \begin{widetext}
    \begin{align*}
    e^{t\G_{(q_0)}(s) }-e^{t\mc L}
    &= e^{t\mc L}\exp_{\mathcal T}\bigg(\int_{0}^t \d \tau_1 e^{-\tau_1 \mc L} \sum_{q=2, 4,\ldots}^{q_0-1} \Phi_{q+1} (st)^{q} e^{\tau_1 \mc L}\bigg)-e^{t\mc L} \\
    &= e^{t\mc L}  \sum_{j=1}^{\infty} \int_{0}^t \d \tau_1 \int_{0}^{\tau_1}\d \tau_2 \cdots \int_{0}^{\tau_{j-1}}\d \tau_j \prod_{j'=j}^{1}\bigg(e^{- \tau_{j'}\mc L}\sum_{q=2, 4,\ldots}^{q_0-1}\Phi_{q+1}s^qt^qe^{\tau_{j'}\mc L}\bigg) \\
    &= \sum_{q=2,4,\ldots}^{\infty}s^{q} \underbrace{\sum_{j=1}^{\infty} \sum_{\substack{q_1+\dots+q_j = q \\ q_{j'} \in \{2, 4, \dots, q_0-1\}}} \int_{0}^t \d \tau_1 \int_{0}^{\tau_1}\d \tau_2 \cdots \int_{0}^{\tau_{j-1}}\d \tau_j \prod_{j'=j}^{1}\big(e^{(\tau_{j'-1}- \tau_{j'})\mc L}\Phi_{q_{j'}+1}t^{q_{j'}}\big)e^{ \tau_{j}\mc L}}_{=:{\mc E}_{(q_0),q}}\\
    \end{align*}
where the third line follows from the Dyson expansion, and we define $\tau_0=t$ in the fourth line. The coefficient $\mc E_{(q_0),q}$ can be bounded by
    \begin{align*}
    \|\mc E_{(q_0), q}\|_{\diamond} &\le \sum_{j=1}^{\infty} \sum_{\substack{q_1+\dots+q_j = q \\ q_{j'} \in \{2, 4, \dots, q_0-1\}}} \int_{0}^t \d \tau_1 \int_{0}^{\tau_1}\d \tau_2 \cdots \int_{0}^{\tau_{j-1}}\d \tau_j \prod_{j'=j}^{1}\big(\big\|e^{(\tau_{j'-1}- \tau_{j'})\mc L}\big\|_{\diamond}\big\|\Phi_{q_{j'}+1}\big\|_{\diamond}t^{q_{j'}}\big)\big\|e^{ \tau_{j}\mc L}\big\|_{\diamond} \\
    &\le \sum_{j=1}^{\infty} \sum_{\substack{q_1+\dots+q_j = q \\ q_{j'} \in \{2, 4, \dots, q_0-1\}}} \int_{0}^t \d \tau_1 \int_{0}^{\tau_1}\d \tau_2 \cdots \int_{0}^{\tau_{j-1}}\d \tau_j \prod_{j'=j}^{1}\big(\big\|\Phi_{q_{j'}+1}\big\|_{\diamond}t^{q_{j'}}\big) \\
    &\le \sum_{j=1}^{\infty} \sum_{\substack{q_1+\dots+q_j = q \\ q_{j'} \in \{2, 4, \dots, q_0-1\}}}  \frac{t^j}{j!} \mu_{\mathrm{comm}, q_0}^{q+j}t^{q}\le \sum_{j=1}^{\lfloor q/2\rfloor} \frac{1}{2^j j!}(2\mu_{\mathrm{comm}, q_0}t)^{q+j}.
    \end{align*} 
    \end{widetext}
where the second line follows from $\tau_{j'-1} \ge \tau_{j'}$ and hence $\|e^{(\tau_{j'-1} -\tau_{j'})\mc L} \|_{\diamond}=1$, the fourth line follows from \begin{align*}
    \sum_{\substack{q_1+\dots+q_j = q \\ q_{j'} \in \{2, 4, \dots, q_0-1\}}} 1 \le \sum_{\substack{q_1+\dots+q_j = q \\ q_{j'} \ge 1}} 1 =\binom{q-1}{j-1}\le 2^{q},
\end{align*}
and $q=q_1+\cdots +q_j \ge 2j$.
If $2\mu_{\mathrm{comm}, q_0} t\ge 1$, we have \begin{align*}
    \|\mc E_{(q_0), q}\|_{\diamond}
    \le{}& \sum_{j=1}^{\lfloor q/2\rfloor}
    \frac{1}{2^j j!}(2\mu_{\mathrm{comm}, q_0}t)^{q+q/2}\\
    \le{}& (2\mu_{\mathrm{comm}, q_0}t)^{3q/2}
    \sum_{j=1}^\infty\frac{1}{2^j j!}
    \le (2\mu_{\mathrm{comm}, q_0}t)^{3q/2}.
\end{align*}
If $2\mu_{\mathrm{comm}, q_0} t< 1$, we have \begin{align*}
    \|\mc E_{(q_0), q}\|_{\diamond}
    \le{}& (2\mu_{\mathrm{comm}, q_0}t)^{q}
    \sum_{j=1}^{\infty} \frac{1}{2^j j!}
    \le (2\mu_{\mathrm{comm}, q_0}t)^{q}.
\end{align*}
Both cases imply the bound stated in \lem{bound-E}.
\end{proof}
The same proof applies without truncating the BCH series. If $\sup_{q=3,5,\ldots}(\alpha_{\rm comm}^{(q)})^{1/q}<\infty$, the full BCH series converges absolutely for sufficiently small $s$, and $\mc S(st)^{1/s}$ has an even expansion in $s$. At each fixed total grade, only finitely many compositions occur, so the bounds in \lem{bound-E} hold with $\mu_{\mathrm{comm},q_0}$ replaced by this supremum.

We define the function
\begin{align*}
f(s) := \tr[O\mc S(st)^{1/s}\rho_0] ,\quad s\in(0,1),
\end{align*}
which approximates the exact evolution in the limit $s\to 0$. Our goal is to estimate the value
\begin{align*}
f(0):=\lim_{s\to 0} f(s) = \tr[Oe^{t\mc L}\rho_0],
\end{align*}
via extrapolation.

\begin{algorithm}[t]
    \caption{Lindbladian simulation with step-size extrapolation}
    \label{algo:simu-extrapolation}
    \KwInput{Lindbladian $\mc L = \sum_{j=1}^m \mc L_j$, simulation time $t$, observable $O$, initial state $\rho_0$, accuracy $\eps$, BCH truncation order $q_0$.}
    \KwOutput{$\tr[Oe^{t\mc L} \rho_0]\pm \eps \|O\|$.}
    Set nested-commutator parameters $\alpha_{\mathrm{comm}, q_0}(t)$ and  $\mu_{\mathrm{comm}, q_0}$ as per Eq.~\eq{beta-pf} and Eq.~\eq{def-mu-comm}, respectively \;
    Find the smallest integer $p\ge2$ satisfying $e^{-2p} \le \eps/(8C \log p)$, where $C\ge1$ is chosen as in \lem{richardson-even}\;
    Compute $r_1, \ldots, r_p$ according to \lem{richardson-even}\;
    Find the largest $s_0=1/n$, $n\in\mathbb Z_+$, such that, with $s_p=s_0/r_p$,
    $s_p \le e^{-1}(1+2\mu_{\mathrm{comm}, q_0}t)^{-3/2}$ and
    $\alpha_{\mathrm{comm}, q_0}(s_p t)\le s_p\eps/(4e^2C \log p)$\label{lin:s-cond}\;
    Set $s_1\gets s_0/r_1, \ldots, s_p\gets s_0/r_p$, and compute $b_1, \ldots, b_p$ according to \lem{richardson-even}\;
    $S\gets \Big\lceil8\frac{\|\bm{b}\|_1^2}{\eps^2}\log(6p)\Big\rceil$\;
    \For{$i=1, \ldots, p$}{
    \For{$j=1, \ldots, S$}{
    Apply the product formula $\mc S(s_i t)$ sequentially $1/s_i$ times to $\rho_0$ and measure the observable $O$, denoting the outcome by $\mu_{i,j}$;}
    $\mu_i\gets\frac{\sum_{j=1}^S \mu_{i,j}}{S}$\;}
    \Return{$\sum_{i=1}^p b_i \mu_i$}\;
\end{algorithm}

\begin{theorem}
\label{thm:general-extra-algo}
    There exists an algorithm (\algo{simu-extrapolation}) that, given any simulation time $t>0$, precision $\eps\in(0,1)$, observable $O$, and initial state $\rho_0$, estimates $\tr[O e^{t\mc L}\rho_0]$ to within additive error $\eps\|O\|$ with probability at least $2/3$. Let $s_p$ be the step size satisfying the nested commutator condition in \lin{s-cond}. The algorithm uses a total of \begin{align*}
        \mathcal{O}\bigg( \frac{\log(1/\eps)(\log\log (1/\eps))^4}{\eps^2s_p}\bigg)  = \widetilde{\mathcal{O}}\Big(\frac{1}{\eps^2s_p}\Big)
    \end{align*} Trotter steps. This involves $\widetilde{\mathcal{O}}(1/\eps^2)$ independent circuit runs, each comprising at most 
    \begin{align*}
        \mathcal{O}\bigg(\frac{\log(1/\eps)}{s_p}\bigg) = \widetilde{\mathcal{O}}\Big( \frac{1}{s_p} \Big)
    \end{align*} Trotter steps.
\end{theorem}
\begin{proof}
    By \lem{bound-E}, function $f(s)$ admits the expansion
    \begin{widetext}
    \begin{align}
    \label{eq:expansion-f}
        f(s) = \underbrace{\tr[Oe^{t\mc L}\rho_0]+ \sum_{q=2,4,\ldots}^{2p-2}\tr[O\mc E_{(q_0),q}\rho_0] s^q}_{Q_{2p}(s)}+ \underbrace{\sum_{q=2p, 2p+2, \ldots}^{\infty}\tr[O\mc E_{(q_0),q}\rho_0] s^q+\tr[O\big(\mc S(st)^{1/s}-e^{t\mc G_{(q_0)}(s)}\big)\rho_0]}_{R_{2p}(s)}.
    \end{align}
    Since $1/s_0\in \mb{Z}_{+}$ and $r_j\in \mb{Z}_{+}$ by \lem{richardson-even}, $1/s_j = r_j/s_0 \in \mb{Z}_{+}$ for all $j\in [p]$. Also, for any $j\in [p]$, the inverse step size  \begin{align*}
        \frac{1}{s_j}= \frac{r_j}{r_ps_p} = \frac{\Big\lceil\frac{\sqrt{8} p}{\pi \sin (\pi(2 j-1) / 8 p)}\Big\rceil}{\Big\lceil\frac{\sqrt{8} p}{\pi \sin (\pi(2 p-1) / 8 p)}\Big\rceil}\frac{1}{s_p}
    \end{align*} satisfying $s_j \le s_p$ and $1/s_j = \Theta(p/(js_p))$.
    The defining series for $\alpha_{\mathrm{comm}, q_0}(st)$ has non-negative coefficients and lowest degree $q_0+1$. Since it is finite at $s=s_p$ by \lin{s-cond}, $\alpha_{\mathrm{comm}, q_0}(st)/s$ is monotonically increasing on $[0,s_p]$. Since $s_j \le s_p$ for all $j\in [p]$, we have \begin{align}
        \label{eq:cond-s-j}
        s_j \le e^{-1}(1+2\mu_{\mathrm{comm}, q_0}t)^{-3/2}, \quad e\alpha_{\mathrm{comm}, q_0}(s_jt)\le s_j\frac{\eps}{4eC \log p},
    \end{align}
    for all $j\in[p]$.  
    For any $s_j$, the remainder $R_{2p}(s_j)$ satisfies
    \begin{equation}
    \label{eq:bound-R-2p}
        \begin{aligned}
        |R_{2p}(s_j)| & \le \sum_{q=2p, 2p+2, \ldots}^{\infty}\|O\|\|\mc E_{(q_0), q} \|_{\diamond} s_j^q + \|O\|\big\|\mc S(s_jt)^{1/s_j}-e^{t\mc G_{(q_0)}(s_j)}\big\|_{\diamond} \\
        &\le \sum_{q=2p, 2p+2, \ldots}^{\infty}\|O\|((1+2\mu_{\mathrm{comm}, q_0}t)^{3/2} s_j)^q + \frac{1}{s_j}\|O\|\max\big\{1, \big\|e^{s_jt \mc G_{(q_0)}(s_j)}\big\|_{\diamond}^{1/s_j}\big\}\big\|\mc S(s_jt)-e^{s_jt \mc G_{(q_0)}(s_j)}\big\|_{\diamond} \\
        &\le \|O\|e^{-2p}\sum_{q = 0, 2, \ldots}^{\infty} e^{-q} + \frac{1}{s_j}(1+e\alpha_{\mathrm{comm}, q_0}(s_jt))^{1/s_j}e\alpha_{\mathrm{comm}, q_0}(s_jt)\|O\|\\
        &\le 2\|O\|e^{-2p}+(1+s_j)^{1/s_j}\frac{\eps}{4eC\log p}\|O\| \\
        &\le \frac{\eps}{2C\log p} \|O\|,
    \end{aligned}
    \end{equation}
    \end{widetext}
	where the first line follows from $|\tr[O\mc N(\rho_0)]|\le \|O\|\|\mc N(\rho_0)\|_1\le \|O\|\|\mc N\|_{\diamond}$ for any superoperator $\mc N$, the second line follows from \lem{bound-E} and \lem{power-bound}, the third line follows from Eq.~\eq{cond-s-j} and \thm{bch-pf}, and the fourth line follows from Eq.~\eq{cond-s-j}. The Richardson coefficients in \lem{richardson-even} satisfy $\sum_{j=1}^p b_js_j^{2\ell}=\delta_{\ell,0}$ for $\ell=0,\ldots,p-1$. Since $Q_{2p}$ is an even polynomial of degree at most $2p-2$, it follows that
    \begin{align*}
        \sum_{j=1}^p b_jQ_{2p}(s_j)=Q_{2p}(0)=f(0).
    \end{align*}
    Therefore,
    \begin{align*}
        \bigg|\sum_{j=1}^p b_j f(s_j)-f(0)\bigg|
        &=\bigg|\sum_{j=1}^p b_jR_{2p}(s_j)\bigg|\\
        &\le \|\mathbf{b}\|_1 \max_{j\in[p]}|R_{2p}(s_j)| \\
        &\le C\log p\frac{\eps}{2C\log p}\|O\|\\
        &\le \frac{\eps }{2}\|O\|.
    \end{align*}
    Since the measurement outcome $\mu_{i,j}$ is bounded by $\|O\|$ and $\mathbb{E}[\mu_{i,j}] = \tr[O \mc S(s_it)^{1/s_i} \rho_0] = f(s_i)$, Hoeffding's inequality gives \begin{align*}
        &\Pr\bigg[\bigg|\frac{1}{S}\sum_{j=1}^S \mu_{i,j}-f(s_i)\Bigg|
        \ge \frac{\eps\|O\|}{2\|\bm{b}\|_1}\bigg]\\
        &\qquad \le 2\exp\bigg(-\frac{S\eps^2}{8\|\bm{b}\|_1^2}\bigg)
        \le \frac{1}{3p}.
    \end{align*}
    By the union bound, with probability at least $2/3$, we have $|\mu_i-f(s_i)|\le \eps \|O\|/(2\|\bm{b}\|_1)$ for all $i\in [p]$ and hence \begin{align*}
        &\bigg|\sum_{j=1}^p b_j \mu_j-f(0)\bigg|\\
        \le{}& \bigg|\sum_{j=1}^p b_j (\mu_j-f(s_j))\bigg|
        + \bigg|\sum_{j=1}^p b_jf(s_j) -f(0)\bigg|\\
        \le{}& \frac{1}{2}\|O\|\eps +  \frac{1}{2}\|O\|\eps
        = \|O\|\eps.
    \end{align*}
    The algorithm uses a total of \begin{align*}
        S \sum_{j=1}^p \frac{1}{s_j}
        ={}& \mathcal O\bigg(\frac{\|\bm{b}\|_1^2}{\eps^2}\log p
        \sum_{j=1}^p \frac{p}{j s_p}\bigg)\\
        ={}& \mathcal O\bigg(\frac{p\log^4 p}{\eps^2 s_p}\bigg)
        = \mathcal O\bigg(\frac{\log(1/\eps)(\log\log(1/\eps))^4}
        {\eps^2 s_p}\bigg)
    \end{align*}
    Trotter steps. The first equality follows from $\|\bm{b}\|_1\le C\log p$ by \lem{richardson-even} and $\sum_{j=1}^p 1/j = \Theta(\log p)$. The maximum number of Trotter steps per coherent run is \begin{align*}
       \max_{j\in [p]} \frac{1}{s_j} = \frac{1}{s_1} = \Theta\Big(\frac{p}{s_p}\Big) = \Theta\bigg(\frac{\log(1/\eps)}{s_p}\bigg).
    \end{align*}
\end{proof}
\subsection{Application to \texorpdfstring{$(\Gamma,k)$}{(Gamma,k)}-local Lindbladians}
\label{append:k-local-example}

Consider the $(\Gamma,k)$-local Lindbladian on a lattice $\Lambda = [N]$ as
\begin{align}
    \mc L(\rho) = -\mathrm{i} \sum_{\mu=1}^{m_C} [ H_\mu, \rho] + \sum_{\nu=1}^{m_D} \Big(L_\nu \rho L_\nu ^{\dagger}-\frac{1}{2}\{ L_\nu^{\dagger} L_\nu , \rho\}\Big),
\end{align}
where each $H_\mu$ is supported on at most $k$ sites and each $L_\nu$ is a sum of at most $\Gamma$ operators supported on at most $k$ sites,
\begin{align}
    L_{\nu} &= \sum_{\gamma=1}^{\Gamma } d_{\nu,\gamma}, \quad |\supp(d_{\nu,\gamma})| \le k.
\end{align}
Based on this structure, we identify two levels of decomposition relevant to our analysis.

\myparagraph{Trotter decomposition.}
For the construction of product formulas, we decompose $\mc L$ into Lindbladians generated by a single Hamiltonian or jump operator. We write
\begin{align}
\label{eq:trotter-decompose}
    \mc L = \sum_{j=1}^{m} \mc L_j,
\end{align}
where $m = m_C + m_D$. Each $\mathcal{L}_j$ represents the Lindbladian corresponding to either a Hamiltonian term $-\mathrm{i}[H_\mu, \cdot]$ or a dissipator generated by $L_\nu$.

\myparagraph{Local superoperator decomposition.}
For the commutator analysis, we use the elementary decomposition in Eq.~\eq{local-decompose}. Each $\mc K_v$ is supported on at most $2k$ sites, and the family $\{\mc K_v\}_{v=1}^M$ satisfies the $g$-extensiveness condition in Eq.~\eq{def-extensive-main}. The resulting total-norm estimate in Eq.~\eq{1-norm-bound} can be overly pessimistic for physically relevant systems because it does not exploit locality. We instead use commutator bounds that incorporate the local structure.

In the context of Hamiltonian simulation, \citet{mizuta2026commutator} provided an upper bound on the sum of nested commutators of local Hermitian operators. Their proof uses only the commutativity structure determined by the supports and the submultiplicative commutator bound. The same argument therefore applies to local superoperators. We state the resulting form below.

\begin{lemma}[Superoperator form of {\cite[Lemma~7]{mizuta2026commutator}}]
\label{lem:commu-bound}
     Let $\mc K_1, \ldots, \mc K_{M}$ and $\mc A$ be superoperators on an $N$-site lattice, each supported on at most $2k$ sites. Suppose that $\{\mc K_v\}_{v=1}^M$ satisfy the $g$-extensiveness condition in Eq.~\eq{def-extensive-main}. For any non-negative integer $q$ and any $q'\in\{0,\ldots,q\}$, we have
    \begin{align*}
        \sum_{v_1, \ldots, v_q=1}^{M}
        &\big\|[\mc K_{v_1}, \ldots, \mc K_{v_{q'}}, \mc A,\\
        &\quad \mc K_{v_{q'+1}}, \ldots, \mc K_{v_q}]
        \big\|_{\diamond}\le q!(4kg)^q\|\mc A\|_{\diamond},
    \end{align*}
    where empty index sequences are omitted; in particular, the case $q=0$ is the identity bound $\|\mc A\|_{\diamond}\le\|\mc A\|_{\diamond}$.
\end{lemma}

A useful corollary can be obtained by summing the inequality \lem{commu-bound} for $\mc A=\mc K_v$ over all $v\in[M]$.

\begin{corollary}
\label{cor:commu-bound-sum}
    Under the assumptions of \lem{commu-bound}, for any positive integer $q$, the sum of the norms of the nested commutator involving $q$ local superoperators is bounded by
    \begin{align*}
        \sum_{v_1, \ldots, v_q=1}^{M} \big\|[\mc K_{v_1}, \ldots, \mc K_{v_q}]\big\|_{\diamond}  \le \frac{1}{4kq}q!(4kg)^{q}N,
    \end{align*}
    where $N$ is the number of sites in the lattice.
\end{corollary}

\begin{proof}
We apply \lem{commu-bound} with its $q$ replaced by $q-1$, $q'=q-1$, and $\mc A = \mc K_{v_q}$. This includes $q=1$ through the trivial $q=0$ case of the lemma. Summing over the index $v_q$, we obtain
\begin{align*}
    \sum_{v_1, \ldots, v_q=1}^{M} \big\|[\mc K_{v_1}, \ldots, \mc K_{v_q}]\big\|_{\diamond} 
    & \le \sum_{v_q=1}^{M} (q-1)! (4kg)^{q-1}  \|\mc K_{v_q}\|_{\diamond} \\
    &= (q-1)! (4kg)^{q-1}  \sum_{v_q=1}^{M}\|\mc K_{v_q}\|_{\diamond}\\
    &\le (q-1)!(4kg)^{q-1}  Ng\\
    &= \frac{1}{4kq}q!(4kg)^{q}N,
\end{align*}
where the second inequality follows from Eq.~\eq{1-norm-bound}.
\end{proof}

We can now bound the doubly right-nested commutators defined in Eq.~\eq{def-beta-bar-general} using \lem{commu-bound}. 
\begin{theorem}
\label{thm:doubly-nested-bound}
    Let $\mc K_1, \ldots, \mc K_{M}$ be superoperators on an $N$-site lattice, each supported on at most $2k$ sites. Suppose that $\{\mc K_v\}_{v=1}^M$ satisfy the $g$-extensiveness condition in Eq.~\eq{def-extensive-main}. For a sequence of positive integers $q_1, \ldots, q_{d}$, define the cumulative counts $P_{r} = \sum_{j=r}^{d} q_j$ for $r \in [d]$ and $P_{d+1}=1$. Then, the following bound holds:
    \begin{align}
    \label{eq:double-comm-k-local-N}
        &\sum_{v_{1,1}, \ldots, v_{d, q_{d}} = 1}^M
        \big\|\big[ [\mc K_{v_{1,1}}, \ldots, \mc K_{v_{1, q_1}}], \ldots,[\mc K_{v_{d,1}}, \ldots, \mc K_{v_{d, q_{d}}}] \big]
        \big\|_{\diamond}\nonumber\\
        &\le\frac{1}{4kq_d}
        \bigg(\prod_{r=1}^{d} P_{r+1} q_r! (4kg)^{q_r}\bigg) N.
    \end{align}
\end{theorem}

\begin{proof}
    For any layer $r \in [d]$, let $\vec{v}_r = (v_{r,1}, \ldots, v_{r, q_r})$ denote the vector of indices ranging from $1$ to $M$. We define the nested commutator for layer $r$ as
    \begin{align*}
        {\mc C}_{\vec{v}_r} := [{\mc K}_{v_{r,1}}, \ldots, {\mc K}_{v_{r, q_r}}].
    \end{align*}
    The left hand side of Eq.~\eqref{eq:double-comm-k-local-N} can be rewritten as
    \begin{align*}
        \sum_{\vec{v}_{1}, \ldots, \vec{v}_{d}} \big\| [ {\mc C}_{\vec{v}_1},\ldots, {\mc C}_{\vec{v}_d} ] \big\|_{\diamond}.
    \end{align*}
    Define the partial nested commutator accumulated from layer $d$ up to $r$ as
    \begin{align*}
        \mathcal{W}_{\vec{v}_{r}, \ldots, \vec{v}_d} := [ {\mc C}_{\vec{v}_r}, \ldots, {\mc C}_{\vec{v}_{d-1}}, {\mc C}_{\vec{v}_d} ].
    \end{align*}
    We prove the following bound by backward induction on $r$ from $d$ to $1$:
    \begin{align} \label{eq:induct-W-N}
        \sum_{\vec{v}_{r}, \ldots, \vec{v}_{d}}
        \big\|\mathcal{W}_{\vec{v}_{r}, \ldots, \vec{v}_d}\big\|_{\diamond}
        \le \frac{1}{4kq_d}
        \bigg(\prod_{j=r}^{d} P_{j+1} q_j! (4kg)^{q_j}\bigg) N.
    \end{align}
    
    Consider the base case $r=d$. The term is $\sum_{\vec{v}_d} \|{\mc C}_{\vec{v}_d}\|_{\diamond} = \sum_{v_{d,1}, \ldots, v_{d, q_d}} \| [\mc K_{v_{d,1}}, \ldots, \mc K_{v_{d, q_d}}] \|_{\diamond}$. Applying \cor{commu-bound-sum} yields the bound \begin{align*}
        \sum_{\vec{v}_d} \|{\mc C}_{\vec{v}_d}\|_{\diamond}  \le \frac{1}{4kq_d} q_d!(4kg)^{q_d} N,
    \end{align*} which implies that the base case holds since $P_{d+1}=1$.

    Assume Eq.~\eqref{eq:induct-W-N} holds for $r$. We consider the sum for $r-1$:
    \begin{align*}
        \sum_{\vec{v}_{r-1}, \ldots, \vec{v}_d} \big\|\mathcal{W}_{\vec{v}_{r-1}, \ldots, \vec{v}_d}\big\|_{\diamond} 
        &= \sum_{\vec{v}_{r}, \ldots, \vec{v}_d} \sum_{\vec{v}_{r-1}} \big\| [ {\mc C}_{\vec{v}_{r-1}}, \mathcal{W}_{\vec{v}_{r}, \ldots, \vec{v}_d} ] \big\|_{\diamond}.
    \end{align*}
    Let $X_{\mathcal{W}} = \supp(\mathcal{W}_{\vec{v}_{r}, \ldots, \vec{v}_d})$. The superoperator $\mathcal{W}_{\vec{v}_{r}, \ldots, \vec{v}_d}$ is composed of nested commutators involving a total of $P_r = \sum_{j=r}^d q_j$ local terms. Since each local superoperator has a support size of at most $2k$, the total support size is bounded by $|X_{\mathcal{W}}| \le 2k P_r$. 
    
    The inner commutator $[ {\mc C}_{\vec{v}_{r-1}}, \mathcal{W}_{\vec{v}_{r}, \ldots, \vec{v}_d} ]$ is non-zero only if at least one operator in the sequence ${\mc C}_{\vec{v}_{r-1}}$ has a support that overlaps with $X_{\mathcal{W}}$. We bound the sum over $\vec{v}_{r-1}$ by iterating over the index $j \in [q_{r-1}]$ of the first operator in the sequence that overlaps with $X_{\mathcal{W}}$. Let this operator be $\mc K_{u_{r-1, j}}$. We have
    \begin{widetext}
    \begin{align*}
        \sum_{\vec{v}_{r-1}} \big\| [ {\mc C}_{\vec{v}_{r-1}}, \mathcal{W}_{\vec{v}_{r}, \ldots, \vec{v}_d} ] \big\| _{\diamond}
        \le~& \sum_{j=1}^{q_{r-1}} \sum_{\substack{u_{r-1, j} \in [M] \\ \supp({\mc K}_{u_{r-1, j}}) \cap X_{\mathcal{W}} \neq \emptyset}} \sum_{\substack{u_{r-1, \ell} \in [M] \\ (\forall \ell \neq j)}} \big\| [ [{\mc K}_{u_{r-1,1}}, \ldots, {\mc K}_{u_{r-1, q_{r-1}}}], \mathcal{W}_{\vec{v}_{r}, \ldots, \vec{v}_d} ] \big\|_{\diamond} \\
        \le~& \sum_{j=1}^{q_{r-1}} \sum_{\substack{u_{r-1, j} \in [M] \\ \supp({\mc K}_{u_{r-1, j}}) \cap X_{\mathcal{W}} \neq \emptyset}} (q_{r-1}-1)! (4kg)^{q_{r-1}-1} \|\mc K_{u_{r-1, j}}\|_{\diamond} \cdot 2 \big\|\mathcal{W}_{\vec{v}_{r}, \ldots, \vec{v}_d}\big\|_{\diamond} \\
        \le~& \sum_{j=1}^{q_{r-1}} (2k P_r g) \cdot (q_{r-1}-1)! (4kg)^{q_{r-1}-1} \cdot 2 \big\|\mathcal{W}_{\vec{v}_{r}, \ldots, \vec{v}_d}\big\|_{\diamond} \\
        =~& q_{r-1} \cdot (2k P_r g) \cdot (q_{r-1}-1)! (4kg)^{q_{r-1}-1} \cdot 2 \big\|\mathcal{W}_{\vec{v}_{r}, \ldots, \vec{v}_d}\big\|_{\diamond} \\
        =~& P_r \cdot q_{r-1}! \cdot (4kg)^{q_{r-1}} \big\|\mathcal{W}_{\vec{v}_{r}, \ldots, \vec{v}_d}\big\|_{\diamond}.
    \end{align*}
    \end{widetext}
    The second inequality follows from $\|[X,Y]\|_{\diamond}\le 2\|X\|_{\diamond}\|Y\|_{\diamond}$ and \lem{commu-bound} with its $q$ replaced by $q_{r-1}-1$, $q'=j-1$, and $\mc A=\mc K_{u_{r-1,j}}$, after summing over all indices $u_{r-1,\ell}$ with $\ell\neq j$. The third inequality uses the $g$-extensiveness condition in Eq.~\eq{def-extensive-main}, where the sum of norms over operators overlapping with $X_{\mathcal{W}}$ is bounded by $|X_{\mathcal{W}}|g \le 2k P_r g$.
    
    Substituting this result back into the sum over $\vec{v}_r, \ldots, \vec{v}_d$ and applying the inductive hypothesis Eq.~\eqref{eq:induct-W-N} yields
	    \begin{align*}
	        &\sum_{\vec{v}_{r-1}, \ldots, \vec{v}_d} \big\|\mathcal{W}_{\vec{v}_{r-1}, \ldots, \vec{v}_d}\big\|_{\diamond} 
	        \\
            \le{}& P_r q_{r-1}! (4kg)^{q_{r-1}} \sum_{\vec{v}_{r}, \ldots, \vec{v}_d} \big\|\mathcal{W}_{\vec{v}_{r}, \ldots, \vec{v}_d}\big\|_{\diamond} \\
	        \le{}&P_r q_{r-1}! (4kg)^{q_{r-1}} \cdot \frac{1}{4kq_d}\bigg(\prod_{j=r}^{d} P_{j+1} q_j! (4kg)^{q_j}\bigg) N \\
	        ={}&\frac{1}{4kq_d}
	        \bigg(\prod_{j=r-1}^{d} P_{j+1} q_j! (4kg)^{q_j}\bigg) N.
	    \end{align*}
    This completes the induction for $r=1$.
\end{proof}

Plugging this bound into $\alpha_{\comm,q_0}(t)$ defined in \thm{bch-pf}, we can bound the truncation error for any $(\Gamma,k)$-local Lindbladian.
\begin{lemma}
\label{lem:bch-approx} 
Let $\mc L=\sum_{j=1}^m \mc L_j$ be a $(\Gamma,k)$-local Lindbladian on an $N$-site lattice $\Lambda$ satisfying the $g$-extensiveness condition.
If $8e^2q_0kgt\le 1$, the commutator bound $\alpha_{\comm,q_0}(t)$ defined in \thm{bch-pf} for $\mc L$ is bounded by \begin{align*}
    \alpha_{\comm,q_0}(t)\le N e^{-q_0}.
\end{align*}
\end{lemma}
\begin{proof} Let $\mc L= \sum_{v=1}^{M} \mc K_{v}$ be the decomposition of $\mc L$ into $2k$-local superoperators in Eq.~\eq{local-decompose}.
    For any positive integer sequence $q_1, \ldots, q_{d}$, let $
        P_{r} = \sum_{j=r}^{d} q_j,$
    for $r\le d$, $P_{d+1} = 1$, and $q = q_1+\cdots+q_{d}$. Then, we have 
	\begin{align*}
	\alpha_{\comm}^{(q_1, \ldots, q_{d})}(t)
	&= \sum_{j_{1,1}, \ldots, j_{d, q_{d}} = 1}^m
	\begin{aligned}[t]
	\big\|&\big[ [\mc L_{j_{1,1}}, \ldots, \mc L_{j_{1,q_1}}], \ldots,\\
	&\quad [\mc L_{j_{d,1}}, \ldots, \mc L_{j_{d,q_{d}}}] \big]
	\big\|_{\diamond}\,t^{q}
	\end{aligned}\\
	& \le  \sum_{v_{1,1}, \ldots, v_{d, q_{d}} = 1}^M
	\begin{aligned}[t]
	\big\|&\big[ [\mc K_{v_{1,1}}, \ldots, \mc K_{v_{1,q_1}}], \ldots,\\
	&\quad [\mc K_{v_{d,1}}, \ldots, \mc K_{v_{d,q_{d}}}] \big]
	\big\|_{\diamond}\,t^q
	\end{aligned} \\
	& \le  \bigg(\prod_{r=1}^{d}P_{r+1}q_r!(4kg)^{q_r }\bigg)Nt^q,
	\end{align*}
where the second line follows by decomposing each $\mc L_j$ into local superoperators $\mc K_v$ and applying the triangle inequality, the third line follows from \thm{doubly-nested-bound}.
For $q_1, \ldots, q_{d}\le q_0$, we have 
\begin{align*}
    \alpha_{\comm}^{(q_1, \ldots, q_{d})}(t)\le q^{d}\bigg(\prod_{r=1}^{d}q_0^{q_{r}}(4kg)^{q_r }\bigg)N t^q = q^{d} (4q_0kgt)^q N,
\end{align*}
where the inequality follows from $P_{r+1}\le q$. 
Then, writing $Q$ for the total grade, we have
\begin{align*}
    \alpha_{\comm,q_0}(t)
    &\le N\sum_{Q=q_0+1}^{\infty}(4q_0kgt)^Q
    \sum_{d=1}^{Q}\frac{Q^d}{d!}
    \sum_{\substack{1\le q_1,\ldots,q_d\le q_0\\q_1+\cdots+q_d=Q}}1\\
    &\le N\sum_{Q=q_0+1}^{\infty}(8eq_0kgt)^Q
    \le Ne^{-q_0}.
\end{align*}
For the composition count, we use
\begin{align*}
    \sum_{d=1}^{Q}
    \sum_{\substack{1\le q_1,\ldots,q_d\le q_0\\q_1+\cdots+q_d=Q}}1
    \le \sum_{d=1}^{Q}\binom{Q-1}{d-1}=2^{Q-1}.
\end{align*}
Together with $\sum_{d=1}^{Q}Q^d/d!\le e^Q$, this gives the second line above. The last inequality follows from $8e^2q_0kgt\le1$.
\end{proof}
Based on the convergence condition in Eq.~\eq{converg-bch}, the bound in Eq.~\eq{trivial-alpha} guarantees convergence of the BCH series $\sum_{q=1}^{\infty}\Phi_q$ when $\sum_{v=1}^M\|\mc K_v\|_{\diamond}$ is sufficiently small, corresponding to $g=\mathcal O(1/N)$ with a sufficiently small constant. In contrast, for fixed $t$, \lem{bch-approx} shows that the exponential of the truncated BCH formula approximates the exponential product in the larger regime $g=\mathcal{O}(1/(k\log N))$, where the full BCH series is not guaranteed to converge. By combining these bounds with \thm{general-extra-algo}, we obtain the following Trotter number bounds for simulating $(\Gamma,k)$-local Lindbladians.
\begin{theorem}
\label{thm:local-extrapolation}
    Let $\mc L=\sum_{j=1}^m \mc L_j$ be a $(\Gamma,k)$-local Lindbladian on an $N$-site lattice $\Lambda$ satisfying the $g$-extensiveness condition. There exists an algorithm that, given any simulation time $t>0$, precision $\eps\in(0,1)$, observable $O$, and initial state $\rho_0$, estimates $\tr[Oe^{t \mc L} \rho_0]$ to within additive error $\eps \|O\|$ with probability at least $2/3$. The algorithm uses a total of \begin{align*}
        &\mathcal{O}\Big(\frac{\log(1/\eps)(\log\log(1/\eps))^4}{\eps^2}
        (kgt)^{3/2 }(\sqrt{N}+ \log^{3/2}(Nkgt/\eps))\Big)
    \end{align*} Trotter steps. This involves $\widetilde{\mathcal{O}}(1/\eps^2)$ independent circuit runs, each comprising at most 
    \begin{align*}
        &\mathcal{O}\big(\log(1/\eps)(kgt)^{3/2 }
        (\sqrt{N}+ \log^{3/2}(Nkgt/\eps))\big)
    \end{align*} Trotter steps.
\end{theorem}
\begin{proof}
    We use \algo{simu-extrapolation} to estimate the expectation value. To analyze the complexity, let $\mc L=\sum_{v=1}^M\mc K_v$ be the local decomposition in Eq.~\eq{local-decompose}, and let $q_0$ be the smallest odd integer satisfying
    \begin{align}
    \label{eq:def-q-0}
        q_0\ge\max\bigg\{3,2\log\bigg(\frac{32e^4NCkgt\log p}{\eps}\bigg)\bigg\},
    \end{align}
    and set
    \begin{align}
    \label{eq:local-s-cond}
        s_{p,1}:=(8e^2q_0kgt)^{-1}.
    \end{align}
    These choices satisfy $q_0=\mathcal O(\log(Nkgt/\eps))$ and $s_{p,1}^{-1}=\mathcal O(kgt\log(Nkgt/\eps))$. Since $q_0\le e^{q_0/2}$ for $q_0\ge3$, Eq.~\eqref{eq:def-q-0} implies $Ne^{-q_0}\le s_{p,1}\eps/(4e^2C\log p)$. By \lem{bch-approx},
    \begin{align*}
        \alpha_{\mathrm{comm},q_0}(s_{p,1}t)
        \le \frac{s_{p,1}\eps}{4e^2C\log p}.
    \end{align*}
    Since $\alpha_{\mathrm{comm},q_0}(st)/s$ is increasing where the defining series is finite, every $s_p\le s_{p,1}$ satisfies
    \begin{align*}
        \alpha_{\mathrm{comm},q_0}(s_pt)
        \le \frac{s_p\eps}{4e^2C\log p}.
    \end{align*}
    
    The right-nested commutator bound $\alpha_{\rm comm}^{(q)}$ satisfies \begin{align*}
        \alpha_{\rm comm}^{(q)}
        ={}& \sum_{j_1, \ldots, j_q = 1}^m
        \|[\mc L_{j_1}, \ldots, \mc L_{j_q}]\|_{\diamond}\\
        \le{}& \sum_{v_1, \ldots, v_q=1}^M
        \|[\mc K_{v_1}, \ldots, \mc K_{v_q}]\| _{\diamond}\\
        \le{}& Ng(q-1)!(4kg)^{q-1}.
    \end{align*}
    The first inequality follows from the triangle inequality by decomposing each $\mc L_j$ into  terms of $\mc K_v$. The second inequality follows from \cor{commu-bound-sum}. Then, we have \begin{align*}
        \mu_{\mathrm{comm},q_0}
        ={}& \max_{q=3, 5,\ldots, q_0} \big(\alpha_{\mathrm{comm}}^{(q)}\big)^{1/q}\\
        \le{}& 4kg \max_{q=3, 5, \ldots, q_0} qN^{1/q}\\
        \le{}& 4kg\max\{3N^{1/3}, q_0 N^{1/q_0}\}\\
        \le{}& 4kg\max\{3N^{1/3}, e q_0 \}.
    \end{align*}
    Here $xN^{1/x}$ decreases for $x<\log N$ and increases for $x>\log N$, so its maximum over $[3,q_0]$ occurs at an endpoint. If $q_0\le\log N$, the second endpoint is bounded by the first; if $q_0\ge\log N$, then $N^{1/q_0}\le e$. Therefore, with
    \begin{align*}
        s_{p,2}:=e^{-1}\big[1+8kgt\max\{3N^{1/3},eq_0\}\big]^{-3/2},
    \end{align*}
    we have
    \begin{align*}
        s_{p,2}\le e^{-1}(1+2\mu_{\mathrm{comm},q_0}t)^{-3/2}.
    \end{align*}

    To ensure $1/s_0=1/(s_pr_p)\in\mathbb Z_+$, consider the feasible choice
    \begin{align*}
        \frac{1}{s_p}:=\left\lceil
        \frac{\max\{s_{p,1}^{-1},s_{p,2}^{-1}\}}{r_p}
        \right\rceil r_p.
    \end{align*}
    This choice satisfies $s_p\le\min\{s_{p,1},s_{p,2}\}$ and
    \begin{align*}
        \max\{s_{p,1}^{-1},s_{p,2}^{-1}\}
        \le \frac{1}{s_p}
        <\max\{s_{p,1}^{-1},s_{p,2}^{-1}\}+r_p.
    \end{align*}
    The largest admissible $s_0$ selected in \algo{simu-extrapolation} can only decrease $1/s_p$. Using $r_p=\mathcal O(\log(1/\eps))$, Eqs.~\eqref{eq:def-q-0} and \eqref{eq:local-s-cond} therefore give, in the nontrivial asymptotic regime, the bound
    \begin{align*}
        \frac{1}{s_p}
        =\mathcal{O}\big((kgt)^{3/2}
        (\sqrt{N}+\log^{3/2}(Nkgt/\eps))\big).
    \end{align*}
    The complexity bounds then follow from \thm{general-extra-algo}.
\end{proof}

\section{Numerical Experiment Details}\label{append:numerics}
All the results and plots are obtained by numerical simulations on a computing node with 8vCPUs and 64GB RAM via Python 3.11.12.  We use the `QuTiP' library~\cite{qutip5} to store and manipulate the density matrices. For an $N$-qubit superoperator $A(\cdot)B$, the matrix representation $B^T\otimes A$ is of the size $4^N\times 4^N$. 

For the task of observable estimation and Richardson extrapolation in \fig{extrapolation} and \fig{heisenberg-extrapolation}, since the errors after extrapolation are approaching the numerical errors, we directly call the matrix exponential function to calculate $e^{t\L}$ and $\mathcal{S}(t/r)^r$, which achieves high precision but is limited to a small system size $N$.

For the task of channel approximation in \fig{commutator} and \fig{heisenberg-noise}, since computing the diamond norm $\norm{e^{t\L}-\mathcal{S}(t/r)^r}_{\diamond}$ relies on semidefinite programming with the matrix representation of superoperator exponentials, the curse of dimensionality prevents us from scaling beyond $N=6$. In turn, we consider the error of the output density matrix in terms of the trace norm $\norm{(e^{t\L}-\mathcal{S}(t/r)^r)\rho_0}_1$ given an initial $\rho_0$, which is smaller than the diamond distance as worst case metric. On the other hand, this enables us to adopt the forward Euler method to incrementally update the density matrix through $2^N\times 2^N$ matrix multiplication, rather than exponentiating a $4^N\times 4^N$ superoperator.
Specifically, the numerical update methods used to plot \fig{commutator} and \fig{heisenberg-noise} and their errors are presented as follows:
\begin{itemize}
    \item The exact evolution channel $e^{t\L}$: Repeating $\Delta\rho=\L(\rho)\frac{t}{10^5}$ and $\rho\leftarrow\rho+\Delta\rho$ for $10^5$ times. By directly calculating the exponential of the matrix representation of $t\L$ on small systems, we verify that the numerical errors of the Euler method are of the order $10^{-6}$ in terms of the trace distance between the final density matrices.
    \item The Hamiltonian evolution channel $e^{-it\ad_{H_\mu}}$: For system size $N\in[4,10]$, we calculate the exact matrix exponential $\rho\leftarrow e^{-iH_\mu t}\rho e^{iH_\mu t}$ whose computational cost is much smaller than exponentiating a superoperator. 
    \item The jump operator channel $e^{tD_\nu}$: Repeating $\Delta\rho=D_\nu(\rho)\frac{t}{10^3}$ and $\rho\leftarrow\rho+\Delta\rho$ for $10^3$ times, whose numerical errors are shown to be of the order $10^{-7}$ by the similar approach above.
\end{itemize}

\nocite{blanes2009magnus, Arnal_2018, dollard1984product, arnal2021note, stinespring1955positive, dawson2006solovay, childs2012hamiltonian, gilyen2019quantum, low2019hamiltonian, qutip5, Zhao_2025}
\bibliography{ref}

\end{document}